\documentclass[12pt]{article}

\usepackage[T1]{fontenc}
\usepackage[utf8]{inputenc}
\usepackage{lmodern}
\usepackage{booktabs}
\usepackage{amsmath,amssymb,amsthm,mathtools}
\usepackage[left=3.5cm, right=3.5cm, top=2.8cm, bottom=3.2cm]{geometry}
\usepackage{graphicx}
\usepackage{float}
\usepackage{tikz}
\usetikzlibrary{cd,decorations.pathreplacing,matrix,arrows,positioning,automata,shapes,shadows,calc,fadings,decorations,snakes,through,intersections}
\usepackage{enumitem}
\usepackage{authblk}
\usepackage{comment}
\usepackage{nicefrac}
\usepackage{bm}
\usepackage{mathrsfs}
\usepackage{cancel}
\usepackage{dirtytalk}
\usepackage[authoryear]{natbib}

\usepackage{microtype}

\usepackage{hyperref}
\hypersetup{
    pdftitle={Prova},
    pdfauthor={},
    pdfmenubar=false,
    pdffitwindow=true,
    pdfstartview=FitH,
    colorlinks=true,
    linkcolor=blue,
    citecolor=blue,
    urlcolor=cyan
}
\usepackage{cleveref}

\crefname{axiom}{axiom}{axioms}
\Crefname{axiom}{Axiom}{Axioms}
\newtheorem{axiom}{}

\crefname{axiomm}{axiom}{axioms}
\Crefname{axiomm}{Axiom}{Axioms}
\newtheorem{axiomm}{}

\crefname{axiomq}{axiom}{axioms}
\Crefname{axiomq}{Axiom}{Axioms}
\newtheorem{axiomq}{}

\crefname{axiome}{axiom}{axioms}
\Crefname{axiome}{Axiom}{Axioms}

\newtheorem{theorem}{Theorem}
\newtheorem{corollary}{Corollary}
\newtheorem{lemma}{Lemma}
\newtheorem{proposition}{Proposition}

\theoremstyle{definition}
\newtheorem{definition}{Definition}
\let\olddefinition\definition
\renewcommand{\definition}{\olddefinition\normalfont}
\newtheorem{question}{Question}
\let\oldquestion\question
\renewcommand{\question}{\oldquestion\normalfont}
\newtheorem{example}{Example}
\let\oldexample\example
\renewcommand{\example}{\oldexample\normalfont}
\newtheorem{remark}{Remark}
\let\oldremark\remark
\renewcommand{\remark}{\oldremark\normalfont}
\newtheorem{claim}{\textsc{Claim}}

\newcommand{\R}{\mathbb{R}}

\newcommand{\EUp}{\mathbb{E}_p[u]}
\newcommand{\EUq}{\mathbb{E}_q[u]}

\newcommand{\supp}{\mathrm{supp}}
\newcommand{\dd}{\mathbf{d}}
\newcommand{\DX}{\bigtriangleup(X)}
\newcommand{\DXint}{\bigtriangleup_1(X)}
\newcommand{\succi}{\succsim_{\mathbf{c}}}
\newcommand{\cci}{\sim_{\mathbf{c}}}

\begin{document}

\title{\textbf{A distance-based theory of lottery complexity}}
\author{Giulio Principi\thanks{Affiliation: \textit{New York University}. E-mail: gp2187@nyu.edu.\\ I am extremely indebted and grateful to my advisor Efe Ok for invaluable discussions, support, and suggesting the topic of this paper. I am also very grateful to Simone Cerreia-Vioglio, Tommaso Denti, and especially Fabio Maccheroni for their continued support and very helpful discussions. Furthermore, I am also indebted to Andrea Aveni and Francesco Fabbri for many useful discussions. I thank Dilip Abreu, Arjada Bardhi, Alistair Barton, Pierpaolo Battigalli, Fabio Bellini, Roberto Corrao, Joyee Deb, David Dillenberger, Itzhak Gilboa, Andreas Kleiner, Felix-Benedikt Liebrich, Alberto Maccheroni, Erik Madsen, Alfonso Maselli, Benny Moldovanu, Stefania Minardi, Pietro Ortoleva, David Pearce, Debraj Ray, Todd Sarver, Andrew Schotter, Lorenzo Stanca, Kirtivardhan Singh, Quitze Valenzuela-Stookey, Michele Valinoti, Lucrezia Villa, Peter P. Wakker, Ruodu Wang, and Fan Wang for their comments and suggestions. I thank also the participants of D-TEA2026, IMPMS2026, and FUR2026 for their comments.}}
\date{\today\vspace{-3ex}}
\maketitle

\begin{abstract}

We propose a measure of lottery complexity that combines probabilistic dispersion with dissimilarity between outcomes. Taking degenerate lotteries, which yield a single outcome with certainty, as the simplest alternatives, we measure complexity by proximity to this class. Specifically, for each degenerate lottery, we compute the expected distance between its certain outcome and the lottery’s outcomes, and take the minimum across all degenerate lotteries. We provide an axiomatic characterization establishing uniqueness up to positive rescaling, study how complexity changes under outcome compression, and characterize maximally complex lotteries. Finally, we introduce Complexity-Adjusted Expected Utility preferences and characterize adherence to stochastic dominance and strong risk aversion.

\end{abstract}

\section{Introduction}

A substantial body of empirical evidence on risky choice suggests that individuals systematically undervalue complex options. Complexity affects valuation, generates mistakes, and leads decision makers to favor simpler alternatives. These effects matter in many economic environments. In choice under risk, they can shape the design of insurance products, pension plans, and financial contracts, since participation may depend not only on the distribution of payoffs, but also on the cognitive cost of processing the information contained in these alternatives.\footnote{See for instance \cite{Carlin2009}, \cite{CelerierVallee2017}, and \cite{GoodmanPuri2025}.}

This raises a basic question: what does it mean for a lottery to be complex? Arguably, complexity is a vague and multidimensional notion. Some lotteries may appear complex because they have many possible outcomes; others because probability mass is spread across outcomes; others because the outcomes themselves are far apart. A theory of complexity under risk must therefore identify which features of a lottery make it difficult to evaluate.

The guiding idea of the paper is to measure complexity indirectly, by starting from alternatives that are \textit{unambiguously simple}. Degenerate lotteries provide the natural benchmark. A degenerate lottery assigns probability one to a single outcome and contains no distributional uncertainty. It is therefore the simplest type of lottery. The complexity of a lottery can then be measured by how far it is from this benchmark class, or equivalently by how well it can be summarized by a single outcome. In other words, if you tell us what is simple, then we may be able to tell you what is complex.

Adopting this perspective shifts the question from what complexity means to what it means for a lottery to be far from the class of degenerate lotteries. The answer depends on how the decision maker perceives dissimilarities between outcomes. For monetary lotteries, the absolute payoff distance can provide a natural benchmark. Consider multidimensional lotteries in which, for instance, each outcome specifies medical expenses, property damage, and lost income. These outcomes are vectors of monetary losses and can naturally be compared using the Euclidean distance. For categorical outcomes, for example,  health states recorded as clinical categories or types of insured losses, there may be no natural way to say that one pair of outcomes is closer than another. If the decision maker distinguishes only whether two outcomes coincide, the discrete metric is a natural choice.\footnote{That is, the metric $\dd$ such that $\dd(x,y)=1$ if $x\neq y$ and $\dd(x,y)=0$ otherwise, for all $x,y\in X$.} In other settings, the relevant notion of distance may be different and potentially context-dependent, reflecting which features of the outcomes are salient to the decision maker. To accommodate this variety, we deliberately leave the outcome metric general.
 A decision maker may perceive two outcomes as distant depending on the choice context, cognitive ability, tastes, and further subjective aspects. The perceived outcome dissimilarity is then represented by a metric $\dd$ on an abstract outcome space $X$, where $\dd(x,y)$ measures how different the decision maker regards outcomes $x$ and $y$. We then measure the complexity of a lottery as the expected distance from the class of degenerate lotteries:
$$
C_{\dd}(p)=\min_{x\in X}\int \dd(x,y)\mathrm{d}p(y).
$$

This approach differs from existing alternatives such as support-size and entropy-based measures of complexity (\cite{MononenEntropy}, \cite{PuriSupport}). Support-size measures count how many outcomes the decision maker must consider. Entropy measures the dispersion of probability weights across outcomes. Both capture important dimensions of complexity, but neither accounts for the dissimilarity between outcomes. By contrast, the index $C_{\dd}$ treats complexity as a joint property of probabilities and perceived outcome distances.

A simple example illustrates the distinction. Using the absolute-value metric to compare monetary outcomes, consider three lotteries. Lottery $p$ gives equal chances of losing $\$100$ and gaining $\$100$. Lottery $q$ assigns equal probability to a hundred payoffs evenly spaced between a loss of 50 cents and a gain of 50 cents. Lottery $r$ gives equal chances of losing one cent and gaining one cent. Both support-size and entropy measures rank $q$ as the most complex and treat $p$ and $r$ as equally complex. Our distance-based index instead ranks $p$ as more complex than $q$, and $q$ as more complex than $r$. The outcomes of $p$ are far from any single payoff; those of $q$, despite being numerous, are concentrated near zero; and $r$ is closer still to being degenerate. Thus, a lottery with many nearby outcomes need not be more complex than one with only a few widely separated outcomes. Moreover, $C_{\dd}$ also ranks $q$ as slightly more complex than $r$. Thus, while $C_{\dd}$ is not an entropy measure, it remains sensitive to probability weights through their interaction with outcome distances.

Our main result provides an axiomatic foundation for $C_{\dd}$. The primitive is the decision maker's ranking of lotteries by complexity, denoted by $\succi$: $p\succi q$ means that the decision maker regards $p$ as at least as complex as $q$. The axioms begin from the idea that degenerate lotteries are the simplest alternatives, but their central concept is that of a \textit{best simple approximation}. An outcome $x$ is a best simple approximation of a lottery $p$ if the 50-50 mixture of $p$ and $\delta_x$ is no more complex than the analogous mixture using any other outcome. Best simple approximations are therefore decision-maker-specific: they are the outcomes that, according to the decision maker's complexity ranking, most simplify the lottery when mixed with it in equal proportions.

Most of the axioms describe how complexity behaves around these revealed approximations. A best simple approximation shared by two lotteries acts as a stable center: it remains a best approximation when the lotteries are mixed, and moving from a more complex lottery toward a less complex lottery sharing that center cannot increase complexity. Simplifications must also be comparable across lotteries: if two lotteries are equally complex, replacing the same fraction of each with one of its own best simple approximations preserves equal complexity. These requirements capture, in behavioral terms, the linear way in which expected distance from a fixed benchmark responds to changes in probabilities. The remaining axioms prevent randomization from producing spurious simplifications. Mixing equally complex lotteries cannot generate a strictly simpler lottery, and, starting from a 50-50 binary lottery, placing additional probability on an outcome already in its support cannot be more complex than placing it on a possibly new outcome. 

The representation theorem shows that these axioms characterize $C_{\dd}$: the decision maker's complexity ranking satisfies them if and only if it is represented by $C_{\dd}$ for some continuous metric $\dd$. Under this representation, the revealed best simple approximations coincide with the outcomes that minimize expected distance. Moreover,
\[
C_{\dd}\left(\frac{1}{2}\delta_x+\frac{1}{2}\delta_y\right)
=
\frac{1}{2}\dd(x,y),
\]
so the complexity of 50-50 binary lotteries reveals how dissimilar the decision maker regards their supported outcomes. The metric is therefore recovered from the decision maker's complexity ranking rather than specified independently of it. Different rankings may reveal different metrics; for any fixed ranking, however, the metric, and hence the complexity index, is unique up to positive rescaling. This cardinal uniqueness makes relative magnitudes of complexity meaningful. For instance, statements such as ``lottery $p$ is twice as complex as lottery $q$'' have a well-defined meaning.

Beyond the axiomatic characterization, we examine how lotteries can be simplified, which lotteries are maximally complex, and how complexity relates to familiar comparisons of risk. We introduce median-preserving contractions, which move outcomes toward a best simple approximation, and show that they weakly reduce complexity. We then show how the geometry induced by $\dd$ determines the maximally complex lotteries. For convex outcome sets under Euclidean distance, these lotteries are the ones supported on outcomes farthest from the center of the outcome set, the point that minimizes the maximum distance to any feasible outcome, and that have this center as their expectation. Finally, we relate $C_{\dd}$ to the convex order, which regards one lottery as more dispersed than another when the former is a mean-preserving spread of the latter. For norm-induced distances, mean-preserving spreads weakly increase complexity, and this comparison remains valid after both lotteries are mixed in the same proportion with any common third lottery. On the real line with absolute-value distance, the converse also holds: the complexity comparisons that remain valid under every such common mixture coincide exactly with the convex order.

Finally, the paper applies $C_{\dd}$ to choice under risk. We introduce Complexity-Adjusted Expected Utility (CAEU) preferences, represented by
$$
V(p)=\mathbb E_p[u]-\lambda C_{\dd}(p),
$$
where $u$ is a utility function, $\dd$ is the distance capturing the decision maker’s perceived dissimilarity between outcomes, and $\lambda$ measures the decision maker's attitude toward complexity. When $\lambda>0$, the decision maker is averse to complexity; when $\lambda<0$, complexity is preferred. By identifying the local utilities of CAEU preferences, we characterize consistency with first- and second-order stochastic dominance and with strong risk aversion. The analysis shows that stochastic dominance imposes joint restrictions on utility and complexity sensitivity. In particular, consistency with first-order stochastic dominance requires utility monotonicity to be stronger when the decision maker is more sensitive to complexity. The characterization of (strong) risk aversion shows that, under CAEU, effective risk attitude depends not only on the curvature of $u$, but also on the complexity term. In this model, complexity aversion can therefore be interpreted as a cognitive component of risk aversion.

\subsection{Related Literature}

This paper contributes to the literature on complexity and decision making under risk. The closest papers study models in which complex lotteries are evaluated less favorably. \cite{PuriSupport} axiomatizes a criterion in which expected utility is penalized by a cost that is monotone in the support size of the lottery. \cite{MononenEntropy} studies the event-splitting effect and shows how this phenomenon is represented by an expected-utility model with an entropic complexity cost. In subsequent work, \cite{Mononen2025general} axiomatizes a more general expected utility minus a likelihood-based complexity cost representation.

The present paper takes a different approach. Rather than starting from support size, entropy, or a particular behavioral anomaly, it asks how to measure the complexity of a lottery itself. The proposed index measures how far a lottery is from the class of degenerate lotteries. Thus, complexity is modeled not merely as a property of the number of outcomes or their probability weights, but as a joint property of probabilities and perceived dissimilarities among outcomes. In this sense, the paper complements support-size and entropy-based approaches. Support size captures how many outcomes must be processed; entropy captures how probability mass is distributed across outcomes; $C_{\dd}$ captures how difficult the lottery is to summarize by a single outcome.

A related procedural approach is developed by \cite{Hu2023}. In that model, the decision maker simplifies lotteries through rules that can generate support-size and entropy costs, and also through a partition of outcomes that are close in value. This latter idea is close in spirit to the present paper because it recognizes that complexity depends on the similarity between outcomes. The difference is that here outcome dissimilarity is evaluated through the metric $\dd$, rather than partitioning the outcomes that are close enough depending on a certain cutoff level.

The paper is also related to work on complexity in other domains of decision theory. In decision problems under ambiguity, \cite{VALENZUELASTOOKEY202376} studies complex acts through brackets generated by simpler acts. In menu choice, \cite{Ortoleva2013} studies thinking aversion, where agents dislike large choice sets because choosing from them is cognitively costly. More broadly, \cite{GabaixGraeber2024} and \cite{Gabaix2025} develop models in which complexity is tied to the cognitive technology of solving decision problems. The present paper differs by focusing on the complexity of the risky alternatives themselves: a lottery is complex when it is difficult to approximate by a degenerate lottery.

The measurement of lottery complexity is also connected to the literature on costly information acquisition and rational inattention. Since \cite{Blackwell53}, information structures have been ordered by their informational content; more recent work studies how to measure the cost or difficulty of acquiring and processing information \citep{FraenkelKamenica,WoodfordNeighborhood,DentiPosterior,DentiExpcost,Caplindeanlehay2022,Pomattoetal}. The connection is conceptual: lotteries, like signals, contain information that must be processed by the decision maker. The difference is that this paper studies the perceived complexity of outcome distributions, rather than the cost of acquiring information.

Finally, the paper relates to the mathematical literature on stochastic orders, risk measures, deviation measures, and uncertainty quantification \citep{StoyanMuller,ShakedShantikumar,FollmerSchied,Artzner,RockafellarUryasev,gonzalez-garcia2026quantification}. The paper connects comparisons based on expected distances to the convex order, a stochastic order that formalizes when one distribution is a mean-preserving spread of another. This connection relates the proposed complexity index to classical notions of dispersion.

A large empirical literature documents that complexity affects behavior: it generates mistakes, attenuation, caution, choice avoidance, and preferences for simpler options. Evidence in risky choice includes \cite{HuckWeizsacker1999}, \cite{MadorSonsinoBenzion2000}, \cite{SonsinoBenzionMador2002}, \cite{FudenbergPuri2022}, \cite{FudenbergPuri2023}, and \cite{PuriSupport}; related evidence on cognitive uncertainty and behavioral attenuation includes \cite{EnkeGraeber2023CognitiveUncertainty}, \cite{EnkeGraeberOpreaYang2024Attenuation}, and \cite{KendallOprea2024}. Additional experimental evidence is provided in \cite{deClippelEtAl2025}, where it is shown that agents undervalue complex options and react with caution to complex decision problems as if, complexity triggers a form of “internal ambiguity.’’ Closest to the present paper, \cite{EnkeShubatt2023} and \cite{ShubattYang2024} study the complexity of lottery choice problems. Enke and Shubatt construct empirically estimated complexity indices in which the most important role is played by the state-by-state dissimilarity between the lotteries being compared. Although this is not their main objective, they also construct a measure for an individual lottery by comparing it with a safe payoff equal to its expected value. Shubatt and Yang develop a theory in which pronounced tradeoffs make alternatives more difficult to compare. By contrast, the present paper axiomatizes complexity as an intrinsic property of a lottery, without reference to a comparison alternative. Thus, while all three approaches emphasize dissimilarity, the present paper focuses on perceived dissimilarities among the outcomes of a single lottery rather than between alternatives.

\section{Measuring the complexity of lotteries}\label{sect:measurement}

The framework is that of choice under risk. Let $X$ be a compact metric space, endowed with a metric $\rho$. The elements of $X$ are outcomes and, to avoid triviality, $X$ is assumed to contain at least two distinct elements. By $\DX$ we denote the set of Borel probability measures on $X$, endowed with the topology of weak convergence. Elements of $\DX$ are interpreted as lotteries over outcomes. Under the canonical identification $x\leftrightarrow \delta_x$, we make a slight abuse of notation by treating $X$ as a subset of $\DX$.

The objective of this section is to introduce a measure of lottery complexity. The guiding idea is that degenerate lotteries are the simplest choice alternatives. A degenerate lottery assigns probability one to a single outcome and therefore involves no distributional uncertainty. Nondegenerate lotteries are more complex insofar as they are harder to approximate by any degenerate lottery.

To formalize this idea, let $\dd$ be a continuous metric on $X$. Continuity is understood with respect to the background topology induced by $\rho$. The metric $\dd$ need not coincide with $\rho$. Instead, $\dd(x,y)$ represents the decision maker's perceived dissimilarity between outcomes $x$ and $y$. Thus, two outcomes may be close with respect to $\rho$ but far apart according to the decision maker's subjective perception, $\dd$, or vice versa.

\begin{definition}\label{def:best_approximation_complexity}
A function $C:\DX\to[0,\infty)$ is a \textit{best-approximation complexity index} if there exists a continuous metric $\dd$ on $X$ such that, for every $p\in\DX$,
$$
C(p)=\min_{x\in X}\int \dd(x,y)\mathrm{d}p(y).
$$
We write $C_{\dd}$ for the index induced by the metric $\dd$.
\end{definition}

For a fixed outcome $x$, the quantity
$$
\int \dd(x,y)\mathrm{d}p(y)
$$
is the average subjective distance between the outcomes drawn according to $p$ and the degenerate lottery $\delta_x$. The index $C_{\dd}(p)$ selects the degenerate lottery that best approximates $p$. Hence, $C_{\dd}(p)$ is the average subjective distance of $p$ from its closest degenerate lottery. Equivalently, any minimizer
$$
x_p\in\arg\min_{x\in X}\int \dd(x,y)\mathrm{d}p(y)
$$
can be interpreted as a subjective median of $p$. Under this interpretation, $C_{\dd}(p)$ measures the perceived dispersion of the lottery around its subjective median. A lottery is complex when its probability mass is, on average, far from every outcome that can serve as its best degenerate proxy.

The compactness of $X$ and the continuity of $\dd$ guarantee that such minimizers exist. Indeed, for every $p\in\DX$, the map
$$
x\mapsto \int \dd(x,y)\mathrm{d}p(y)
$$
is continuous on $X$ and therefore attains its minimum.

\begin{example}[Monetary lotteries]\label{ex:monetary_complexity}
Let $X=[a,b]\subseteq\mathbb R$, and let $\dd:(x,y)\mapsto |x-y|$. Then, for every $p\in \DX$,
$$
C_{\dd}(p)=\min_{x\in[a,b]}\int |x-y|\mathrm{d}p(y).
$$
The minimizers are medians of $p$, and $C_{\dd}(p)$ is the mean absolute deviation of the lottery from a median.
\end{example}

\begin{example}[Finite categorical outcomes]\label{ex:discrete_metric_complexity}
Suppose $X$ is finite and $\dd$ is the discrete metric:
$$
\dd:(x,y)\mapsto
\begin{cases}
0, & x=y,\\
1, & x\neq y.
\end{cases}
$$
Then, for every $p\in\DX$,
$$
C_{\dd}(p)=1-\max_{x\in X}p(x).
$$
Thus, complexity is low when the lottery has a salient modal outcome and high when probability mass is spread across many categories. In particular, the uniform lottery is maximally complex among lotteries on $X$.
\end{example}

\begin{example}[Incomplete preferences]
Suppose $X$ is a compact metric space and the decision maker has multiexpected utility preferences $\succsim$ over $\DX$. In particular, we assume that there exists an equicontinuous set of utility functions $\mathcal{U}$ on $X$ that separates points\footnote{That is, for all $x\neq y$ in $X$, there exists $u\in \mathcal{U}$ such that $u(x)\neq u(y)$. An example of such a set $\mathcal{U}$ can be found in the Pareto order on $X=[0,1]^2$, i.e., $\mathcal{U}=\{u_1,u_2\}$ with $u_i(x_1,x_2)=x_i$ for $i=1,2$.} and such that for all $u\in \mathcal{U}$ and $x\in X$, $0\leq u(x)\leq 1$, and 
\[
p\succsim q \Longleftrightarrow \forall u\in \mathcal{U}, \EUp\geq \EUq.
\]
A plausible distance in this setting is:
\[
\dd:(x,y)\mapsto \sup_{u\in \mathcal{U}}\left|u(x)-u(y)\right|
\]
where the idea is that the distance between two outcomes $x,y$ depends on the distance of their farthest respective utility values. The associated complexity measure would become:
\[
C_\dd(p)=\min\limits_{x\in X}\int \sup_{u\in \mathcal{U}}\left|u(x)-u(y)\right|\mathrm{d}p(y)
\]
for every $p\in \DX$.
\end{example}

These examples illustrate the role of the subjective metric. With the absolute-value metric on monetary outcomes, complexity is dispersion around a median. With the discrete metric on categorical outcomes, complexity is lack of concentration on a single outcome. More generally, $C_{\dd}$ measures dispersion according to the decision maker's perceived dissimilarity between outcomes.

\par\medskip
\noindent \textbf{An optimal transport perspective.}
The class of distances we study in this paper is the 1-Wasserstein metric induced by a continuous $\dd$ on $X$. To notice that, recall the definition of the 1-Wasserstein metric $W_{\dd}^1$ on $\DX$,
\[
W_{\dd}^1:(p,q)\mapsto \inf_{\gamma\in \Gamma(p,q)}\int \dd(x,y)\mathrm{d}\gamma(x,y)
\]
where $\Gamma(p,q)$ denotes the set of couplings between $p$ and $q$, that is the set of joint probability measures on $X\times X$ whose marginals are $p$ and $q$ on the first and second factors, respectively. Formally,
\[
\Gamma(p,q)=\left\lbrace \gamma\in \bigtriangleup(X\times X):\forall A\in \mathcal{B}(X),\ \gamma(A\times X)=p(A)\ \textnormal{and}\ \gamma(X\times A)=q(A) \right\rbrace.
\]
It is straightforward to see that $\Gamma(p,\delta_x)=\left\lbrace p\otimes \delta_x \right\rbrace$ for all $p\in \DX$ and $x\in X$, and hence,
\[
W^1_\dd(p,\delta_x)=\int \dd(x,y)\mathrm{d}p(y) 
\]
for all $p\in \DX$ and $x\in X$. Therefore, $C_\dd$ can be rewritten as
\[
C_\dd(p)=\min\limits_{x\in X}W^1_\dd(p,\delta_x).
\]
The choice of the Wasserstein metric is \textit{threefold}. First, it admits an appealing interpretation rooted in the optimal transport problem. Indeed, $W^1_\dd(p,\delta_x)$ can be interpreted as the minimum cost of compressing $p$ to the outcome $x$, where transferring one unit of mass from $y \in \supp(p)$ to $x$ incurs cost exactly $\dd(y,x)$. In terms of complexity, we say that a lottery $p$ is simpler than a lottery $q$ the cheaper it is to compress $p$ to any degenerate outcome, compared with compressing $q$. Second, the framework imposes no substantive restrictions on the underlying metric $\dd$ over outcomes, leaving ample modeling flexibility to the decision environment under investigation. Third, for each $x \in X$, the map $p \mapsto W^1_\dd(p,\delta_x)$ is affine, yielding a high degree of mathematical tractability.\footnote{Notwithstanding the reasons above, the motivation underpinning the definition $C_\dd$ admits broader generalizations. There is no \emph{a priori} reason to restrict attention to the Wasserstein distance on $\DX$; other probability metrics could be employed and may yield interesting results.}

\subsection{Comparison with alternative approaches}\label{sect:comparison}

The best-approximation index $C_{\dd}$ can be compared with two natural approaches to measuring lottery complexity. The first measures complexity by the number of possible outcomes. The second measures it by the probabilistic dispersion of the lottery, employing the entropy. For a finitely supported lottery $p$, these measures are formally defined as
\[
S(p)=|\supp(p)|
\qquad\text{and}\qquad
H(p)=-\sum_{x\in\supp(p)}p(x)\log p(x).
\]
Support size captures the number of contingencies to be considered, as in \cite{PuriSupport}, whereas entropy accounts for how evenly probability is distributed across them, as in \cite{MononenEntropy}. The index $C_{\dd}$ can agree with either approach when these features move together with outcome dispersion. Unlike them, however, it also accounts for how dissimilar the outcomes are.
\par\medskip
\noindent \textbf{Support size comparison.} Let $\dd(x,y)=|x-y|$ and define
\[
p=\frac{1}{2}\delta_{-1}+\frac{1}{2}\delta_1,
\qquad
q_{a}
=
\frac{1}{5}\delta_{-2a}+\frac{1}{5}\delta_{-a}+\frac{1}{5}\delta_0+\frac{1}{5}\delta_{a}+\frac{1}{5}\delta_{2a},
\]
where $a>0$. The lottery $q_{a}$ assigns equal probability to $5$ equally spaced outcomes, with $a$ denoting the distance between adjacent outcomes. Support size always ranks $q_{a}$ as more complex than $p$, since
\[
S(q_{a})=5>2=S(p).
\]
Since zero is the median of $q_{a}$,
\[
C_{\dd}(p)=1\ \textnormal{and}\ C_\dd(q_a)=\frac{6}{5}a.
\]
Thus, $C_{\dd}$ agrees with support size when $a>5/6$: the five outcomes of $q_a$ are sufficiently dispersed that $q_a$ is also more complex according to the distance-based measure. When $a<5/6$, however, $C_{\dd}$ ranks $q_a$ as simpler than $p$, despite its larger support, because its five outcomes are concentrated in a sufficiently small region. Additional support points increase distance-based complexity when they are sufficiently far apart, but need not do so when they are closely clustered.

\par\medskip
\noindent\textbf{Entropy comparison.} A similarly concise comparison applies to entropy. For $a>0$ and $\pi\in(0,1/2]$, consider
\[
r_{\pi,a}=\pi\delta_a+(1-\pi)\delta_0.
\]
Under the absolute-value metric $\dd(x,y)=|x-y|$,
\[
C_{\dd}(r_{\pi,a})=\pi a\ \ \textnormal{and}\ \
H(r_{\pi,a})
=-\pi\log\pi-(1-\pi)\log(1-\pi).
\]
Holding $a$ fixed, both measures increase with $\pi$: as the probabilities become more evenly distributed, the lottery becomes more complex according to both entropy and $C_{\dd}$. Holding $\pi$ fixed, however, entropy is unaffected by $a$, whereas $C_{\dd}$ increases proportionally with the distance between the two outcomes. Consequently, entropy and $C_{\dd}$ agree when probabilistic dispersion is the relevant source of complexity but may disagree when differences in outcome dispersion dominate.

\par\medskip

Overall, support size, entropy, and $C_{\dd}$ capture distinct dimensions of complexity. Support size captures the number of possible outcomes. Entropy captures the dispersion of probability mass across these outcomes. The index $C_{\dd}$ captures perceived dispersion in the outcome space. The advantage of $C_{\dd}$ is that it treats complexity as a joint property of probabilistic dispersion and outcome dissimilarity, rather than as a property of probability weights or support alone. An important observation is that, outside of trivial cases, there is no distance $\dd$ so that the entropy or the support size approaches are equivalent to $C_\dd$. If $X$ is finite and $|X|\geq 3$, no metric $\dd$ makes the ranking induced by $C_\dd$ equivalent to entropy. Support size is not ordinally equivalent to some $C_\dd$ whenever $|X|\geq 2$.\footnote{In the appendix this is shown formally, see Appendix \ref{remarksuppentropy}.}

\subsection{Axiomatic characterization and uniqueness}

This section provides a ranking-based foundation for the complexity index $C_{\dd}$. The primitive object is a binary relation $\succsim_{\mathbf{c}}$ on $\DX$, where
\[
p\succsim_{\mathbf{c}}q
\]
is interpreted as ``lottery $p$ is weakly more complex than lottery $q$.'' The objective is to identify conditions under which this ranking can be represented by a best-approximation complexity index,
\[
C_{\dd}(p)=\min_{x\in X}\int \dd(x,y)\mathrm{d}p(y),
\]
for some continuous metric $\dd$ on $X$.

Experimental studies link complexity to noisier choices, weaker sensitivity to payoff-relevant fundamentals, and greater caution when choosing among difficult-to-evaluate options. These findings motivate treating perceived complexity as a systematic feature of lotteries. To provide a behavioral foundation for the proposed measure, we take the decision maker's comparative complexity judgments as primitive and represent them by the binary relation $\succi$. The axioms take degenerate lotteries as benchmarks of simplicity and describe how perceived complexity changes when lotteries are mixed with one another or with suitable degenerate proxies.

We begin with two basic requirements. Axiom \ref{mx1} requires $\succi$ to be a continuous weak order. Axiom \ref{mx3} identifies degenerate lotteries as exactly the least complex lotteries.
\begin{axiom}\label{mx1}
$\succsim_{\mathbf{c}}$ is complete, transitive, and continuous with respect to the weak convergence topology.\footnote{We recall that a binary relation $\succsim$ on $\DX$ is:
\begin{itemize}
\item Complete: for all $p,q\in \DX$, either $p\succsim q$ or $q\succsim p$.
\item Transitive: for all $p,q,r\in \DX$,  $p\succsim q$ and $q\succsim r$ imply $p\succsim r$.
\item Continuous: for all $p\in \DX$, the following sets
\[
\left\lbrace q\in \DX:q\succsim p\right\rbrace\ \textnormal{and}\ \left\lbrace q\in \DX:p\succsim q\right\rbrace
\]
are closed with respect to the weak convergence topology.
\end{itemize}
}
\end{axiom}

\begin{axiom}\label{mx3}
For all $p\in \DX$ and $x\in X$,\ $p\succsim_{\mathbf{c}}\delta_x$ and $p\succ_{\mathbf{c}}\delta_x$ if $p$ is not degenerate.\footnote{By $\succ_{\mathbf{c}}$ we denote the asymmetric part of $\succsim_{\mathbf{c}}$, that is, $p\succ_{\mathbf{c}}q$ if and only if $p\succsim_{\mathbf{c}}q$ and $q\not\succsim_{\mathbf{c}}p$.}
\end{axiom}
\noindent Notice that this axiom implies also that degenerates are all equally simple.
\par\medskip
The ranking is also required to satisfy a convexity property. If two lotteries are equally complex, then randomizing between them cannot make the resulting lottery strictly simpler than both. Intuitively, mixing two objects of the same complexity does not lead to a simplification: the mixture must remain at least as complex as each of the original lotteries.
\begin{axiom}\label{mx4}
For all $p,q\in \DX$,
\[
p\sim_{\mathbf{c}}q \Longrightarrow \forall \alpha\in [0,1],\ \alpha p+(1-\alpha)q\succsim_\mathbf{c} p.\footnote{By $\sim_{\mathbf{c}}$ we denote the symmetric part of $\succsim_{\mathbf{c}}$, that is, $p\sim_{\mathbf{c}}q$ if and only if $p\succsim_{\mathbf{c}}q$ and $q\succsim_{\mathbf{c}}p$.}
\]
\end{axiom}
Before presenting the next axioms, we need some further notation. We denote by $A_{\succsim_{\mathbf{c}}}:\DX\rightrightarrows X$ the correspondence defined as:
\[
A_{\succsim_{\mathbf{c}}}(p)=\left\lbrace x\in X:\forall y\in X,\ \frac{p+\delta_y}{2}\succsim_{\mathbf{c}}\frac{p+\delta_x}{2}  \right\rbrace
\]
for all $p\in \DX$. The set $A_{\succsim_{\mathbf{c}}}(p)$ contains the degenerate lotteries that best simplify $p$. Indeed, $x\in A_{\succsim_{\mathbf{c}}}(p)$ means that mixing $p$ with $\delta_x$ produces a lottery weakly less complex than the lottery obtained by mixing $p$ with any other degenerate lottery. Elements of $A_{\succsim_{\mathbf{c}}}(p)$ are therefore interpreted as best degenerate proxies for $p$.

The next two restrictions govern how these best proxies behave under mixtures. First, if two lotteries share the same best degenerate proxy, then every mixture of them has the same best proxy. Second, along mixtures of lotteries with a common best proxy, complexity is monotone: moving from a more complex lottery toward a less complex one cannot increase complexity. Thus, when $p$ and $q$ are both best summarized by the same outcome $x$, and $p$ is more complex than $q$, the lottery $\alpha p+(1-\alpha)q$ lies between them in the complexity order. Figure \ref{fig:mx6} illustrates how Axiom \ref{mx6}, operating alongside the quasiconcavity guaranteed by axioms \ref{mx1} and \ref{mx4}, ensures the mixture lies exactly between $p$ and $q$ in the complexity order

\begin{axiom}\label{mx5}
For all $p,q\in \DX$, $x\in A_{\succsim_{\mathbf{c}}}(p)\cap A_{\succsim_{\mathbf{c}}}(q)$, and $\alpha\in [0,1]$,
\[
x \in A_{\succsim_{\mathbf{c}}}(\alpha p+(1-\alpha)q).
\]
\end{axiom}

\begin{axiom}\label{mx6}
For all $p,q\in \DX$, $x\in A_{\succsim_{\mathbf{c}}}(p)\cap A_{\succsim_{\mathbf{c}}}(q)$, and $\alpha\in [0,1]$,
\[
p\succsim_{\mathbf{c}}q \Longrightarrow  p\succsim_\mathbf{c} \alpha p+(1-\alpha)q.
\]
\end{axiom}

\begin{figure}[t]
    \centering
    \includegraphics[width=0.82\textwidth]{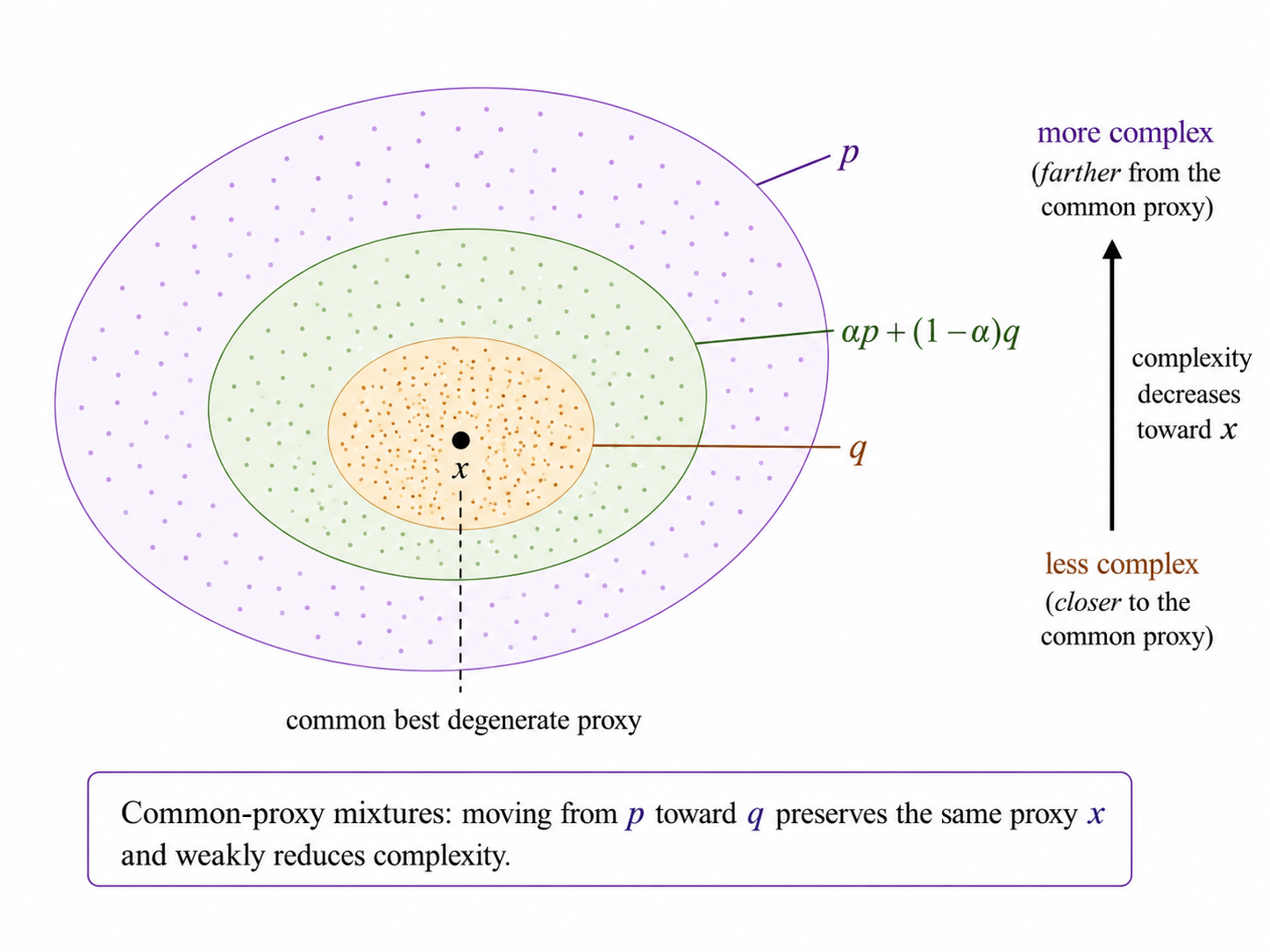}
    \caption{The lotteries $p$, $q$, and $\alpha p+(1-\alpha)q$ share a common best degenerate proxy $x$. The lottery $p$ is depicted as more complex than $q$, while the mixture $\alpha p+(1-\alpha)q$ lies between them and preserves the same proxy. The figure illustrates the requirement that moving from $p$ toward $q$ along common-proxy mixtures weakly reduces complexity.}
    \label{fig:mx6}
\end{figure}

The following axiom disciplines how complexity scales in mixing a lottery with one of its best degenerate proxies. This axiom is a proportionality requirement stated ordinally: if two lotteries are equally complex, then replacing the same fraction of each lottery by one of its own best degenerate proxies preserves indifference. This restriction ensures that complexity scales linearly along mixtures connecting lotteries to their best degenerate approximations.

\begin{axiom}\label{mx7}
For all $p,q\in \DX$, $x\in A_{\succsim_{\mathbf{c}}}(p)$, $y\in A_{\succi}(q)$, and $\alpha\in [0,1]$,
\[
p\cci q \Longrightarrow \alpha p+(1-\alpha)\delta_x \cci \alpha q+(1-\alpha)\delta_y.
\]
\end{axiom}

The last axiom asks complexity of a very specific class of lotteries to be monotonic in the entropy of the lottery. Specifically, suppose that we have a 50-50 lottery. The axiom asks that adding some weight to one of the outcomes in the support of the lottery leads to a simpler lottery than giving that weight to a third, possibly distinct, outcome. Therefore, reducing the entropy of a 50-50 lottery makes it easier.

\begin{axiom}\label{mx9}
For all $x,y,z\in X$ and $\alpha\in [0,1]$,
\[
\alpha\frac{\delta_x+\delta_y}{2}+(1-\alpha)\delta_z\succsim_{\mathbf{c}}\alpha\frac{\delta_x+\delta_y}{2}+(1-\alpha)\delta_y.
\]
\end{axiom}
\noindent This restriction is also consistent with the view that complexity can increase with the number of distinct components the decision maker must process. Starting from a 50-50 lottery, increasing the probability of an outcome already in its support preserves the set of components to be evaluated. Assigning the same probability mass to a possibly distinct outcome instead introduces an additional component into the evaluation. For 50-50 lotteries, the axiom imposes that adding mass to an existing outcome cannot be more complex than adding that mass to a new one.

The main result of this paper is the following theorem that provides an axiomatic foundation for the proposed complexity measure.

\begin{theorem}\label{axiomatic_characterization}
Let $\succsim_{\mathbf{c}}$ be a binary relation on $\DX$. The following are equivalent:
\begin{enumerate}
\item $\succsim_{\mathbf{c}}$ satisfies axioms \ref{mx1}-\ref{mx9}.
\item $\succsim_{\mathbf{c}}$ admits a best-approximation complexity index representation for some continuous metric $\dd$.
\end{enumerate}
\end{theorem}
\noindent The theorem shows that the preceding axioms are jointly necessary and sufficient for the complexity ranking to admit a best-approximation representation. Importantly, the representation connects the behaviorally defined best degenerate proxies to the minimizers of the distance criterion: for every $p\in\DX$,
\[
A_{\succsim_{\mathbf c}}(p)=\arg\min_{x\in X}\int \dd(x,y)\mathrm{d}p(y).
\]
Thus, the subjective medians of a lottery are revealed by the complexity ranking rather than imposed as additional primitives. The theorem also gives behavioral content to the underlying metric: complexity comparisons among lotteries encode perceived dissimilarities between outcomes. In particular, the complexity of 50-50 binary lotteries over $x$ and $y$ encodes the distance $\dd(x,y)$, as
\[
\dd(x,y)=2C_\dd\left(\frac{\delta_x+\delta_y}{2}\right).
\]
The next result strengthens this observation by showing that the metric is identified up to a common positive scale.

\subsubsection{Cardinal uniqueness and the complexity of lottery choice problems}

The representation obtained above is not merely ordinal: the complexity measure $C_{\dd}$ is cardinally unique.

\begin{proposition}\label{prop:measurecardinaluniqueness}
Let $\dd_1$ and $\dd_2$ be continuous metrics on $X$. Then $C_{\dd_1}$ and $C_{\dd_2}$ are ordinally equivalent if and only if there exists $\kappa>0$ such that
\[
\dd_1=\kappa \dd_2.
\]
\end{proposition}
\noindent The key observation behind \Cref{prop:measurecardinaluniqueness} is that the value of $C_{\dd}$ on binary lotteries directly identifies the underlying metric. Indeed, as we observed also above, for all $x,y\in X$,
\[
C_\dd\left(\frac{\delta_x+\delta_y}{2}\right)=\frac{\dd(x,y)}{2}.
\]
Thus, once the complexity ranking is represented by a best-approximation measure, the induced metric is pinned down by the complexity of two-point lotteries. In particular, the only strictly increasing transformations of $C_{\dd}$ that preserve the best-approximation structure are linear. If $\phi\circ C_{\dd}$ is again representable as $C_{\dd'}$ for some metric $\dd'$, then $\phi$ must be linear: $\phi(t)=\kappa t$ for some $\kappa>0$.\footnote{Linearity is intended on the bounded range of $C_\dd$, that is $C_\dd(\DX)$.}

This uniqueness property is useful both conceptually and empirically.\\ Conceptually, it means that the proposed measure does not merely rank lotteries from simpler to more complex. It also assigns meaningful relative magnitudes to differences in complexity. For instance, statements such as ``lottery $p$ is twice as complex as lottery $q$'' are meaningful thanks to cardinal uniqueness. Empirically, this makes the measure suitable for quantitative applications: one can estimate the complexity cost associated with a lottery and compare complexity across treatments (after proper normalization).

The result also extends the scope of the approach beyond the complexity of individual lotteries. This paper focuses on lotteries as the primitive objects whose complexity is being measured. However, a relevant part of the decision-theoretic and experimental literature is interested in the complexity of choice problems. In particular, in this strand of the literature, many papers study how complex the comparison between two lotteries $p$ and $q$ is, so to say, the complexity of the binary choice problem with alternatives $p$ and $q$. Thanks to the cardinal uniqueness result, it is possible to aggregate $C_\dd(p)$ and $C_\dd(q)$ homogeneously to get a value interpretable as the aggregate complexity of a binary choice problem. Two natural examples are:
\[
K_{\dd}(p,q)=C_\dd(p)+C_\dd(q)\ \textnormal{and}\ H_\dd(p,q)=\max\left\lbrace C_\dd(p),C_\dd(q) \right\rbrace
\]
The first, $K_\dd$, captures the total amount of complexity faced by the decision maker, while the second, $H_\dd$, captures the idea that the difficulty of the problem is determined by its most complex alternative. Other aggregators are possible, but the cardinal uniqueness of $C_{\dd}$ disciplines these transformations: since $C_{\dd}$ is unique up to a positive scalar, any homogeneous aggregator inherits the same cardinal property.

This observation is especially useful for models of stochastic choice. Here, the complexity of a choice problem affects the precision with which alternatives are compared rather than their deterministic valuations. Let $\rho$ be a binary stochastic choice rule, where $\rho(p,q)$ denotes the probability that the decision maker chooses lottery $p$ over lottery $q$. Let $v$ denote the decision maker's deterministic valuation of lotteries, and let $G_{\dd}(p,q)$ be either $K_{\dd}(p,q)$ or $H_{\dd}(p,q)$. A simple complexity-dependent stochastic choice rule is
\[
\rho(p,q)
=
F\left(
\frac{v(p)-v(q)}
{\lambda G_{\dd}(p,q)}
\right),
\]
for all lotteries $p,q$ with $G_\dd(p,q)>0$, and where $\lambda>0$ measures the effect of complexity on noise and $F$ is a strictly increasing and continuous distribution function symmetric around zero.

Holding the difference in valuations fixed, greater choice-problem complexity increases the noise scale and moves the choice probability toward one half. Choices therefore become less sensitive to differences in valuation and closer to random choice. The assumption that complexity increases noise is an additional modeling hypothesis; the role of cardinal uniqueness is to ensure that this specification does not depend on an arbitrary numerical representation of complexity. Indeed, if $C_{\dd}$ is multiplied by a positive constant $\kappa$, then $G_{\dd}$ is multiplied by the same constant, and this change can be absorbed by replacing $\lambda$ with $\lambda/\kappa$. By contrast, if the complexity index were only ordinally unique, an arbitrary increasing transformation of it would generally change the resulting choice probabilities and could not be absorbed by a single parameter. 

\section{Compression, maximal complexity, and dispersion}\label{sect:measure_properties}

This section is devoted to explore some of the properties of $C_\dd$ and their implications. First, we formalize the intuitive idea that compressing a lottery towards one of its best degenerate approximation leads to a simpler lottery. Second, we analyze under which assumptions on the set of outcomes and the metric $\dd$ we are able to identify the most complex lotteries. Third, we relate our complexity measure with integral stochastic order, with the purpose to investigate the relation between $C_\dd$ and the most prominent order of dispersion: the convex order.

\subsection{Median-preserving contractions}

The measure $C_\dd$ evaluates the distance of a lottery from its closest degenerate lottery, that is the distance between $p$ and one of its $\dd$-medians.\footnote{We recall that an outcome $x_p$ is a $\dd$-median for a lottery $p$ if:
\[
x_p\in \arg\min_{x\in X}\int \dd(x,y)\mathrm{d}p(y).
\]} It is therefore natural to ask whether a lottery becomes simpler when its outcomes are moved closer to one of its medians. The next result formalizes this idea through median preserving contractions. Fix a continuous metric $\dd$ on $X$, and let $p\in\DX$. A map $T:X\to X$ is a $\dd$-\textit{median-preserving contraction} for $p$ if it is $1$-Lipschitz with respect to $\dd$, that is,
\[
\dd(T(x),T(y))\leq \dd(x,y)
\quad\text{for all }x,y\in X;
\]
and
there exists
\[
x_p\in \arg\min_{x\in X}\int \dd(x,y)\mathrm{d}p(y)
\]
such that $T(x_p)=x_p$. Thus, a median-preserving contraction moves outcomes without increasing any pairwise distance and leaves one best degenerate approximation of $p$ unchanged. In this sense, its application to a lottery leads to a contraction of the lottery around one of its medians.

\begin{proposition}\label{prop:medianprescontr}
Suppose $\dd$ is a continuous metric on $X$, and let $p\in\DX$. If $T$ is a $\dd$-median-preserving contraction for $p$, then
\[
C_\dd(p)\geq C_\dd(T_{\#}p).\footnote{We recall that $T_\#p(A)=p(T^{-1}(A))$ for all Borel sets $A\subseteq X$}.
\]
\end{proposition}
\noindent
The proposition says that applying a median-preserving contraction to a lottery weakly reduces its complexity. The intuition is direct. Since $T$ fixes a median $x_p$ and does not increase distances, every outcome in the transformed lottery is weakly closer to $x_p$ than the corresponding original outcome. The transformed lottery may have a different median, but its average distance from the original median $x_p$ is already weakly lower. Hence its distance from its own closest median, and therefore its complexity, cannot be higher. The following example shows a simple median-preserving contraction.

\begin{example}
Let $X=[-2,2]$, let $\dd(x,y)=|x-y|$, and consider the lottery
$$
p=\frac{1}{4}\delta_{-2}+\frac{1}{2}\delta_0+\frac{1}{4}\delta_2.
$$
The point $x_p=0$ is a median of $p$. Define
$$
T:y\mapsto \min\{1,\max\{-1,y\}\}.
$$
The map $T$ truncates all outcomes to the interval $[-1,1]$. It is $1$-Lipschitz and satisfies $T(0)=0$, hence it is a median-preserving contraction for $p$. The transformed lottery is
$$
T_{\#}p
=
\frac{1}{4}\delta_{-1}
+
\frac{1}{2}\delta_0
+
\frac{1}{4}\delta_1.
$$
Therefore,
$$
C_\dd(p)=1
\qquad\text{and}\qquad
C_\dd(T_{\#}p)=\frac{1}{2}.
$$
Thus, truncating the extreme outcomes toward the median reduces the complexity of the lottery.
\end{example}

\subsection{Maximal and minimal complexity}\label{sect:maximal}
It is obvious that $C_\dd(\delta_x)=0$ for all $x\in X$ and hence finding the minimally complex lotteries is a trivial exercise. On the other hand, the search of maximally complex lotteries is not trivial, and it strongly depends on the geometric and analytical features of $X$ and $\dd$. We start with two results that pave the way to more general observations. In the case of monetary lotteries with $X=[a,b]$ and $\dd$ being the absolute value metric, we have the following lemma.

\begin{proposition}\label{lem:maxcompllott}
If $X=[a,b]\subseteq \R$ and $\dd$ is the absolute value metric, then $1/2\delta_a+1/2\delta_b\in \arg\max_{p\in \DX}C_\dd(p)$.
\end{proposition}
\noindent Under the assumptions of the lemma, the complexity measure $C_\dd$ corresponds to the mean-median deviation on the real line, and its unique maximum is represented by the lottery that assigns equal probability to the interval extremes. This result suggests that if the decision maker perception of outcome dissimilarity is representable by the absolute value metric, then decision maker equates the complexity of a lottery with its level of outcome dispersion.

If the outcome space is categorical and hence very different from the monetary case, and, also, the decision maker can only distinguishes two outcomes apart, then the result changes quite a lot.
\begin{proposition}\label{lem:uniformdiscrete}
If $X$ is finite and $\dd$ is the discrete metric, then 
\[
\sum_{x\in X}\frac{1}{|X|}\delta_x\in \arg\max_{p\in \DX}C_\dd(p).\footnote{Where we recall that by $|X|$ we intend the cardinality of $X$.}
\]
\end{proposition}
\noindent Thus, contrarily to what happens in the case of monetary outcomes (Proposition \ref{lem:maxcompllott}) for discrete metric spaces, the uniform distribution on the whole space is indeed the maximally complex lottery. This highlights how the features of $X$ and $\dd$ play a major role in determining the maxima of $C_\dd$. 

However, these two results are not completely unrelated as one may think at first glance. Indeed, suppose $X$ is finite with just two elements, then the maximally complex lottery is exactly the uniform of the, only, two extremes. More generally, suppose that $X$ is not just a doubleton and that we isometrically embed its elements in a Euclidean space. Then, all these points are equally distant, forming a regular simplex (if $|X|=3$ it would correspond to an equilateral triangle). Therefore, Proposition \ref{lem:uniformdiscrete}, suggests that in a regular simplex, the maximally complex lottery is exactly the uniform over all the extreme points (the vertices of the triangle). This is because the maximally complex lottery corresponds to the one that assigns equal weights to all points that are equally and maximally distant from the barycenter. This is exactly also what happens when $X=[a,b]$, the uniform distribution on the extreme points is the maximally complex lottery, and it is the only one for which each point is equally and maximally distant from the barycenter. 

These observations suggest that, in some settings, it is possible to characterize maximally complex lotteries considering the distances of points from the \say{barycenter} of the set of outcomes. More generally, across convex structures with norm-induced metrics, maximal complexity remains deeply tied to the geometric boundaries of the outcome space. First, we notice that one can always make a lottery more complex by spreading it towards the extreme points.

\begin{proposition}\label{prop:extremeasure}
Suppose $X$ is a convex and compact subset of a finite-dimensional vector space and $\dd$ is induced by a continuous norm.\footnote{Meaning that $\dd(x,y)=\lVert x-y \rVert$ for some continuous norm $\lVert \cdot\rVert$ on the vector space.} For all $p\in \DX$, there exists $p^*\in \bigtriangleup(\textnormal{ext}(X))$ such that $C_\dd(p^*)\geq C_\dd(p)$.\footnote{We recall that $\textnormal{ext}(X)$ denotes the set of extreme points of $X$, i.e., $z\in \textnormal{ext}(X)$ if there are not distinct $x,y\in X$ and $\alpha\in (0,1)$ such that $z=\alpha x+(1-\alpha)y$.}
\end{proposition}
\noindent Therefore, to find \textit{some} maximally complex lotteries, with respect to $C_\dd$, it is sufficient to look for lotteries supported on the extremal outcomes. 

Moreover, the previous discussion about barycenters suggests that it is possible to characterize \textit{all} the maximally complex lotteries. To this end, we recall some useful mathematical jargon. Given a compact, convex, and nonempty subset of a normed space the value 
\[
r^*(X)=\inf\limits_{x\in X}\sup\limits_{y\in X}\lVert x-y\rVert
\]
is finite and it is called the \textit{Chebyshev radius of} $X$. The map $R:X\to \R$ defined as $R(x)=\sup_{y\in X}\lVert x-y\rVert$ for all $x\in X$ is convex and $1$-Lipschitz, therefore
\[
\arg\min_{x\in X}R(x)\neq \emptyset
\]
and these minimizers are called the \textit{Chebyshev centers of} $X$. By definition, $R(x)=r^*(X)$ for all $x\in \arg\min_{x\in X}R(x)$. The Chebyshev radius is a natural upper bound for $C_\dd$, indeed, if $\dd$ is induced by $\lVert\cdot\rVert$, for all $p\in \DX$,
\[
C_\dd(p)=\min\limits_{x\in X}\int \lVert x-y\rVert
\mathrm{d}p(y)\leq\min\limits_{x\in X}\int \sup\limits_{z\in X}\lVert x-z\rVert \mathrm{d}p(y)= r^*(X). 
\]
\noindent Propositions \ref{lem:maxcompllott} and \ref{lem:uniformdiscrete} suggest that maximally complex lotteries $p$ are precisely those that attain the Chebyshev radius. This is shown to be true in much more general settings. 
\begin{remark}
In what follows whenever we refer to $r^*(X)$ and Chebyshev centers we intend those obtained with the continuous norm $\lVert \cdot \rVert$ that will induce $\dd$.
\end{remark}

\begin{proposition}\label{prop:maximalelementsnormedspaces}
Suppose $X$ is a compact, convex, and nonempty subset of finite-dimensional vector space. If $\dd$ is induced by a continuous norm, then,
\[
\arg\max_{p\in \DX}C_\dd(p)=\left\lbrace p\in \DX:C_\dd(p)=r^*(X) \right\rbrace.
\]
\end{proposition}
\noindent The proof of this proposition is based on the application of Sion's minimax theorem. Under additional restrictions on the norm that induces $\dd$ it is possible to refine the previous result. For any Chebyshev center $x^*$ of a set $X$, we denote by
\[
C(x^*)=\left\lbrace y\in X:\lVert y-x^*\rVert=r^*(X) \right\rbrace
\]
the set of contact points of $x^*$. We have the following refinement.
\begin{corollary}\label{coro:maxRnmeasure}
Suppose $X$ is a compact, convex, nonempty subset of $\R^n$ and $\dd$ is induced by the Euclidean norm. Then, $X$ admits a unique Chebyshev center $x^*$,
\begin{equation}\label{eq:equationmaximumRn}
\arg\max_{p\in \DX}C_\dd(p)=\left\lbrace p\in \DX :p(C(x^*))=1,\ \int y\mathrm{d}p(y)=x^*\right\rbrace,
\end{equation}
and $C(x^*)\subseteq \textnormal{ext}(X)$.
\end{corollary}
\noindent Therefore, the maximally complex lotteries are exactly the ones concentrated on $C(x^*)$ that exhibit as barycenter (mean) exactly $x^*$. 

To retrieve a unique maximally complex lotteries, we can restrict to simplices. Indeed, in a simplex $X$, for each point $x\in X$, there is unique vector of weights, called \textit{barycentric coordinates}, such that the convex combinations of the vertices of $X$ with respect to such weights returns exactly $x$. Therefore, by the uniqueness of the Chebyshev center in simplices endowed with the Euclidean norm, we have the following.
\begin{corollary}\label{coro:maxsimplices}
Suppose $X\subseteq \R^n$ is a simplex and $\dd$ is induced by the Euclidean norm. Then, $X$ has a unique Chebyshev center $x^*$ and
\[
\arg\max_{p\in \DX}C_\dd(p)=\left\lbrace \sum_{v\in \textnormal{ext}(X)}\alpha^*_v\delta_v\right\rbrace,
\]
where $(\alpha^*_v)_{v\in \textnormal{ext}(X)}$ are the barycentric coordinates of $x^*$ of $X$.
\end{corollary}
\noindent This corollary yields the following, that is the convex counterpart of Proposition \ref{lem:uniformdiscrete}.
\begin{corollary}\label{coro:maxregularsimplices}
Suppose $X\subseteq \R^n$ is a regular simplex and $\dd$ is induced by the Euclidean norm. Then,
\[
\arg\max_{p\in \DX}C_\dd(p)=\left\lbrace \frac{1}{\lvert \textnormal{ext}(X) \rvert}\sum_{v\in \textnormal{ext}(X)}\delta_v\right\rbrace.
\]
\end{corollary}

\subsection{The affine core of $C_{\dd}$ and the convex order}\label{sect:affinecore}

The measure $C_{\dd}$ can be interpreted as an index of perceived dispersion. In choice under risk, several orders have been proposed to compare the variability of lotteries. Among the most important is the convex order, which ranks lotteries according to their expectations under all convex functions. In this section, we relate $C_{\dd}$ to this tradition by studying the affine core of the order induced by $C_{\dd}$.\footnote{The affine core of a binary relation $\succsim$ on $\DX$, is the largest affine subrelation of $\succsim$. We recall that a binary relation $\succsim$ is \textit{affine} if for all $p,q,r\in \DX$ and $\alpha\in (0,1]$,
\[
p\succsim q \Longleftrightarrow \alpha p+(1-\alpha)r\succsim\alpha q+(1-\alpha)r,
\]
The affinity condition is often called \textit{independence axiom}.}

Fix a continuous metric $\dd$ on $X$. The \textit{affine core} of $C_\dd$ is the binary relation $\geq_{\dd}$ on $\DX$ defined by
$$
p\geq_{\dd} q
\Longleftrightarrow
\forall \alpha\in [0,1),\ \forall \ell\in \DX,\quad
C_{\dd}(\alpha p+(1-\alpha)\ell)
\geq
C_{\dd}(\alpha q+(1-\alpha)\ell).
$$
Thus, $p\geq_{\dd}q$ means that $p$ remains weakly more complex than $q$ after both lotteries are mixed with the same lottery $\ell$, with weight $1-\alpha$. The relation $\geq_{\dd}$ therefore captures the affine part of the complexity ranking induced by $C_{\dd}$. It compares $p$ and $q$ only when the comparison survives every common affine perturbation, $p\geq_{\dd}q$ requires the comparison to be stable across all such mixtures. In this sense, the affine core extracts from $C_{\dd}$ the part of the complexity ranking that behaves linearly with respect to mixtures.

Affine cores are central in decision theory because they identify the \textit{local utilities} of  nonexpected-utility criteria.\footnote{In the sense of \cite{Machina_EU_no_indep}. Local utilities are useful because they describe the local expected-utility approximations of non-EU preferences and allow one to study consistency with stochastic dominance and other integral stochastic orders.} The following proposition shows that, in the present setting, the \say{local utilities} of $C_\dd$ are the distance functions $\dd(x,\cdot)$, indexed by the outcomes $x\in X$.

\begin{proposition}\label{prop:affinecoremeasure}
For all $p,q\in \DX$,
$$
p\geq_{\dd} q
\Longleftrightarrow
\forall x\in X,\quad
\int \dd(x,y)\mathrm{d}p(y)
\geq
\int \dd(x,y)\mathrm{d}q(y).
$$
\end{proposition}
\noindent The proposition shows that the affine core of $C_{\dd}$ coincides with the integral stochastic order generated by the family of functions
$$
\left\{\dd(x,\cdot):x\in X\right\}.
$$
Thus, $p\geq_{\dd}q$ if and only if, for each outcome $x$, lottery $p$ has weakly larger average distance from $x$ than lottery $q$. This characterization clarifies the relation between $C_{\dd}$ and stochastic orders of dispersion. The order $\geq_{\dd}$ is generally incomplete: two lotteries may be difficult to rank if one is farther from some benchmark outcomes while the other is farther from different benchmark outcomes. The index $C_{\dd}$ completes this comparison by focusing on the best approximating degenerate for each lottery. Hence, while $\geq_{\dd}$ compares lotteries by requiring larger average distance from every benchmark, $C_{\dd}$ compares them by their distance from the closest benchmark. The measure $C_{\dd}$ therefore is a completion of the distance-based integral stochastic order induced by $\{\dd(x,\cdot):x\in X\}$.

\subsubsection{Complexity, dispersion, and the convex order}\label{sect}

The characterization above, Proposition \ref{prop:affinecoremeasure}, makes the connection with classical dispersion orders immediate. The affine core is itself an integral stochastic order, generated by the family of distance functions $\{\dd(x,\cdot):x\in X\}$. Integral stochastic orders provide standard tools for comparing the dispersion of random variables and lotteries. In decision theory under risk, the most prominent dispersion order is the convex order. This order formalizes the idea that one lottery is more dispersed than another while preserving its mean.

Let $Y$ be a convex subset of a normed space, and denote by $\mathrm{cvx}(Y)$ the set of continuous and convex functions $\phi:Y\to\mathbb R$. The convex order on $\triangle(Y)$, denoted by $\geq_{\mathrm{cvx}}$, is defined by
\[
p\geq_{\mathrm{cvx}}q
\Longleftrightarrow
\forall \phi\in \mathrm{cvx}(Y),\quad
\int \phi(y)\mathrm{d}p(y)
\geq
\int \phi(y)\mathrm{d}q(y),
\]
In the case of monetary lotteries, $p\geq_{\mathrm{cvx}}q$ is equivalent to saying that $p$ is a mean-preserving spread of $q$.

\begin{example}\label{ex:compactmps}
Let $X=[-1,1]$. The lottery $p=\frac{1}{2}\delta_{-1}+\frac{1}{2}\delta_1$ is a mean-preserving spread of the degenerate lottery $q=\delta_0.$ Hence $p\geq_{\mathrm{cvx}}q$.
\end{example}

The affine core of $C_{\dd}$ is naturally related to the convex order when the outcome space has a linear structure. Suppose $X$ is a convex and compact subset of a normed space $V$, and we say that the decision maker's dissimilarity metric is induced by a norm $\lVert\cdot\rVert$, if $\dd:(x,y)\mapsto \lVert x-y \rVert$. For all $x\in X$, the function
\[
y\mapsto \lVert x-y\rVert
\]
is convex. Hence, if $p$ is larger than $q$ in the convex order, then $p$ has weakly larger expected distance from every outcome $x$.

\begin{lemma}\label{lem}
If $X$ is a convex and compact subset of a normed space $V$, and $\dd$ is induced by a continuous norm $\lVert\cdot\rVert$, then
\[
p\geq_{\mathrm{cvx}}q
\Longrightarrow
p\geq_{\dd} q
\]
for all $p,q\in\DX$.
\end{lemma}

Thus, whenever outcomes have a linear structure and subjective dissimilarity is norm-based, the affine core of $C_{\dd}$ is consistent with the standard convex-order comparison of dispersion. If $p$ is a mean-preserving spread of $q$, then $p$ is also farther than $q$, on average, from every possible benchmark outcome. In this setting, greater dispersion implies greater complexity in the affine-core sense.

For product spaces, the $L^1$-metric provides a transparent characterization. Let $X=\prod_{i\in [k]} I_i$ for some compact intervals $I_i\subseteq\mathbb R$, we say that $\dd$ is induced by the $L^1$-norm if
\[
\dd:(x,y)\mapsto \lVert x-y\rVert_1=\sum_{i\in [k]} |x_i-y_i|.
\]
For a lottery $p\in\DX$, denote by $p^i$ its $i$-th marginal. The previous lemma gives a one-way implication for norm-induced metrics. For the $L^1$-metric on product spaces, the affine core admits a sharper characterization. Because the $L^1$-distance is additively separable across coordinates, the order $\geq_{\dd}$ decomposes into one-dimensional convex-order comparisons of the marginals.

\begin{proposition}\label{prop:marginals_convexorder}
Suppose $X=\prod_{i\in [k]} I_i$ for some compact intervals $I_i\subseteq \mathbb R$, with $k\geq 1$, and suppose $\dd$ is induced by $\lVert\cdot\rVert_1$. Then, for all $p,q\in\DX$,
\[
p\geq_\dd q
\Longleftrightarrow
\forall i\in [k],\quad
p^i\geq_{\mathrm{cvx}}q^i.
\]
\end{proposition}
\noindent The proposition shows that, under the $L^1$-metric, lottery $p$ is more complex than $q$ in the affine-core sense if and only if each marginal distribution of $p$ is more dispersed than the corresponding marginal distribution of $q$ in the convex order. The reason is the additivity of the $L^1$-norm: expected distance from a benchmark $x$ separates into the sum of the expected coordinatewise distances from $x_i$,
\[
\int \lVert x-y\rVert_1\mathrm{d}p(y)=\sum_{i\in [k]}\int \lvert x_i-t\rvert \mathrm{d}p^i(t)
\]
for all $p\in \DX$ and $x\in X$.

This characterization also shows that the affine core of $C_{\dd}$ is generally weaker than the convex order. It only compares marginal dispersion and therefore ignores dependence across coordinates.

\begin{example}
Let $X=[-1,1]^2$, and let $\dd$ be induced by the $L^1$-norm. Consider the lotteries
\[
p=\frac{1}{2}\delta_{(1,1)}+\frac{1}{2}\delta_{(-1,-1)}
\]
and
\[
q=\frac{1}{2}\delta_{(1,-1)}+\frac{1}{2}\delta_{(-1,1)}.
\]
The two lotteries have the same marginals. Indeed, for each coordinate $i=1,2$,
\[
p^i=q^i=\frac{1}{2}\delta_{-1}+\frac{1}{2}\delta_1.
\]
Therefore, by Proposition \ref{prop:marginals_convexorder},
\[
p\geq_{\dd}q
\quad\text{and}\quad
q\geq_{\dd}p.
\]
However, $p$ and $q$ are not comparable in the multivariate convex order. To see this, consider the convex functions
\[
\phi:(x_1,x_2)\mapsto (x_1+x_2)^2
\quad\text{and}\quad
\psi:(x_1,x_2)\mapsto (x_1-x_2)^2.
\]
Then
\[
\int \phi\mathrm{d}p=4>0=\int \phi\mathrm{d}q,
\]
whereas
\[
\int \psi \mathrm{d}p=0<4=\int \psi\mathrm{d}q.
\]
Thus, neither lottery dominates the other in the convex order. The example shows that, under the $L^1$-metric, the affine core captures marginal dispersion but ignores dependence across coordinates.
\end{example}

In the one-dimensional case, the distinction disappears. The $L^1$-metric is simply the absolute-value metric, and the affine core coincides with the usual convex order.

\begin{corollary}\label{prop}
If $X$ is a compact interval in $\mathbb R$ and $\dd:(x,y)\mapsto |x-y|$, then, for all $p,q\in\DX$,
\[
p\geq_\dd q
\Longleftrightarrow
p\geq_{\mathrm{cvx}} q.
\]
\end{corollary}

\noindent
Thus, on the real line, the affine core of $C_{\dd}$ is exactly a pure dispersion order. In this case, the affine core of $C_\dd$ does not merely agree with the convex order; it reproduces exactly the standard mean-preserving-spread comparison.

\begin{remark}
The coordinatewise characterization in Proposition \ref{prop:marginals_convexorder} relies on the additivity of the $L^1$-norm. For $L^p$-metrics with $p>1$, expected distance from a benchmark generally depends on the joint distribution and cannot be reduced to marginal convex-order comparisons. The general implication from convex order to $\geq_{\dd}$ remains valid for all norm-induced metrics, but the converse and the marginal characterization are special to the $L^1$ case.
\end{remark}

\subsubsection{Maximal complexity with respect to the affine core}\label{sect:maximalorder}

The connection with convex order has a further implication. It allows one to use the geometry of the outcome space to identify where maximal complexity can arise. If complexity is strongly associated with dispersion, and dispersion can be increased by mean-preserving spreads, then maximal complexity should be associated with lotteries that put probability on the extreme points of the outcome set.

Assuming a linear structure on $X$ and that $\dd$ is induced by a norm, the relation between $\geq_\dd$ and the convex order suggests that one can increase the complexity of a lottery by moving probability mass away from interior outcomes and toward extreme outcomes, while preserving the barycenter of the lottery. Thus, maximally complex lotteries should place probability only on outcomes that cannot themselves be decomposed as mixtures of other outcomes.

The next result formalizes this idea. It shows that every lottery can be made weakly more complex, in the affine-core sense, by replacing its outcomes with lotteries supported on the extreme points of $X$.

\begin{proposition}\label{prop:extremepoints}
Suppose $X$ is a nonempty, compact, convex subset of a normed space, and suppose $\dd$ is induced by a continuous norm. Then, for every $p\in \DX$, there exists $p^*\in \DX$ such that
$$
p^*(\mathrm{ext}(X))=1\ \textnormal{and}\
p^*\geq_\dd p.
$$
\end{proposition}
\noindent The proposition says that any lottery can be ``complexified'' by spreading its mass toward the extreme points of the outcome space. Since $\dd$ is induced by a norm, convex order dominance implies dominance in the affine core. Hence, it is enough to replace each outcome by a lottery over extreme outcomes with the same barycenter. The resulting lottery is a mean-preserving spread of the original one and is therefore weakly more complex according to $\geq_{\dd}$.

This result also gives existence of maximally complex lotteries. Since $X$ is compact, $\DX$ is compact in the weak topology; moreover, the affine-core relation $\geq_{\dd}$ is continuous. Therefore, maximal elements of $(\DX,\geq_{\dd})$ exist.\footnote{We recall here the definition of maximal elements. Given a binary relation $\succsim$ on $\DX$, possibly incomplete, and a subset $A\subseteq \DX$, we denote by $\textnormal{MAX}(A,\succsim)$ $\succsim$-maximal elements in $A$ i.e.,
\begin{align*}
\textnormal{MAX}(A,\succsim)&=\left\lbrace p\in A:\not\exists q\in A\ \textnormal{s.t.}\ q \succ p\right\rbrace.
\end{align*} 
}

Proposition \ref{prop:extremepoints} implies that the search for such maximal elements can be restricted to lotteries supported on the extreme points of $X$. Also the converse holds: every maximally complex lottery must be concentrated on extreme outcomes. 

\begin{proposition}\label{prop:extremepointsuniqueness}
Suppose $X$ is a nonempty, compact, convex subset of a normed space, and suppose $\dd$ is induced by a continuous norm. If $p\in \mathrm{MAX}(\DX,\geq_\dd)$, then $p(\mathrm{ext}(X))=1.$
\end{proposition}
\noindent The idea of the proof is as follows. If $p$ assigns positive probability to a nonextreme outcome, then that outcome can be expressed as a nontrivial convex combination of two distinct outcomes in $X$. Replacing this outcome by the corresponding lottery preserves its barycenter and therefore generates a mean-preserving spread. Convexity of the norm implies that the dilation weakly increases expected distance from every benchmark. If the dilation is nontrivial on a set of positive probability, we show that the increase is strict for at least one benchmark (the details are in the Appendix).

Taken together, Propositions \ref{prop:extremepoints} and \ref{prop:extremepointsuniqueness} show that, under norm-based dissimilarity, maximal complexity is attained at the boundary of the outcome space.

\section{Choice under risk with complexity}\label{Section:CAEU}

We now use the complexity index $C_{\dd}$ to study choice under risk. The primitive object in this section is a preference relation $\succsim$ on $\DX$, interpreted as the decision maker's ranking of lotteries. The objective is to incorporate complexity concerns into an otherwise standard expected-utility criterion.

We say that preferences  $\succsim$ admit a \textit{Complexity-Adjusted Expected Utility representation} (CAEU) if they are represented by a function $V:\DX\to \R$ defined by
\[
V:p\mapsto\mathbb E_p[u]-\lambda C_{\dd}(p).
\]
for some continuous utility function on outcomes $u$, $\lambda\in \R$, and continuous metric $\dd$ on $X$. The first term is standard expected utility of the lottery. The second term adjusts expected utility by the perceived complexity of the lottery. The parameter $\lambda$ captures the decision maker's attitude toward complexity. If $\lambda>0$, complexity is costly: holding expected utility fixed, the decision maker prefers simpler lotteries. If $\lambda<0$, complexity is preferred. If $\lambda=0$, the representation reduces to expected utility.

Because $C_{\dd}$ is concave and expected utility is affine in probabilities, the curvature of $V$ depends on the sign of $\lambda$. When $\lambda\geq0$, the term $-\lambda C_{\dd}$ is convex, and hence $V$ is convex. In this case, the decision maker dislikes randomization over lotteries: if $p\sim q$, then mixtures of $p$ and $q$ are weakly worse than the original lotteries. By contrast, when $\lambda\leq0$, the complexity term is concave, and hence $V$ is concave; in that case, the decision maker exhibits a preference for randomization.\footnote{Here preference for randomization refers to the convexity of preferences: if $p\sim q$, then $\alpha p+(1-\alpha)q\succsim p$ for all $\alpha\in[0,1]$. Dislike for randomzation will be referred to as \textit{concavity of preferences}, defined analogously.}

When $\lambda\geq0$, complexity aversion can generate a stronger form of dislike for randomization, which we call \textit{preference for concentration}. This property is stated for convex outcome spaces. Suppose $X$ is a convex subset of a vector space. A preference relation $\succsim$ exhibits preference for concentration if, for all $x,y\in X$, $p,q\in\DX$, and $\alpha\in[0,1]$,
\[
\left(\delta_x\sim p\ \textnormal{and}\ \delta_y\sim q\right)
\Longrightarrow
\delta_{\alpha x+(1-\alpha)y}
\succsim
\alpha p+(1-\alpha)q.
\]
The property compares two ways of combining equally valued objects. On the right-hand side, the decision maker randomizes between the lotteries $p$ and $q$. On the left-hand side, each lottery is first replaced by an indifferent degenerate outcome, and these degenerate outcomes are then averaged into a single certain outcome. Preference for concentration requires the concentrated alternative to be weakly preferred.

\begin{proposition}\label{prop:prefforconcent}
Suppose that $X$ is a nonempty, compact, and convex subset of a vector space. If $\succsim$ admits a CAEU representation with $u$ concave, $\dd$ a continuous metric, and $\lambda\geq 0$, then $\succsim$ exhibits preference for concentration.
\end{proposition}

\noindent Preference for concentration goes beyond ordinary dislike for randomization. Indeed, concavity of preferences compares mixtures of indifferent lotteries with the lotteries themselves. Preference for concentration instead compares a mixture of lotteries with a degenerate lottery obtained by first replacing each component with an indifferent certainty equivalent and then averaging these certainty equivalents. Hence, the property captures the idea that the agent dislikes not only risk in outcomes, but also the additional complexity generated by keeping random alternatives unresolved.

\subsection{Monotonicity}\label{sect}

We now study when CAEU preferences are consistent with standard stochastic dominance requirements. A lottery that dominates another in the usual stochastic sense may also be more complex, and a sufficiently complexity-sensitive decision maker may fail to rank it better. The purpose of this subsection is to characterize exactly when such failures cannot occur.

Following \cite{Machina_EU_no_indep} and the more general approach of \cite{Stoch_dom_no_indep}, we study monotonicity through the local utilities of CAEU preferences. The relevant object is the affine core of the preference relation.

Given a binary relation $\succsim$ on $\DX$, its \textit{affine core}, $\succsim^*$, is defined by
$$
p\succsim^* q
\Longleftrightarrow
\forall \alpha\in(0,1],\ \forall \ell\in\DX,\quad
\alpha p+(1-\alpha)\ell
\succsim
\alpha q+(1-\alpha)\ell.
$$
Thus, $p\succsim^*q$ means that $p$ remains weakly preferred to $q$ after both lotteries are mixed with any common lottery. The affine core extracts the part of $\succsim$ that behaves linearly in mixtures. It is well known that, when $\succsim$ is a preorder, $\succsim^*$ is itself a preorder and is the largest affine subrelation of $\succsim$ \citep[Proposition 3.3]{cerreia2018rational}.\footnote{The affine core was introduced in the context of risky choice by \cite{cerreia_JMP}. In decision under ambiguity, related notions appear in \cite{Differentiating_ambiguity} and \cite{Objective_GMMS}.}

The next proposition characterizes the affine core of CAEU preferences. It shows that the relevant local utilities are obtained by subtracting from $u$ the distance from each possible benchmark outcome.

\begin{proposition}\label{prop:affinecoreCAEU}
Suppose $\succsim$ admits a CAEU representation with utility $u$, $\lambda\in \R$, and a continuous metric $\dd$ on $X$. Then, for all $p,q\in\DX$,
$$
p\succsim^* q
\Longleftrightarrow
\forall x\in X,\
\int \left[u(y)-\lambda\dd(x,y)\right]\mathrm{d}p(y)
\geq
\int \left[u(y)-\lambda\dd(x,y)\right]\mathrm{d}q(y).
$$
\end{proposition}
\noindent Following the terminology of \cite{Machina_EU_no_indep}, we call the functions
$$
\left\lbrace v_x:y\mapsto u(y)-\lambda\dd(x,y):x\in X\right\rbrace
$$
the \textit{local utilities} of CAEU preferences. Proposition \ref{prop:affinecoreCAEU} is useful because it reduces questions about stochastic dominance for CAEU preferences to questions about the family of local utilities. In particular, CAEU preferences are consistent with a given integral stochastic order whenever all their local utilities belong to the corresponding test-function class.

To state this point generally, recall that a binary relation $\geq_{\mathcal W}$ on $\DX$ is an \textit{integral stochastic order} if there exists a family $\mathcal W$ of continuous functions on $X$ such that
$$
p\geq_{\mathcal W}q
\Longleftrightarrow
\forall w\in\mathcal W,\quad
\int w\mathrm{d}p\geq\int w\mathrm{d}q.
$$
A preference relation $\succsim$ is \textit{consistent with} $\geq_{\mathcal W}$ if
$$
p\geq_{\mathcal W}q
\Longrightarrow
p\succsim q
$$
for all $p,q\in\DX$. Proposition \ref{prop:affinecoreCAEU} implies that CAEU preferences are consistent with $\geq_{\mathcal W}$ whenever each local utility $v_x$ belongs to the closed convex conic hull of $\mathcal W$, augmented with constant functions. The first and second stochastic dominance orders are obtained by choosing $\mathcal W$ to be the class of increasing functions, or the class of increasing concave functions, respectively (whenever concavity and monotonicity can be properly defined).

We first consider first-order stochastic dominance. Suppose $X$ is endowed with a partial order $\succeq$. First-order stochastic dominance is the integral stochastic order generated by increasing functions.\footnote{That is, functions $v:X\to \R$ such that $v(x)\geq v(y)$ whenever $x\succeq y$.}

\begin{corollary}\label{prop:FSDCAEU}
Suppose $\succsim$ admits a CAEU representation with utility $u$, $\lambda\in \R$, and a continuous metric $\dd$ and $X$ is endowed with a partial order $\succeq$. Then, the following are equivalent:
\begin{enumerate}
\item $\succsim$ is consistent with first-order stochastic dominance;
\item for all $x,y\in X$,
$$
x\succeq y
\Longrightarrow
u(x)-u(y)\geq |\lambda|\dd(x,y).
$$
\end{enumerate}
In particular, if $\succsim$ is consistent with first-order stochastic dominance, then $u$ is increasing.
\end{corollary}
\noindent This corollary shows that monotonicity of $u$ is not sufficient for CAEU preferences to respect first-order stochastic dominance. The utility gain from moving to a better outcome must be large enough to dominate the maximal complexity effect that such a move can generate. The strength of this requirement depends on $|\lambda|$. If the decision maker is highly sensitive to complexity, then even small changes in complexity can interfere with dominance comparisons. Conversely, if the decision maker is quite insensitive to complexity, then FSD-monotonicity is easier to satisfy. This provides a simple explanation for why complexity can generate violations of stochastic dominance. 

We next consider second-order stochastic dominance. Suppose $X$ is a compact convex subset of a vector space and is endowed with a partial order. Second-order stochastic dominance is the integral stochastic order generated by increasing and concave functions. Hence, by Proposition \ref{prop:affinecoreCAEU}, CAEU preferences are consistent with second-order stochastic dominance exactly when every local utility is increasing and concave.

\begin{corollary}\label{prop:SSDCAEU}
Suppose $\succsim$ admits a CAEU representation with utility $u$, $\lambda\in \R$, and a continuous metric $\dd$ and $X$ is a compact and convex subset of a normed space endowed with a partial order. Then, the following are equivalent:
\begin{enumerate}
\item $\succsim$ is consistent with second-order stochastic dominance;
\item for every $x\in X$, the function
$$
y\mapsto u(y)-\lambda\dd(x,y)
$$
is increasing and concave.
\end{enumerate}
Moreover, if $\succsim$ is consistent with second-order stochastic dominance, then $u$ is increasing. In addition, if $\dd$ is induced by a norm,
\begin{enumerate}
\item if $\lambda<0$, then $u$ is concave;
\item if $u$ is convex, then $\lambda\geq0$.
\end{enumerate}
\end{corollary}
\noindent The second-order condition highlights the interaction between risk attitude and complexity attitude. Suppose, in particular, that $\dd$ is induced by a norm. Then consistency with second-order stochastic dominance requires every local utility
$$
y\mapsto u(y)-\lambda\|x-y\|
$$
to be increasing and concave. If $\lambda>0$, the complexity penalty subtracts a convex distance function from $u$. This makes concavity easier to satisfy, but makes monotonicity more demanding. If $\lambda<0$, the agent values complexity; the local utilities then add a convex distance term to $u$, so second-order dominance can be respected only if $u$ is sufficiently concave to offset this additional convexity. Thus, under norm-induced dissimilarity, consistency with second-order stochastic dominance imposes joint restrictions on utility curvature and complexity attitude.

This reveals a useful distinction between primitive utility curvature and effective risk attitude. In expected utility, risk attitude is governed by the curvature of $u$. Under CAEU, by contrast, risk attitudes are determined by the curvature of the local utilities. Complexity attitude therefore becomes one component of the decision maker's effective attitude toward risk: complexity aversion can reinforce risk aversion, while complexity seeking attitudes can weaken it.

\subsection{Risk aversion}\label{sect}

Suppose that $X$ is a convex and compact subset of a normed space. The \textit{concave order} is the binary relation $\geq_{\mathrm{cve}}$ on $\DX$ defined by
\[
p\geq_{\mathrm{cve}}q
\Longleftrightarrow
\int \varphi\mathrm{d}p\geq \int \varphi\mathrm{d}q
\quad
\text{for every continuous concave function }
\varphi:X\to\mathbb R.
\]
A preference relation $\succsim$ satisfies \textit{strong risk aversion} if
\[
p\geq_{\mathrm{cve}}q
\Longrightarrow
p\succsim q
\]
for all $p,q\in\DX$. Thus, strongly risk-averse preferences are consistent with every comparison made by all concave utility functions. For CAEU preferences, strong risk aversion can be characterized through the local utilities.

\begin{corollary}\label{coro:CAEUstrongrisk}
Suppose $\succsim$ admits a CAEU representation with utility $u$, $\lambda\in \R$, and a continuous metric $\dd$ and $X$ is a convex and compact subset of a normed space. Then, the following are equivalent:
\begin{enumerate}
\item $\succsim$ exhibits strong risk aversion,
\item for every $x\in X$, the function
\[
y\mapsto u(y)-\lambda\dd(x,y)
\]
is concave.
\end{enumerate}
Moreover, if $\dd$ is induced by a norm:
\begin{enumerate}
\item if $\lambda<0$ and $\succsim$ exhibits strong risk aversion, then $u$ is concave,
\item if $u$ is convex and $\succsim$ exhibits strong risk aversion, then $\lambda\geq0$.
\end{enumerate}
\end{corollary}
\noindent Corollary \ref{coro:CAEUstrongrisk} shows that, under CAEU, risk aversion is not governed by the curvature of $u$ alone. As for stochastic dominance orders, the relevant objects are the local utilities. Suppose that $\dd$ is induced by a norm. If $\lambda>0$, the complexity penalty subtracts a convex distance function from $u$, thereby contributing to concavity. Complexity aversion can therefore reinforce risk aversion. If $\lambda<0$, the agent values complexity, and the local utilities add a convex distance term to $u$. In that case, strong risk aversion can hold only if the curvature of $u$ is sufficiently concave to offset the convexity introduced by complexity seeking. Hence, for CAEU preferences, complexity attitude is one component of the decision maker's effective attitude toward risk.

A weaker notion of risk aversion compares each lottery with its barycenter. Suppose $X$ is a compact convex subset of a finite-dimensional vector space. A preference relation $\succsim$ satisfies \textit{weak risk aversion} if
\[
\delta_{\mu(p)}\succsim p
\]
for every $p\in\DX$, where
\[
\mu(p)=\int y\mathrm{d}p(y)
\]
is the barycenter of $p$. Under CAEU, this condition becomes
\[
u(\mu(p))-\mathbb E_p[u]\geq -\lambda C_{\dd}(p)
\]
for every $p\in\DX$. Thus, weak risk aversion may arise either because $u$ is concave, because complexity is costly, or because the two effects jointly dominate the expected-utility gain from risk. In particular, if $u$ is concave and $\lambda\geq0$, then CAEU preferences satisfy weak risk aversion. It is easy to notice that if $u$ is affine and $\dd$ is induced by a norm, then weak risk aversion holds if and only if $\lambda\geq 0$.

\section{Final remarks and future research}

As a final \textit{coda} of the paper, we discuss some possible extensions and alternative approaches we did not develop in this paper.

\subsection{Extensions to different choice problems}

\textbf{Choice under uncertainty.} In choice under uncertainty a decision maker chooses among alternatives called: \textit{acts}. Given a finite set of states $\Omega$ and a set of outcomes $X$, an \textit{act} is a map $f:\Omega\to X$. The set of acts is denoted by $\mathbf{F}$. By analogy with the approach developed in this paper, one could consider as the \textit{simplest} acts, the degenerate ones, that is, the constant acts, namely those acts that deliver the same outcome in every state. Given a distance $\dd$ on $\mathbf{F}$, then it is possible to replicate the analysis provided in this paper using $\dd$ to assess the complexity of acts. The choice of the metric can be based on the primitives $\Omega$ and $X$ are endowed with. If for instance $X$ has a metric $d$ and there is a fully-supported\footnote{The full-support hypothesis is necessary if one wants to work with metrics, otherwise it can be dispensed accepting to work with semimetrics.} probability measure $\mu$ on $\Omega$, then one could endow $\mathbf{F}$ with the following distance:
\[
\dd_\mu:(f,g)\mapsto \int d(f(\omega),g(\omega))\mathrm{d}\mu(\omega).
\]
According to this distance, the induced complexity measure would declare that an act is complex when its state-contingent outcomes are far, on average, from the closest constant act.

Other possibilities could exploit the multiple priors of a decision maker under ambiguity. For instance, if such a decision maker has fully-supported multiple priors, whose (compact and convex) set is denoted by $M$, then 
\[
\dd_M:(f,g)\mapsto \sup\limits_{\mu\in M}\int d(f(\omega),g(\omega))\mathrm{d}\mu(\omega)
\]
could be interpreted as the worst-case scenario distance for the decision maker. The corresponding complexity measure would evaluate how far an act is from being constant under the most unfavorable prior in $M$. This illustrates how the same construction can incorporate ambiguity directly into the measurement of complexity.
\par\medskip
\noindent \textbf{Menu choice.} 
The same idea can be applied to menus. If the choice alternatives are finite menus $A\subseteq X$, the natural simplicity benchmarks are singleton menus. A simple dissimilarity function is obtained from symmetric difference:
$$
\dd^\triangle(A,B)=|A\triangle B|.
$$
For nonempty finite menus,
$$
\min_{x\in X}\dd^\triangle(A,\{x\})=|A|-1.
$$
This recovers the idea that menu complexity increases with the number of available alternatives, while assigning zero complexity to singleton menus. This echoes the idea of thinking aversion in \cite{Ortoleva2013}.
\par\medskip
If the structure of $X$ is endowed with a metric $\dd$, then one could provide a more sophisticated measure. Indeed, one may want to take under consideration how different are the alternatives within a menu. As a consequence a natural choice is the measure
$$
\min_{x\in X}
\sum_{y\in A}\dd(y,x)
$$
that evaluates how far the menu is from its closest singleton approximation. Unlike cardinality, it distinguishes menus with the same number of elements according to the dispersion of their alternatives. In this sense, the measure can be interpreted as a cost of contemplating a menu whose elements cannot be well summarized by a single alternative.\footnote{For compact metric spaces, analogous constructions can be developed on the space $\mathcal K(X)$ of nonempty compact subsets of $X$, using, for example, Hausdorff distance or Wasserstein distances between probability measures supported on menus.}
\par\medskip

\noindent \textbf{Intertemporal choice.} Analogous considerations could be provided in the complexity evaluation of consumption streams. In particular, constant streams of outcomes can be used as \textit{simplicity benchmarks}. Suppose $(X,\dd)$ is a compact metric space and the agent contemplates as alternatives intertemporal streams of outcomes 
\[
\mathbf{C}=\left\lbrace (x_t)_{t\geq 1}:\forall t\geq 1,\ x_t\in X \right\rbrace.
\]
A stream of outcome will be denoted by $\mathbf{x}$. The agent's level of patience may yield a certain discount rate $\beta\in (0,1)$. This discount factor allows to define the following distance:
\[
\dd^\beta:(\mathbf{x},\mathbf{y})\mapsto \sum_{t\geq 1}\beta^t\dd(x_t,y_t)
\]
that is well-defined since $X$ is compact. The associated measure of complexity is 
\[
C_{\dd^\beta}:\mathbf{x}\mapsto \min\limits_{x\in X}\sum_{t\geq 1}\beta^t\dd(x_t,x).
\]
This measure evaluates the distance of each stream of outcomes from its closest constant stream of outcomes. The generality of $X$ allows for many important discrete intertemporal decision problems. The continuous case can be discussed analogously employing exponential discounting (and Borel measurable streams):
\[
\dd^\gamma:(\mathbf{x},\mathbf{y})\mapsto \int_0^\infty e^{-\gamma t}\dd(x_t,y_t)\mathrm{d}t 
\]
where $\gamma\in (0,\infty)$.\footnote{In this case $\dd^\gamma$ may not be a metric, rather just a semimetric.}




\subsection{Beyond metric dissimilarity}

Representing perceived dissimilarity among outcomes by a metric is a substantive restriction. It is particularly important for our analysis as it yields
\[
\begin{aligned}
C_{\dd}\left(\frac{1}{2}\delta_x+\frac{1}{2}\delta_y\right)=
\frac{1}{2}\min_{z\in X}
\left\{\dd(z,x)+\dd(z,y)\right\}=
\frac{1}{2}\dd(x,y).
\end{aligned}
\]
Hence, each endpoint is a best degenerate proxy for the binary lottery, and the complexity of the lottery directly reveals the perceived distance between its outcomes. This link is central to both the identification of $\dd$ and the cardinal uniqueness of $C_{\dd}$.

Axiom \ref{mx9} provides the behavioral counterpart of this property. Starting from a 50-50 lottery, assigning additional probability to an outcome already in its support cannot be more complex than assigning that probability to an arbitrary third outcome. The axiom therefore rules out the possibility that a third outcome provides a strictly better way to simplify a 50-50 lottery than either of its endpoints. In this sense, the metric structure is not imposed directly on a primitive dissimilarity function; it is recovered from the decision maker's complexity comparisons.

Suppose now that $\dd$ is replaced by a continuous function $\Delta:X\times X\to [0,\infty)$ such that $\Delta(x,x)=0$ that need not satisfy the triangle inequality, and define
\[
C_{\Delta}(p)
=
\min_{z\in X}
\int_X \Delta(z,y)\mathrm{d}p(y).
\]
For a 50-50 lottery,
\[
C_{\Delta}\left(\frac{1}{2}\delta_x+\frac{1}{2}\delta_y\right)
=
\frac{1}{2}\min_{z\in X}
\left\{\Delta(z,x)+\Delta(z,y)\right\}.
\]
If the triangle inequality fails, an outcome outside the support may be a better common proxy for $x$ and $y$ than either endpoint. Binary-lottery complexity then reveals the dissimilarity of the lottery from the best common proxy, rather than the primitive dissimilarity $\Delta(x,y)$. The best-approximation interpretation therefore remains meaningful, but the underlying pairwise dissimilarities can no longer be recovered from the complexity ranking, and the identification and uniqueness results need not follow. This distinction separates the core benchmark-based idea of the paper from the stronger consequences obtained under metric dissimilarity. Nonetheless, non-metric dissimilarities may be appropriate when the salient features of outcomes vary across comparisons, as emphasized by \cite{Tversky1977} and \cite{TverskyGati1982}, but their recovery would require additional restrictions or richer data.

\subsubsection*{$\dd$ as difficulty to compare\footnote{I thank Pietro Ortoleva for suggesting this alternative approach.}}

Going beyond the metric approach, one possibility is to replace $\dd$ with a function measuring how difficult outcomes are to compare. While $\dd(x,y)$ represents the perceived distance between $x$ and $y$, complexity could instead depend on the difficulty of comparing them. The two notions need not coincide: distant outcomes may be easy to compare, whereas nearby outcomes may be difficult to rank. A possibility is to consider the following
\[
\tau_u(x,y)=\frac{\rho(x,y)}{\eta+|u(x)-u(y)|}.
\]
where $\eta>0$, $u$ and $\rho$ are respectively a (normalized and cardinal) utility function and a metric over outcomes. A large utility difference provides a decisive basis for comparison and therefore reduces $\tau_u(x,y)$, whereas, holding the utility difference fixed, a large value of $\rho(x,y)$ increases comparison difficulty. The parameter $\eta$ keeps the measure finite when the decision maker is indifferent between the outcomes. Although $\tau_u(x,x)=0$ and $\tau_u$ is symmetric, it need not satisfy the triangle inequality.

Given a generic (continuous) measure of comparison difficulty among outcomes $\tau$, the corresponding complexity measure $C_\tau$ would be
\[
C_\tau(p)=\min_{x\in X}\int_X\tau(x,y)\mathrm{d}p(y).
\]
In this case $C_\tau$ measures the average difficulty of comparing the outcomes of $p$ with the degenerate lottery that best facilitates comparison. This differs conceptually from metric dispersion: the best proxy is selected because it makes the outcomes easier to compare, rather than because it is the closest.

\bibliographystyle{abbrvnat}
\bibliography{biblio_cleaned}

\section{Appendix}
\addcontentsline{toc}{section}{Appendices}
\renewcommand{\thesubsection}{\Alph{subsection}}

\subsection{Additional properties}

\begin{lemma}\label{basic properties}
If $\dd$ is a continuous metric on $X$, then
\begin{enumerate}
\item $C_\dd$ is concave.
\item $C_\dd$ is 1-Lipschitz with respect to the 1-Wasserstein distance induced by $\dd$, i.e., for all $p,q\in \DX$,
\[
\lvert C_\dd(p)-C_\dd(q)|\leq W_1^\dd(p,q).
\]
\item $C_\dd$ is continuous with respect to the weak convergence topology.
\item $C_\dd(\delta_x)=0$ for all $x\in X$.
\item $C_\dd(p)>0$ for all nondegenerate lotteries $p\in \DX$.

\item For all $\alpha\in [0,1/2]$, $p\in \DX$, and $x\in X$,
\[
C_\dd(\alpha p+(1-\alpha)\delta_x)=\alpha \int \dd(x,y)\mathrm{d}p(y).
\]
\item For all $\alpha\in [0,1]$ and $x,y\in X$,
\[
C_\dd(\alpha \delta_x+(1-\alpha)\delta_y)=\min\lbrace\alpha,1-\alpha\rbrace\dd(x,y).
\]
\item For all $\alpha\in [0,1]$ and $p\in \DX$,
\[
x\in \arg\min_{x\in X}\int \dd(x,y)\mathrm{d}p(y)\Longrightarrow C_\dd(\alpha p+(1-\alpha)\delta_x)=\alpha C_\dd(p).
\]
\end{enumerate}
\end{lemma}
\begin{proof}
\begin{enumerate}
\item Let $p,q\in \DX$ and $\alpha\in [0,1]$. Since each $W_\dd^1(\cdot,\delta_x)$ is affine, it follows that
\begin{align*}
C_\dd(\alpha p+(1-\alpha)q)&=\min\limits_{x\in X}\left[\alpha W_\dd^1(p,\delta_x)+(1-\alpha)W_\dd^1(q,\delta_x)\right]\\
&\geq \alpha \min\limits_{x\in X}W^1_\dd(p,\delta_x)+(1-\alpha) \min\limits_{y\in X}W^1_\dd(q,\delta_y)\\
&=\alpha C_\dd(p)+(1-\alpha)C_\dd(q).
\end{align*}
\item Let $p,q\in \DX$ and $x^*_q\in \textnormal{argmin}_{x\in X}\int\dd(x,y)\mathrm{d}q(y)$. Then, it follows that
\begin{align*}
C_\dd(p)-C_\dd(q)&=\min\limits_{x\in X}\int\dd(x,y)\mathrm{d}p(y)-\min\limits_{x\in X}\int\dd(x,y)\mathrm{d}q(y)\\
&\leq \int\dd(x^*_q,y)\mathrm{d}p(y)-\int\dd(x^*_q,y)\mathrm{d}q(y)\\
&\leq \sup_{x\in X} \left\lbrace W^1_\dd(p,\delta_x)-W^1_\dd(q,\delta_x)\right\rbrace\leq W^1_\dd(p,q).
\end{align*}
Interchanging $p$ and $q$ we get that $C_\dd$ is 1-Lipschitz with respect to $W^1_\dd$.
\item Since $\dd$ is continuous with respect to $\rho$ and $C_\dd$ is 1-Lipschitz with respect to $W^1_\dd$ the claim follows.
\item Since $C_\dd\geq 0$, it follows that $0\leq C_\dd(\delta_x)\leq W_\dd^1(\delta_x,\delta_x)=0$ for all $x\in X$.
\item Suppose $p$ is not degenerate. Since $W^1_\dd$ is a metric, we have that $W^1_\dd(p,\delta_x)>0$ for all $x\in X$, and hence, by compactness of $X$ and continuity of $W^1_\dd(p,\cdot)$, it follows that $C_\dd(p)=\min_{x\in X}W^1_\dd(p,x)>0$.
\item Since $\alpha\leq 1/2$, we have that
\begin{align*}
C_\dd(\alpha p+(1-\alpha)\delta_y)&=\min\limits_{z\in X}\left\lbrace\alpha \left[W^1_\dd(p,\delta_z)+\dd(z,y)\right]+(1-2\alpha)\dd(z,y)\right\rbrace\\
&\geq \alpha W^1_\dd(p,\delta_y)+(1-2\alpha)\min\limits_{z\in X}\dd(z,y)\\
&\geq \alpha W^1_\dd(p,\delta_y)\\
&\geq \min\limits_{z\in X}\alpha W^1_\dd(p,\delta_z)+(1-\alpha)\dd(y,z)\\
&=C_{\dd}(\alpha p+(1-\alpha)\delta_y).
\end{align*}
Therefore, $
C_\dd(\alpha p+(1-\alpha)\delta_y)=\alpha\int \dd(x,y)\mathrm{d}p(x).$
\item Suppose $\alpha\in [0,1]$ and $x,y\in X$. Then, either $\alpha\geq 1/2$ or $1-\alpha\geq 1/2$. Thus, the claim follows by the previous point.
\item For all $p\in \DX$ and $\alpha\in [0,1]$, we have that
\begin{align*}
C_\dd(\alpha p+(1-\alpha)\delta_{x_p})&=\min\limits_{x\in X}\alpha W^1_\dd(p,\delta_{x})+(1-\alpha)\dd(x_p,x)\\
&\leq \alpha W^1_\dd(p,\delta_{x_p})+(1-\alpha)\dd(x_p,x_p)=\alpha C_\dd(p).
\end{align*}
By concavity of $C_\dd$, we have
\[
C_\dd(\alpha p+(1-\alpha)\delta_{x_p})\geq \alpha C_\dd(p)+(1-\alpha)C_\dd(\delta_{x_p})=\alpha C_\dd(p).
\]
The two inequalities yield the claim.
\end{enumerate}
\end{proof}

\subsection{Proof of Theorem \ref{axiomatic_characterization}}

Fix a function $C:\DX\to [0,\infty)$ and define the correspondence $A_C:\DX\rightrightarrows X$, 
\[
A_C(p)=\arg\min\limits_{x\in X}C\left(\frac{p+\delta_x}{2}\right).
\]
We prove Theorem \ref{axiomatic_characterization} in several steps.

\begin{axiomm}\label{ax1}
$C$ is continuous with respect to the weak convergence topology.
\end{axiomm}

\begin{axiomm}\label{ax2}
For all $x\in X$, $C(\delta_x)=0$.
\end{axiomm}

\begin{axiomm}\label{ax3}
For all nondegenerate $p\in \DX$, $C(p)>0$.
\end{axiomm}

\begin{axiomm}\label{ax4}
$C$ is concave.
\end{axiomm}

\begin{axiomm}\label{ax5}
If $p,q\in \DX$ and $x\in A_C(p)\cap A_C(q)$, then, for all $\alpha\in [0,1]$,
\[
x\in A_C(\alpha p+(1-\alpha)q)\ \textnormal{and}\ C(\alpha p+(1-\alpha)q)\leq \alpha C(p)+(1-\alpha)C(q).
\]
\end{axiomm}

\begin{axiomm}\label{ax6}
For all $x,y,z\in X$ and $\alpha\in [0,1]$, we have
\[
C\left(\alpha\frac{\delta_x+\delta_y}{2}+(1-\alpha)\delta_y\right)\leq C\left(\alpha\frac{\delta_x+\delta_y}{2}+(1-\alpha)\delta_z\right).
\]
\end{axiomm}

\begin{proposition}\label{thm1}
Let $(X,\rho)$ be a compact metric space and $C:\bigtriangleup(X)\to [0,\infty)$ a function. The following are equivalent:
\begin{enumerate}
\item $C$ satisfies axioms \ref{ax1}-\ref{ax6}.
\item $C$ admits a best approximation with metric $\dd$ continuous with respect to $\rho$.
\end{enumerate}
\end{proposition}
\begin{proof}
We start showing sufficiency. We divide the proof in several claims.
Define $\dd:X\times X\to [0,\infty)$ as
\[
\dd:(x,y)\mapsto 2C\left(\frac{\delta_x+\delta_y}{2}\right).
\]
First notice that axioms \ref{ax2} and \ref{ax3} yield that, for all $x,y\in X$,
\[
C\left(\frac{\delta_x+\delta_y}{2}\right)=0 \Longrightarrow x=y.
\]
\begin{claim}\label{claim1}
$\dd$ is a metric and it is continuous in the product topology over $X\times X$ induced by $\rho$.
\end{claim}
\begin{proof}[Proof of Claim \ref{claim1}]
It is trivial to see that $\dd(x,y)=\dd(y,x)$ for all $x,y\in X$. Moreover, by axiom \ref{ax2} and \ref{ax3}, we have that $\dd(x,y)=0$ if and only if $x=y$. Now we show that $\dd$ satisfies the triangle inequality. To this end, let $x,y,z\in X$. By axioms \ref{ax2}, \ref{ax3}, and \ref{ax6} we have that:
\begin{align}
&x\in A_C(\delta_x)\cap A_C\left(\frac{\delta_x+\delta_y}{2}\right)\cap A_C\left(\frac{\delta_x+\delta_z}{2}\right),\\
& y\in A_C(\delta_y)\cap A_C\left(\frac{\delta_x+\delta_y}{2}\right)\cap A_C\left(\frac{\delta_y+\delta_z}{2}\right),\\
& z\in A_C(\delta_z)\cap A_C\left(\frac{\delta_x+\delta_z}{2}\right)\cap A_C\left(\frac{\delta_y+\delta_z}{2}\right).
\end{align}
Therefore, it follows, by axioms \ref{ax2}, \ref{ax4}, and \ref{ax5}, that
\begin{equation}\label{equation1}
C\left(\frac{\delta_x+\frac{\delta_x+\delta_z}{2}}{2}\right)=\frac{1}{2}C\left(\frac{\delta_x+\delta_z}{2}\right)
\end{equation}
and hence, by axiom \ref{ax6}, we have that
\begin{equation}\label{equation2}
\frac{1}{2}C\left(\frac{\delta_x+\delta_z}{2}\right)=C\left(\frac{\delta_x+\frac{\delta_x+\delta_z}{2}}{2}\right)\leq C\left(\frac{\delta_y+\frac{\delta_x+\delta_z}{2}}{2}\right).
\end{equation}
Now by (3), \ref{ax4}, and \ref{ax5}, we have that:
\begin{align*}
C\left(\frac{\delta_y+\frac{\delta_x+\delta_z}{2}}{2}\right)&=C\left(\frac{\delta_x+2\delta_y+\delta_z}{4}\right)\\
&=C\left(\frac{1}{2}\frac{\delta_x+\delta_y}{2}+\frac{1}{2}\frac{\delta_y+\delta_z}{2}\right)\\
&=\frac{1}{2}C\left(\frac{\delta_x+\delta_y}{2}\right)+\frac{1}{2}C\left(\frac{\delta_y+\delta_z}{2}\right).
\end{align*}
To conclude, this last equation together with \eqref{equation2} yield the triangle inequality:
\[
\dd(x,z)=2C\left(\frac{\delta_x+\delta_z}{2}\right)\leq 2C\left(\frac{\delta_x+\delta_y}{2}\right)+2C\left(\frac{\delta_y+\delta_z}{2}\right)\leq \dd(x,y)+\dd(y,z).
\]
Thus, $\dd$ is a metric. Passing to continuity, if $\rho(x_n,x)\to 0$, then, by \ref{ax1}, $C((\delta_{x_n}+\delta_x)/2)\to C(\delta_x)=0$ and hence, $\dd(x_n,x)\to 0$. Now suppose that $(x_n,y_n)\to (x,y)$. Then, $\rho(x_n,x)\to 0$ and $\rho(y_n,y)\to 0$ and hence
\[
\dd(x_n,y_n)-\dd(x,y)\leq \dd(x_n,x)+\dd(x,y)+\dd(y_n,y)-\dd(x,y)\to 0
\]
and analogously,
\[
\dd(x,y)-\dd(x_n,y_n)\leq \dd(x,x_n)+\dd(x_n,y_n)+\dd(y_n,y)-\dd(x_n,y_n)\to 0
\]
proving that $\dd(x_n,y_n)\to \dd(x,y)$. Thus $\dd$ is continuous over $X\times X$.
\end{proof}

Now denote by $\bigtriangleup^0(X)$ the set of simple probability measures.
\begin{claim}\label{claim2}
If $p\in \bigtriangleup^0(X)$ and $x\in A_C(p)$, then
\[
C(p)=\int \dd(x,y)\mathrm{d}p(y).
\]
\end{claim}
\begin{proof}[Proof of Claim \ref{claim2}]
Let $p\in \bigtriangleup^0(X)$ and assume without loss of generality that 
\[
p=\sum_{i=1}^n\alpha_i\delta_{x_i}.
\]
for some $n\geq 1$, $(\alpha_i)\in \mathbb{R}^n_+$ with $\sum_{i=1}^n\alpha_i=1$, and $x_1,x_2,\ldots,x_n$ in $X$. If $x\in A_C(p)$, then by \ref{ax2} $x\in A_C(p)\cap A_C(\delta_x)$. By axioms \ref{ax4} and \ref{ax5}, it follows that:
\[
C\left(\frac{p+\delta_x}{2}\right)=\frac{1}{2}C(p).
\]
Moreover, 
\[
\frac{p+\delta_x}{2}=\sum_{i=1}^n\alpha_i \frac{\delta_{x_i}+\delta_x}{2}
\]
and, by axiom \ref{ax6}, for all $i\in [n]$, we have that $x\in A_C((\delta_{x_i}+\delta_{x})/2)$. By axioms \ref{ax4} and \ref{ax5}, this implies that:
\[
\frac{1}{2}C(p)=C\left(\frac{p+\delta_x}{2}\right)=C\left(\sum_{i=1}^n\alpha_i \frac{\delta_{x_i}+\delta_x}{2}\right)=\sum_{i=1}^n\alpha_i C\left(\frac{\delta_{x_i}+\delta_x}{2}\right)
\]
and hence,
\[
C(p)=\sum_{i=1}^n\alpha_i \dd(x,x_i)=\int \dd(x,y)\mathrm{d}p(y)
\]
that proves the claim.
\end{proof}

\begin{claim}\label{claim3}
For all $p\in \bigtriangleup^0(X)$,
\[
A_C(p)\subseteq \arg\min\limits_{x\in X}\int \dd(x,y)\mathrm{d}p(y).
\]
\end{claim}
\begin{proof}[Proof of Claim \ref{claim3}]
Let $p\in \bigtriangleup^0(X)$ and assume without loss of generality that $p=\sum_{i=1}^n\alpha_i \delta_{x_i}$. Let $x\in A_C(p)$, by definition, we have that:
\[
C\left(\frac{p+\delta_x}{2}\right)\leq C\left(\frac{p+\delta_y}{2}\right) 
\]
for all $y\in X$. Therefore, by the same steps used in the proof of Claim \ref{claim2} and since, for all $i\in [n]$, $y\in A_C((\delta_{x_i}+\delta_y)/2)$ we have that 
\[
C\left(\frac{p+\delta_x}{2}\right)=\frac{1}{2}C(p)=\frac{1}{2}\sum_{i=1}^n\alpha_i\dd(x_i,x)\ \textnormal{and}\ C\left(\frac{p+\delta_y}{2}\right)=\frac{1}{2}\sum_{i=1}^n\alpha_i \dd(x_i,y)
\]
and hence,
\[
\int \dd(x,z)\mathrm{d}p(z)=\sum_{i=1}^n\alpha_i\dd(x_i,x)=C(p)\leq \sum_{i=1}^n\alpha_i \dd(x_i,y)=\int \dd(y,z)\mathrm{d}p(z).
\]
By the arbitrariness of $y\in X$, it follows that $x\in \arg\min\limits_{x\in X}\int \dd(x,y)\mathrm{d}p(y)$. Proving the claim. 
\end{proof}
Claims \ref{claim2} and \ref{claim3} yield:
\[
C(p)=\min\limits_{x\in X}\int \dd(x,y)\mathrm{d}p(y)
\]
for all $p\in \bigtriangleup^0(X)$. Indeed, let $p\in \bigtriangleup^0(X)$. Let $x\in A_C(p)$, connecting our previous two claims, we have that:
\[
C(p)=\int \dd(x,y)\mathrm{d}p(y)\leq \int \dd(z,y)\mathrm{d}p(y)
\]
for all $z\in X$. To conclude, we now need to extend these results to all probability measures over $X$. 
\begin{claim}\label{claim4}
For all $p\in \DX$,
\[
C(p)=\min\limits_{z\in X}\int \dd(z,y)\mathrm{d}p(y).
\]
\end{claim}
\begin{proof}[Proof of Claim \ref{claim4}]
Suppose $p\in \DX$. Then, by Theorem 15.10 in \cite{AliBorder2006}, there exists a sequence $(p_n)$ in $\bigtriangleup^0(X)$ such that $(p_n)$ weakly converges to $p$. For all $n\geq 1$, let
\[
x_n\in \arg\min_{x\in X}\int \dd(x,y)\mathrm{d}p_n(y).
\]
By compactness, $(x_n)$ admits a subsequence $(x_{n_k})$ converging to some $x\in X$. By \ref{ax1} we have that $C(p_{n_k})\to C(p)$. Moreover,
\begin{align*}
&\left\lvert \int \dd(x_{n_k},y)\mathrm{d}p_{n_k}(y)-\int \dd(x,y)\mathrm{d}p(y) \right\rvert\leq \\
&\left\lvert \int \dd(x_{n_k},y)\mathrm{d}p_{n_k}(y)-\int \dd(x,y)\mathrm{d}p_{n_k}(y) \right\rvert+\left\lvert \int \dd(x,y)\mathrm{d}p_{n_k}(y)-\int \dd(x,y)\mathrm{d}p(y) \right\rvert
\end{align*}
and hence, since $(p_{n_k})$ weakly converges to $p$ and, by Claim \ref{claim1}, $\dd$ is continuous we have that:
\[
C(p_{n_k})=\int \dd(x_{n_k},y)\mathrm{d}p_{n_k}(y)\to \int \dd(x,y)\mathrm{d}p(y)
\]
which implies $C(p)=\int \dd(x,y)\mathrm{d}p(y)$. Now suppose that $z\in X$, then,
\[
\int \dd(y,z)\mathrm{d}p_n(y)\to \int \dd(y,z)\mathrm{d}p(y)
\]
and hence
\[
\int \dd(y,z)\mathrm{d}p(y)=\lim\limits_{n\to \infty}\int \dd(y,z)\mathrm{d}p_n(y)\geq \lim\limits_{n\to \infty}C(p_n)=C(p).
\]
Thus, $x\in \arg\min\limits_{z\in X}\int \dd(z,y)\mathrm{d}p(y)$ and $C(p)=\min\limits_{z\in X}\int \dd(z,y)\mathrm{d}p(y)$ proving the claim.
\end{proof}
\noindent Thus, the proof of the sufficiency is completed.
\par\medskip
Now we show necessity. Since $\dd$ is continuous with respect to the product topology induced by $\rho$, it follows that $C$ is continuous in the topology of weak convergence.\footnote{The formal argument is completely analogous to that of Claim \ref{claim4} and therefore it is omitted.} It is trivial to see that $C(\delta_x)=0$ for all $x\in X$. Thus axioms \ref{ax1} and \ref{ax2} hold. Let $p\in \DX$ be nondegenerate. For some $x_p\in X$, we have $C(p)=\int \dd(x_p,y)\mathrm{d}p(y)=W^1_\dd(\delta_{x_p},p)$. Since $\delta_{x_p}$ is degenerate and $W^1_\dd$ is a metric, we have $W^1_\dd(\delta_{x_p},p)>0$ and hence Axiom \ref{ax3} holds. To prove concavity, let $\alpha\in [0,1]$ and $p,q\in \DX$. Then, it follows that:
\begin{align*}
C(\alpha p+(1-\alpha)q)&=\min\limits_{x\in X}\alpha \int \dd(x,y)\mathrm{d}p(y)+(1-\alpha)\int \dd(x,y)\mathrm{d}q(y)\\
&\geq \alpha \min\limits_{x\in X} \int \dd(x,y)\mathrm{d}p(y)+(1-\alpha)\min\limits_{x\in X} \int \dd(x,y)\mathrm{d}q(y)\\
&=\alpha C(p)+(1-\alpha)C(q).
\end{align*}
Thus, \ref{ax4} holds. Passing to \ref{ax5}, suppose that $x\in A_C(p)\cap A_C(q)$ and $\alpha\in [0,1]$. First notice that:
\[
\frac{1}{2}C(p)=C\left(\frac{p+\delta_x}{2}\right)
\]
and also
\[\frac{1}{2}\int \dd(x,y)\mathrm{d}p(y)=C\left(\frac{p+\delta_x}{2}\right)\leq C\left(\frac{p+\delta_z}{2}\right)=\frac{1}{2}\int \dd(z,y)\mathrm{d}p(y)
\]
for all $z\in X$ and hence $C(p)=\int \dd(x,y)\mathrm{d}p(y)$. The same holds for $q$. Since $C$ is concave, we have that:
\begin{align*}
\alpha C(p)+(1-\alpha)C(q)
&=\alpha \int \dd(x,y)\mathrm{d}p(y)+(1-\alpha)\int \dd(x,y)\mathrm{d}q(y)\\
&=\int \dd(x,y)\mathrm{d}\left(\alpha p+(1-\alpha)q\right)(y)\geq C(\alpha p+(1-\alpha)q)
\end{align*}
thus \ref{ax5} holds. To conclude, let $x,y,z\in X$ and $\alpha\in [0,1]$, then for all $w \in X$,
\begin{align*}
\int \dd(v,w)\mathrm{d}\left(\alpha\frac{\delta_x+\delta_y}{2}+(1-\alpha)\delta_z\right)(v)&\geq \frac{\alpha}{2}\left[\dd(x,w)+\dd(y,w)\right]\\
&\geq \frac{\alpha}{2}\dd(x,y)\\
&=C_\dd\left(\alpha\frac{\delta_x+\delta_y}{2}+(1-\alpha)\delta_y\right)
\end{align*}
where the last inequality follows observing that $y\in A_{C_\dd}(\alpha(\delta_x+\delta_y)/2+(1-\alpha)\delta_y)$. Thus, by arbitrariness of $w\in X$,
\begin{align*}
C_\dd\left(\alpha\frac{\delta_x+\delta_y}{2}+(1-\alpha)\delta_z\right)\geq C_\dd\left(\alpha\frac{\delta_x+\delta_y}{2}+(1-\alpha)\delta_y\right).
\end{align*}
Thus, \ref{ax6} holds.
\end{proof}

\noindent For the second step, we sharpen the representation. Next are the characterizing properties. 

\begin{axiomq}\label{qx1}
$C$ is continuous with respect to the weak convergence topology.
\end{axiomq}

\begin{axiomq}\label{qx2}
For all $x\in X$, $C(\delta_x)=0$.
\end{axiomq}

\begin{axiomq}\label{qx3}
For all nondegenerate $p\in \DX$, $C(p)>0$.
\end{axiomq}

\begin{axiomq}\label{qx4}
$C$ is quasiconcave.
\end{axiomq}

\begin{axiomq}\label{qx7}
If $p,q\in \DX$ and $x\in A_C(p)\cap A_C(q)$, then, for all $\alpha\in [0,1]$,
\[
x\in A_C(\alpha p+(1-\alpha)q).
\]
\end{axiomq}

\begin{axiomq}\label{qx5}
If $p,q\in \DX$ and $x\in A_C(p)\cap A_C(q)$, then, for all $\alpha\in [0,1]$,
\[
C(\alpha p+(1-\alpha)q)\leq \max\left\lbrace C(p),C(q)\right\rbrace.
\]
\end{axiomq}

\begin{axiomq}\label{qx6}
If $p\in \DX$ and $x\in A_C(p)$, then, for all $\alpha\in [0,1]$,
\[
C(\alpha p+(1-\alpha)\delta_x)=\alpha C(p).
\]
\end{axiomq}

\begin{axiomq}\label{qx8}
For all $x,y,z\in X$, we have
\[
C\left(\frac{1}{2}\frac{\delta_x+\delta_y}{2}+\frac{\delta_y}{2}\right)\leq C\left(\frac{1}{2}\frac{\delta_x+\delta_y}{2}+\frac{\delta_z}{2}\right).
\]
\end{axiomq}

\begin{proposition}\label{quasi}
Let $C:\DX\to \R_+$. The following are equivalent:
\begin{enumerate}
\item $C$ satisfies axioms \ref{qx1}-\ref{qx8}.
\item $C$ is a best approximation function for some continuous metric $\dd$.
\end{enumerate}
\end{proposition}
\begin{proof}
Necessity of axioms \ref{qx1}-\ref{qx5} and \ref{qx8} is trivial in light of Proposition \ref{thm1}. Now we show that axiom \ref{qx6} is also necessary. Let $p\in \DX$, $x\in A_C(p)$, and $\alpha\in [0,1]$, by Proposition \ref{thm1} $C$ must satisfy \ref{ax4} and \ref{ax5},
\begin{align*}
C(\alpha p+(1-\alpha)\delta_x)=\alpha C(p)+(1-\alpha)C(\delta_x)=\alpha C(p)
\end{align*}

Passing to suffuciency, by inspection of the proof of Proposition \ref{thm1}, it is enough to show that \ref{qx1}-\ref{qx8} imply that for all $p,q\in \DX$ and $\alpha\in [0,1]$,
\begin{equation}\label{eqquasiquasi}
x\in A_C(p)\cap A_C(q)\Longrightarrow C(\alpha p+(1-\alpha)q)=\alpha C(p)+(1-\alpha)C(q).
\end{equation}
To this end let $p,q\in \DX$, $\alpha\in [0,1]$, and $x\in A_C(p)\cap A_C(q)$. If $C(p)=C(q)$, then there is nothing to show by axioms \ref{qx4} and \ref{qx5}. If $\alpha=0$ or $\alpha=1$, again the conclusion is immediate, and hence we assume $\alpha\in (0,1)$. Let $t=C(p)$ and $s=C(q)$ and assume, without loss of generality, that $t>s$ and let $\beta=s/t\in [0,1)$. Suppose first that $s=0$. Then, by axioms \ref{qx2}, \ref{qx3}, and the fact that $x\in A_C(q)$, we must have that $q=\delta_x$. Thus, by axiom \ref{qx6}, we have
\[
C(\alpha p+(1-\alpha)q)=C(\alpha p+(1-\alpha)\delta_x)=\alpha C(p)=\alpha C(p)+(1-\alpha )C(q)
\]
proving the claim. Therefore, we can also assume that $s>0$. Now let $p_\beta=\beta p+(1-\beta)\delta_x$. By axiom \ref{qx6}, we have $C(p_\beta)=\beta C(p)=s$. Therefore, $C(p_\beta)=C(q)$. Moreover, $x\in A_C(p_\beta)\cap A_C(q)$, and hence, by axioms \ref{qx4} and \ref{qx5}, it follows that:
\[
C(\theta p_\beta+(1-\theta)q)=s
\]
for all $\theta\in [0,1]$. In particular, now let $m=\alpha C(p)+(1-\alpha)C(q)$ and $\theta=\alpha t/m$. Then, we have that:
\begin{align*}
\theta p_\beta+(1-\theta)q&=\theta \beta p+\theta(1-\beta)\delta_x+(1-\theta)q\\
&=\frac{\alpha t}{m}\frac{s}{t}p+\frac{\alpha t}{m}\frac{t-s}{t}\delta_x+\frac{m-\alpha t}{m}q\\
&=\frac{\alpha s}{m}p+\frac{\alpha (t-s)}{m}\delta_x+\frac{(1-\alpha)s}{m}q\\
&=\frac{s}{m}(\alpha p+(1-\alpha)q)+\alpha\frac{t-s}{m}\delta_x\\
&=\frac{s}{m}(\alpha p+(1-\alpha)q)+\left(1-\frac{s}{m}\right)\delta_x.
\end{align*}
By axiom \ref{qx7}, we have that $x\in A_C(\alpha p+(1-\alpha)q)$. Then, by the last equality and axiom \ref{qx6}, we have that:
\begin{align*}
C(q)&=s=C(\theta p_\beta+(1-\theta)q)\\
&=C\left(\frac{s}{m}(\alpha p+(1-\alpha)q)+\left(1-\frac{s}{m}\right)\delta_x\right)\\
&=\frac{s}{m}C(\alpha p+(1-\alpha)q).
\end{align*}
Thus, since $s>0$, it follows that $m=C(\alpha p+(1-\alpha)q)$ and hence \eqref{eqquasiquasi} holds.
\end{proof}

Before finally giving the proof of Theorem \ref{thm1}, we need the following lemma.\footnote{The result is folklore, but I could not find any reference.}
\begin{lemma}\label{lemma_qc}
Let $S$ be a convex subset of a real vector space, and let
$f:S\to \mathbb{R}$ be continuous. Suppose that for all $x,y\in S$
and all $\alpha\in[0,1]$,
\[f(x)=f(y)
    \Longrightarrow
    f(\alpha x+(1-\alpha)y)\leq f(x).\]
Then, for all $x,y\in S$ and all $\alpha\in[0,1]$, $f(\alpha x+(1-\alpha)y)\leq\max\{f(x),f(y)\}.$
\end{lemma}

\begin{proof}
Fix $x,y\in S$. Since $S$ is convex, the line segment joining
$x$ and $y$ is contained in $S$. Define $g:[0,1]\to \mathbb{R}$ as
\[
g:t\mapsto f((1-t)x+ty).
\]
By continuity of $f$, the function $g$ is continuous on $[0,1]$. Let $M= \max\{g(0),g(1)\}=\max\{f(x),f(y)\}.$ We claim that $g(t)\leq M$ for every $t\in[0,1]$. Suppose, for contradiction, that there exists $t_0\in[0,1]$ such that $g(t_0)>M.$ Since $g(0)\leq M$ and $g(1)\leq M$, we must have $t_0\in(0,1)$.
By continuity of $g$, in particular by the intermediate value theorem applied on $[0,t_0]$ and $[t_0,1]$, there exist $a,b\in[0,1]$ with
\[
    a<t_0<b
    \quad\text{and}\quad
    g(a)=g(b)=M.
\]
Set
\[u := (1-a)x+ay, \qquad v := (1-b)x+by.
\]
Then $u,v\in S$, and
\[
    f(u)=g(a)=M=g(b)=f(v).
\]
Since $t_0\in(a,b)$, there exists $\lambda\in[0,1]$ such that
\[
    (1-t_0)x+t_0y
    =
    \lambda u+(1-\lambda)v.
\]
By the hypothesis applied to $u$ and $v$, we obtain
\[
    f\bigl(\lambda u+(1-\lambda)v\bigr)
    \leq f(u)=M.
\]
Therefore,
\[
    g(t_0)
    =
    f((1-t_0)x+t_0y)
    =
    f\bigl(\lambda u+(1-\lambda)v\bigr)
    \leq M,
\]
which contradicts $g(t_0)>M$. Hence $g(t)\leq M$ for all
$t\in[0,1]$. Finally, for any $\alpha\in[0,1]$, taking $t=1-\alpha$ gives
\[
    f(\alpha x+(1-\alpha)y)
    =
    g(1-\alpha)
    \leq
    M
    =
    \max\{f(x),f(y)\}.
\]
This proves the claim.
\end{proof}

\begin{lemma}\label{redundancymx8}
Suppose $\succsim_{\mathbf{c}}$ is a binary relation that satisfies axioms \ref{mx1}, \ref{mx3}, \ref{mx5}, \ref{mx6}, and \ref{mx7}. For all nondegenerate $p\in \DX$, $x\in A_{\succi}(p)$, and $\alpha,\beta\in (0,1)$,
\[
\alpha>\beta \Longrightarrow \alpha p+(1-\alpha)\delta_x\succ_{\mathbf{c}}\beta p+(1-\beta)\delta_x.
\]
\end{lemma}
\begin{proof}
Let $p\in \DX$ be nondegenerate and $x\in  A_{\succi}(p)$. Define $p_t=t p+(1-t)\delta_x$ for all $t\in [0,1]$ and let $\alpha,\beta\in (0,1)$ with $\alpha>\beta$. Then,
\begin{align*}
p_\beta& =\beta p+(1-\beta)\delta_x=\frac{\beta \alpha}{\alpha}p+(1-\beta)\delta_x\\
&=\frac{\beta}{\alpha}(\alpha p+(1-\alpha)\delta_x-(1-\alpha)\delta_x)+(1-\beta)\delta_x\\
&=\frac{\beta}{\alpha}p_\alpha +\left(1-\beta - \frac{\beta(1-\alpha)}{\alpha}\right)\delta_x=\frac{\beta}{\alpha}p_\alpha+\left(1-\frac{\beta}{\alpha}\right)\delta_x.
\end{align*}
By axiom \ref{mx5}, we have $x\in A_{\succi}(p_t)$ for all $t\in [0,1]$. Therefore, $x\in A_{\succi}(p_\alpha)\cap A_{\succi}(p_\beta)\cap A_{\succi} (\delta_x)$ and hence, by axiom \ref{mx6}, we have $p_\alpha\succi p_\beta$. Suppose now that $p_\alpha\cci p_\beta$. Then, by axiom \ref{mx7}, we have
\[
\forall t\in [0,1],\ t p_\alpha +(1-t)\delta_x\cci t p_\beta +(1-t)\delta_x.
\]
Now let $\theta=\beta/\alpha$. Then, we have
\begin{align*}
p_\alpha \cci p_\beta=p_{\theta\alpha} \cci p_{\theta\beta} = p_{\theta^2\alpha}
\end{align*}
and hence, by repeating the argument, $p_\alpha\cci p_{\theta^k\alpha}$ for all $k\geq 1$. Letting $k\to \infty$, by axiom \ref{mx1}, we get $p_\alpha \cci \delta_x$ that contradicts axiom \ref{mx3} since $p_\alpha$ is nondegenerate. Thus, $p_\alpha \succ_{\mathbf{c}}p_\beta$.
\end{proof}

\begin{proof}[Proof of Theorem \ref{axiomatic_characterization}]
By axiom \ref{mx1} Debreu representation Theorem yields that there exists $D:\DX\to \R$ such that $D$ represents $\succi$. By Axiom \ref{mx3}, we have that $D(p)\geq D(\delta_x)$ for all $p\in \DX$ and $x\in X$. As a consequence, we have that $D(\delta_x)=D(\delta_y)$ for all $x,y\in X$ and without loss of generality we can set $D(\delta_x)=0$ for all $x\in X$. Moreover, by compactness of $\DX$ and the fact that $X$ is not a singleton, $D$ attains a maximum $M>0$ at some $p^*\in \DX$. By axiom \ref{mx4} we have that $D$ is quasiconcave.\footnote{See Lemma 56 in \cite{cerreia2011uncertainty}} Clearly, we have that $A_D(p):=\arg\min_{x\in X}D((p+\delta_x)/2)=A_{\succi}(p)$ for all $p\in \DX$. While by axiom \ref{mx5}, we have that for all $p,q\in \DX$ and $\alpha\in [0,1]$ if $x\in A_{D}(p)\cap A_D(q)$, then
\[
x\in A_D(\alpha p+(1-\alpha)q).
\]
In addition, by axiom \ref{mx6}, we have that if $x\in A_{D}(p)\cap A_{D}(q)$ and $p\cci q$, then $D(\alpha p+(1-\alpha)q)\leq D(p)$. Consider the following set:
\[
A(x)=\left\lbrace p\in \DX:x\in A_D(p) \right\rbrace.
\]
If $p,q\in A(x)$, then $x\in A_D(p)\cap A_D(q)$ and hence $x\in A_D(\alpha p+(1-\alpha)q)$ that yields $\alpha p+(1-\alpha)q\in A(x)$ for all $\alpha\in [0,1]$. Therefore, $A(x)$ is convex. Moreover, let $D'=D|_{A(x)}$. We have that if $D'(p)=D'(q)$, then $x\in A_D(p)\cap A_D(q)$ and $D(p)=D(q)$ implying in turn that $D'(\alpha p+(1-\alpha)q)\leq D'(p)$. By Lemma \ref{lemma_qc}, it follows that $D'$ is quasiconvex. Thus, for all $p,q\in \DX$ with $x\in A_D(p)\cap A_D(q)$ we have
\[
D(\alpha p+(1-\alpha)q)\leq \max\left\lbrace D(p),D(q)\right\rbrace
\]
for all $\alpha\in [0,1]$.
\par\medskip
Moreover, by axiom \ref{mx9}, we have that for all $x,y,z\in X$,
\[
D\left(\frac{1}{2}\frac{\delta_x+\delta_y}{2}+\frac{1}{2}\delta_z\right)\geq D\left(\frac{1}{2}\frac{\delta_x+\delta_y}{2}+\frac{1}{2}\delta_y\right).
\]
Now define the following. Fix $x^*\in A_D(p^*)$. For all $\alpha\in [0,1]$, $r_\alpha=\alpha p^*+(1-\alpha)\delta_{x^*}$, and the function:
\begin{align*}
\varphi&:[0,1]\to [0,M]\\
&\alpha \mapsto D(r_\alpha)
\end{align*}
By continuity of $D$, $\varphi$ is continuous and also $\varphi(0)=0$ and $\varphi(1)=M$. Moreover, by Lemma \ref{redundancymx8}, $\varphi$ is strictly increasing. Letting $\psi:=\varphi^{-1}$, we define $C=\psi\circ D$. Since $\psi$ is strictly increasing and continuous, we have that $C$ is continuous, $A_C=A_D$, and satisfies all the properties aforementioned about $D$. In addition, $C$ satisfies proportionality (that is, axiom \ref{qx6} in Proposition \ref{quasi}). To prove this let $p\in \DX$, $x\in A_C(p)$, and $\alpha\in [0,1]$. We aim to show that:
\[
C(\alpha p+(1-\alpha)\delta_x)=\alpha C(p).
\]
First let:
\[
\beta=\psi(D(p))=C(p).
\]
Then, $D(p)=\varphi(\beta)=D(r_\beta)$. Thus, $p\cci r_\beta$. Since, $x^*\in A_D(p^*)\cap A_D(\delta_{x^*})$, it follows that $x^*\in A_D(r_\beta)$. By axiom \ref{mx7}, it follows that:
\begin{align*}
D\left(\alpha p+(1-\alpha)\delta_x\right)&=D\left(\alpha r_\beta+(1-\alpha)\delta_{x^*}\right)\\
&=D\left(\alpha\beta p^*+(1-\alpha\beta)\delta_{x^*}\right)=D(r_{\alpha\beta})
\end{align*}
and hence,
\[
C\left(\alpha p+(1-\alpha)\delta_x\right)=\psi(D(r_{\alpha\beta}))=\alpha\beta=\alpha C(p).
\]
Therefore, $C$ satisfies all the axioms of Proposition \ref{quasi} and hence the proof of sufficiency is completed. 

Let us pass to the necessity. By the representation we have that $C_\dd$ satisfies axioms \ref{qx1}-\ref{qx8} reported above Proposition \ref{quasi}. In particular, by \ref{qx1}, we have that $\succsim_\mathbf{c}$ induced by $C_\dd$ satisfies axiom \ref{mx1}. By axioms \ref{qx2} and \ref{qx3}, it follows that $\succsim_{\mathbf{c}}$ satisfies axiom \ref{mx3}. By quasiconcavity, axiom \ref{qx4}, we have that $\succsim_{\mathbf{c}}$ satisfies axiom \ref{mx4}. By axioms \ref{qx7} and \ref{qx5}, it follows that $\succsim_{\mathbf{c}}$ satisfies satisfies axioms \ref{mx5} and \ref{mx6}. By axiom \ref{qx6}, we have that $\succsim_{\mathbf{c}}$ satisfies axiom \ref{mx7}. To show that $\succsim_{\mathbf{c}}$ satisfies axiom \ref{mx9}, suppose that $x,y,z\in X$ and $\alpha\in [0,1]$, then for all $w \in X$,
\begin{align*}
\int \dd(v,w)\mathrm{d}\left(\alpha\frac{\delta_x+\delta_y}{2}+(1-\alpha)\delta_z\right)(v)&\geq \frac{\alpha}{2}\left[\dd(x,w)+\dd(y,w)\right]\\
&\geq \frac{\alpha}{2}\dd(x,y)\\
&=C_\dd\left(\alpha\frac{\delta_x+\delta_y}{2}+(1-\alpha)\delta_y\right)
\end{align*}
where the last inequality follows observing that $y\in A_{C_\dd}(\alpha(\delta_x+\delta_y)/2+(1-\alpha)\delta_y)$. Thus, by arbitrariness of $w\in X$,
\begin{align*}
C_\dd\left(\alpha\frac{\delta_x+\delta_y}{2}+(1-\alpha)\delta_z\right)\geq C_\dd\left(\alpha\frac{\delta_x+\delta_y}{2}+(1-\alpha)\delta_y\right)
\end{align*}
proving that $\succsim_{\mathbf{c}}$ satisfies axiom \ref{mx9}.
\end{proof}

\subsection{Proofs of Section \ref{sect:measure_properties} and Proposition \ref{prop:measurecardinaluniqueness}}

\begin{proof}[Proof of Proposition \ref{prop:measurecardinaluniqueness}]
If $\dd_1=\kappa \dd_2$, then $C_{\dd_1}=\kappa C_{\dd_2}$ and the claim follows. Conversely, suppose that $C_{\dd_1}$ and $C_{\dd_2}$ are ordinally equivalent. Then,
\[
C_{\dd_1}(p)\geq C_{\dd_1}(q)\Longleftrightarrow C_{\dd_2}(p)\geq C_{\dd_2}(q)
\]
for all $p,q\in \DX$. Then, there exists a strictly increasing function $\phi:\R\to \R$ such that $C_{\dd_1}=\phi\circ C_{\dd_2}$. By Lemma \ref{basic properties}, for all $x,y\in X$ and $\alpha\in [0,1/2]$, we have
\begin{equation}\label{eq:proportionality}
\alpha\dd_1(x,y)=C_{\dd_1}(\alpha \delta_y+(1-\alpha)\delta_x)=\phi\left(C_{\dd_2}(\alpha \delta_y+(1-\alpha)\delta_x)\right)=\phi(\alpha \dd_2(x,y)).
\end{equation}
Now let $y_1,y_2,x_1,x_2\in X$ with $\dd_2(x_1,y_1)>0$ and $\dd_2(x_2,y_2)>0$. Then, define
\[
\alpha_1=\frac{\dd_2(x_2,y_2)}{2(\dd_2(x_1,y_1)+\dd_2(x_2,y_2))}\ \textnormal{and}\ \alpha_2=\frac{\dd_2(x_1,y_1)}{2(\dd_2(x_1,y_1)+\dd_2(x_2,y_2))}
\]
we have that $\alpha_1\dd_2(x_1,y_1)=\alpha_2\dd_2(x_2,y_2)$ and hence, by \eqref{eq:proportionality}, it follows that
\[
\frac{\dd_1(x_1,y_1)}{\dd_2(x_1,y_1)}=\frac{\phi(\alpha_1 \dd_2(x_1,y_1))}{\alpha_1\dd_2(x_1,y_1)}=\frac{\phi(\alpha_2 \dd_2(x_2,y_2))}{\alpha_2\dd_2(x_2,y_2)}=\frac{\dd_1(x_2,y_2)}{\dd_2(x_2,y_2)}.
\]
Therefore, set $\kappa=\dd_1(x_2,y_2)/\dd_2(x_2,y_2)$, by the arbitrariness of $x_1,y_1\in X$, it follows that $\dd_1=\kappa \dd_2$.
\end{proof}

\begin{proof}[Proof of Proposition \ref{prop:medianprescontr}]
Let $p\in \DX$, $\dd$ a continuous metric on $X$, and let $T$ be a $\dd$-median preserving contraction for $p$. Let $x_p$ be in $\arg\min_{x\in X} \int \dd(x,y)\mathrm{d}p(y)$ be a fixed point for $T$. Then,
\begin{align*}
C_\dd(p)&=\int \dd(x_p,y)\mathrm{d}p(y)\\
&\geq \int \dd(T(x_p),T(y))\mathrm{d}p(y)\\
&=\int \dd(x_p,T(y))\mathrm{d}p(y)\\
&=\int \dd(x_p,y)\mathrm{d}T_{\#}p(y)\geq C_\dd(T_{\#}p)
\end{align*}
proving the claim.
\end{proof}

\subsubsection*{Proofs of Section \ref{sect:maximal}}

The following lemma is \textit{folklore}, but we report it here with the proof for completeness.
\begin{lemma}\label{lemma_median}
Suppose that $X=[a,b]$ and $\dd$ is induced by the absolute value metric. For all $p\in \DX$, there exists $\textnormal{med}(p)$ such that $p([a,\textnormal{med}(p)])\geq 1/2$ and $p([\textnormal{med}(p),b])\geq 1/2$ and
\[
C_\dd(p)=\int |\textnormal{med}(p)-y|\mathrm{d}p(y).
\]
\end{lemma}
\begin{proof}
Let $p\in \DX$ and denote by $P$ its cumulative distribution function, i.e., $P:x\mapsto p([a,x])$,
\[
\textnormal{med}(p)=\inf\left\lbrace x\in [a,b]:P(x)\geq \frac{1}{2} \right\rbrace.
\]
Since $P(b)=1$, $\textnormal{med}(p)$ exists and $p([a,\textnormal{med}(p)])=P(\textnormal{med}(p))\geq 1/2$. Moreover, 
\[
p([a,\textnormal{med}(p)))=\sup\limits_{x<\textnormal{med}(p)}P(x)\leq \frac{1}{2}
\]
and hence $p([\textnormal{med}(p),b])=1-p([a,\textnormal{med}(p)))\geq 1/2$. To conclude, notice that $F_p:x\mapsto \int |x-y|\mathrm{d}p(y)$ is convex and
\[
F'_{p,+}(x)=2p([a,x])-1\ \textnormal{and}\ F'_{p,-}(x)=2p([a,x))-1.
\]
Thus, $m\in X$ minimizes $F_p$ if and only if
\[
2p([a,m])-1=F'_{p,+}(m)\geq 0\geq F'_{p,-}(m)=2p([a,m))-1
\]
and hence, if and only if $p([a,m])\geq 1/2$ and $p([m,b])\geq 1/2$. By substituting $m$ with $\textnormal{med}(p)$ these inequalities hold and hence we have that $\textnormal{med}(p)$ minimizes $F_p$.
\end{proof}

\begin{proof}[Proof of Proposition \ref{lem:maxcompllott}]
For all $p\in \bigtriangleup([a,b])$, denote by $\mathrm{med}(p)$ the median outcome of $p$. By Lemma \ref{lemma_median}, we have that
\begin{align*}
C_\dd(p)&=\int \lvert y-\mathrm{med}(p)\rvert \mathrm{d}p(y)\\
&=\int_{[a,\mathrm{med}(p)]} \lvert y-\mathrm{med}(p)\rvert \mathrm{d}p(y)+\int_{[\mathrm{med}(p),b]} \lvert y-\mathrm{med}(p)\rvert \mathrm{d}p(y)\\
&=\int_{[a,\mathrm{med}(p))} \lvert y-\mathrm{med}(p)\rvert \mathrm{d}p(y)+\int_{\{\mathrm{med}(p)\}}\lvert y-\mathrm{med}(p)\rvert \mathrm{d}p(y)\\
&+\int_{(\mathrm{med}(p),b]} \lvert y-\mathrm{med}(p)\rvert \mathrm{d}p(y)\\
&\leq \int_{[a,\mathrm{med}(p))} \lvert a-\mathrm{med}(p)\rvert \mathrm{d}p(y)+\int_{(\mathrm{med}(p),b]} \lvert b-\mathrm{med}(p)\rvert \mathrm{d}p(y)\\
&\leq \lvert a-\mathrm{med}(p)\rvert p([a,\mathrm{med}(p)))+\lvert b-\mathrm{med}(p)\rvert p((\mathrm{med}(p),b]) \\
&\leq \frac{\lvert a-\mathrm{med}(p)\rvert}{2}+\frac{\lvert b-\mathrm{med}(p)\rvert}{2}\\
&=\frac{b-a}{2}=C_\dd\left(\frac12 \delta_a+\frac12 \delta_b\right).
\end{align*}
Thus, $C_\dd(1/2\delta_a+1/2\delta_b)\geq C_\dd(p)$ for all $p\in \DX$. To conclude we also show that $1/2\delta_a+1/2\delta_b$ is the unique maximizer. To this end let $p\in \DX$ and suppose that
\[
C_\dd(p)=\frac{b-a}{2}.
\]
If $\textnormal{med}(p)=a$, then 
\begin{equation}\label{suicidioimminente}
\frac{b-a}{2}=C_\dd(p)=\int_{(a,b]} (y-a) \mathrm{d}p(y)\leq (b-a)p((a,b])\leq \frac{b-a}{2}
\end{equation}
and hence $p((a,b])=1/2$ that yields $p(\left\lbrace a \right\rbrace)=1/2$. Moreover, notice that \eqref{suicidioimminente} also implies
\[
\int_{(a,b]} (b-y)\mathrm{d}p(y)=\int_{(a,b]} (b-a)-(y-a)\mathrm{d}p(y)=(b-a)p((a,b])-\int_{(a,b]} (y-a) \mathrm{d}p(y)=0
\]
and hence, by the strict positivity of $y\mapsto b-y$ on $(a,b)$ we must have $p((a,b))=0$. Therefore, $p(\{b\})=1/2$. The proof for the case $\textnormal{med}(p)=b$ is analogous. Therefore, let assume that $\textnormal{med}(p)\in (a,b)$. First, notice that since $p$ is a maximum, the inequalities we used before yield
\begin{align}\label{align}
\frac{b-a}{2}&=C_\dd(p)\leq (\textnormal{med}(p)-a)p([a,\textnormal{med}(p)))+(b-\textnormal{med}(p))p((\textnormal{med}(p),b])\nonumber \\
&\leq \frac{(\textnormal{med}(p)-a)}{2}+\frac{(b-\textnormal{med}(p))}{2}=\frac{b-a}{2}
\end{align}
and hence, given that $\textnormal{med}(p)-a>0$ and $b-\textnormal{med}(p)>0$, we have
\[
p([a,\textnormal{med}(p)))=\frac{1}{2}=p((\textnormal{med}(p),b])\ \textnormal{and}\ p(\{\textnormal{med}(p)\})=0.
\]
Since we have that
\[
\int_{[a,\textnormal{med}(p))}(\textnormal{med}(p)-y)\mathrm{d}p(y)\leq (\textnormal{med}(p)-a)p([a,\textnormal{med}(p)))
\]
and
\[
\int_{(\textnormal{med}(p),b]}(y-\textnormal{med}(p))\mathrm{d}p(y)\leq (b-\textnormal{med}(p))p((\textnormal{med}(p),b]),
\]
equality \eqref{align} yields $(\textnormal{med}(p)-y)=(\textnormal{med}(p)-a)$ and $(b-\textnormal{med}(p))=(y-\textnormal{med}(p))$ $p$-almost surely for all $y\in [a,\textnormal{med}(p))$ and $(\textnormal{med}(p),b]$, respectively. This holds if and only if $p(\{a\})=p(\{b\})=1/2$.\footnote{For the doubftul reader, notice that $(\textnormal{med}(p)-y)=(\textnormal{med}(p)-a)$ $p$-almost surely for all $y\in [a,\textnormal{med}(p))$ leads to:
\begin{align*}
1=p(y=a|[a,\textnormal{med}(p)))=\frac{p(\left\lbrace a\right\rbrace)}{p([a,\textnormal{med}(p)))}=2p(\left\lbrace a\right\rbrace).
\end{align*}
The same steps work for $b$.}
\end{proof}

\begin{proof}[Proof of Proposition \ref{lem:uniformdiscrete}]
Suppose, without loss of generality, that $X=[n]$. Let $(\alpha_i)_{i=1}^n\in \bigtriangleup^{n-1}$ and denote by $p^{\alpha}\in \DX$ the lottery $p^\alpha=\sum_{i\in [n]}\alpha_i\delta_i$. Assume that $C_\dd(p^{\alpha})\geq C_\dd\left(\sum_{x\in X}\frac{1}{|X|}\delta_x\right)$. Then, for all $i\in [n]$
\[
\sum_{j\neq i}\alpha_j\geq C_\dd(p^\alpha)\geq C_\dd\left(\sum_{x\in X}\frac{1}{|X|}\delta_x\right)= \frac{n-1}{n}.
\]
Since
\[
1-\alpha_i=\sum_{j\neq i}\alpha_j\geq\frac{n-1}{n},
\]
we have $\alpha_i\leq 1/n$ for every $i\in[n]$. Given that
$\sum_{i=1}^n\alpha_i=1$, it follows that
$\alpha_i=1/n$ for every $i\in[n]$ and hence
\[
p^\alpha=\sum_{x\in X}\frac{1}{|X|}\delta_x.
\]
Thus, the claim follows.
\end{proof}

\begin{proof}[Proof of Proposition \ref{prop:extremeasure}]
It follows from Proposition \ref{prop:extremepoints} proved below.
\end{proof}

\begin{proof}[Proof of Proposition \ref{prop:maximalelementsnormedspaces}]
Since $X$ admits a finite radius $r^*(X)$, it is sufficient to show that $C_\dd$ attains it for some $p\in \DX$. To this end notice that $X,\DX$ are both convex and compact. Moreover, the map $F:X\times \DX\to [0,\infty)$ defined as
\[
F:(x,p)\mapsto \int \lVert x-y\rVert\mathrm{d}p(y)
\]
is such that:
\begin{itemize}
\item For all $p\in \DX$, $x\mapsto F(x,p)$ is convex and $1$-Lipschitz.
\item For all $x\in X$, $p\mapsto F(x,p)$ is affine and $1$-Lipschitz.
\end{itemize}
As a consequence by Sion minimax Theorem (\cite{SionMinimax}), we have that:
\[
\max\limits_{p\in \DX}C_\dd(p)=\max\limits_{p\in \DX}\min\limits_{x\in X}F(x,p)=\min\limits_{x\in X}\max\limits_{p\in \DX}F(x,p)
\]
since $F(x,\cdot)$ is affine and $\textnormal{ext}(\DX)=X$,\footnote{We recall that we identify $X$ a subset of $\DX$, under the usual identification $x\leftrightarrow \delta_x$.} it follows that:
\[
\max\limits_{p\in \DX}C_\dd(p)=\min\limits_{x\in X}\max\limits_{p\in \DX}F(x,p)=\min\limits_{x\in X}\max\limits_{y\in X}\lVert x-y\rVert=r^*(X).
\]
Thus, there exists $p^*\in \DX$ such that $C_\dd(p^*)=r^*(X)$.
\end{proof}

\begin{lemma}\label{lem:uniquechebyshevcenter}
Suppose $X$ is a compact, convex, nonempty subset of $\R^n$ endowed with a strictly convex norm.\footnote{We recall that a norm $\lVert \cdot \rVert$ is said to be \textit{strictly convex} if whenever $x\neq y$ and $\lVert x\rVert=\lVert y\rVert=1$, we have $\lVert x+y\rVert<2$. In $\R^n$, the Euclidean norm is an example of a strictly convex norm.} Then, $X$ has a unique Chebyshev center.
\end{lemma}
\begin{proof}
Define $R:X\to \bar{\R}$ as $R(x)=\sup_{y\in X}\lVert x-y\rVert$ for all $x\in X$. Since $R$ is Lipschitz and $X$ is compact, there exists $x^*\in X$ the minimizes $R$. Suppose there are two distinct minimizers $x^*$ and $z^*$. Then, we have that
\[
R(x^*)=R(z^*)=r^*(X).
\]
This implies that, for all $y\in X$, $\lVert x^*-y\rVert\leq r^*(X)$ and  $\lVert z^*-y\rVert\leq r^*(X)$,
\[
\left\lVert \frac{x^*+z^*}{2}-y \right\rVert\leq \frac{1}{2}\left\lVert (x^*-y)+(z^*-y)\right\rVert\leq \frac{\lVert x^*-y\rVert+\lVert z^*-y\rVert}{2}\leq r^*(X).
\]
Then, we have two cases. If $\lVert x^*-y\rVert<r^*(X)$ or $\lVert z^*-y\rVert<r^*(X)$, then
\[
\left\lVert \frac{x^*+z^*}{2}-y \right\rVert<r^*(X).
\]
If instead, $\lVert x^*-y\rVert=\lVert z^*-y\rVert=r^*(X)$, then also, by strict convexity of the norm,
\[
\left\lVert x^*-y+z^*-y \right\rVert<2r^*(X)
\]
and hence, again,
\[
\left\lVert \frac{x^*+z^*}{2}-y \right\rVert<r^*(X).
\]
By compactness and continuity it follows that $R\left((x^*+z^*)/2\right)\in \mathbb{R}$ and in particular $R\left((x^*+z^*)/2\right)<r^*(X)$. Thus, we have a contradiction of the minimality of $x^*$ and $z^*$. Therefore, $x^*=z^*$, proving the uniqueness of the Chebyshev center.
\end{proof}

\begin{lemma}\label{lem:differentiability-contact-median}
Let $X\subseteq\mathbb{R}^n$ be nonempty, compact, and convex, endowed with the Euclidean norm. Let $p\in\Delta(X)$, $x^*\in X$, and $r>0$. Define
\[
K(x^*,r)
=
\left\{
y\in X:\|x^*-y\|=r
\right\},
\]
and suppose that $
p(K(x^*,r))=1.$
Let $F_p:\mathbb{R}^n\to\mathbb{R}$ be defined by
\[
F_p:x\mapsto \int_X\|x-y\|\,\mathrm{d}p(y).
\]
Then:
\begin{enumerate}
\item $F_p$ is continuously differentiable on the open ball $B(x^*,r)$
\item For every $x\in B(x^*,r)$,
\[
\nabla F_p(x)
=
\int_X\frac{x-y}{\|x-y\|}\,\mathrm{d}p(y)\ \textnormal{and}\ \nabla F_p(x^*)
=
\frac{x^*-\int_X y\,\mathrm{d}p(y)}{r}.
\]
\item If $x^*\in\arg\min_{x\in X}F_p(x)$,
then
\[
\int_X y\,\mathrm{d}p(y)=x^*
\qquad\text{and}\qquad
\nabla F_p(x^*)=\mathbf{0}.
\]
\end{enumerate}
\end{lemma}

\begin{proof}
Fix $x\in B(x^*,r)$. Since $p(K(x^*,r))=1$, we have
\[
\|x^*-y\|=r
\]
for $p$-almost every $y\in X$. Hence, by the reverse triangle inequality,
\[
\|x-y\|
\geq
\|x^*-y\|-\|x-x^*\|
=
r-\|x-x^*\|
>
0
\]
for $p$-almost every $y\in X$. In particular, $x\neq y$ for $p$-almost every $y$. Now fix $y\neq x$ and define the function $h_y:z\mapsto \lVert z-y\rVert$. Clearly, $h_y$ is differentiable at $x$, with
\[
\nabla\|x-y\|=\frac{x-y}{\|x-y\|}.
\]
Let $e_i$ denote the $i$-th vector of the canonical basis of $\mathbb{R}^n$. For every $t\neq0$, the reverse triangle inequality gives
\[
\left|
\frac{\|x+te_i-y\|-\|x-y\|}{t}
\right|
\leq 1.
\]
Moreover, for $p$-almost every $y$,
\[
\lim_{t\to0}
\frac{\|x+te_i-y\|-\|x-y\|}{t}
=
\frac{x_i-y_i}{\|x-y\|}.
\]
The dominated convergence theorem therefore implies
\[
\frac{\partial F_p}{\partial x_i}(x)
=
\int_X
\frac{x_i-y_i}{\|x-y\|}\mathrm{d}p(y).
\]
Since this holds for every coordinate $i\in [n]$, it follows that
\[
\nabla F_p(x)
=
\int_X
\frac{x-y}{\|x-y\|}\mathrm{d}p(y).
\]

To establish continuity of the gradient, let $(x_n)$ be a sequence in $B(x^*,r)$ converging to some $x\in B(x^*,r)$. For $p$-almost every $y\in X$,
\[
\frac{x_n-y}{\|x_n-y\|}
\longrightarrow
\frac{x-y}{\|x-y\|}.
\]
Furthermore,
\[
\left\|
\frac{x_n-y}{\|x_n-y\|}
\right\|=1.
\]
Another application of the dominated convergence theorem gives
\[
\nabla F_p(x_n)\longrightarrow\nabla F_p(x).
\]
Thus, $F_p$ is continuously differentiable on $B(x^*,r)$. Evaluating the gradient at $x^*$ and using $\|x^*-y\|=r$
$p$-almost surely we obtain
\begin{align*}
\nabla F_p(x^*)
&=
\int_X
\frac{x^*-y}{\|x^*-y\|}\mathrm{d}p(y)=
\frac{1}{r}
\int_X(x^*-y)\mathrm{d}p(y)=
\frac{x^*-\int_Xy\mathrm{d}p(y)}{r}.
\end{align*}
Now suppose that $x^*$ minimizes $F_p$ over $X$. Since $X$ is convex and $F_p$ is differentiable at $x^*$, the first-order optimality condition gives
\[
\left\langle
\nabla F_p(x^*),x-x^*
\right\rangle
\geq0
\qquad
\text{for every }x\in X.
\]
Because $X$ is compact and convex, $\int_Xy\mathrm{d}p(y)\in X$. Thus,
\begin{align*}
0
&\leq
\left\langle
\nabla F_p(x^*),\int_Xy\mathrm{d}p(y)-x^*
\right\rangle\\
&=
\left\langle
\frac{x^*-\int_Xy\mathrm{d}p(y)}{r},
\int_Xy\mathrm{d}p(y)-x^*
\right\rangle\\
&=
-\frac{\|x^*-\int_Xy\mathrm{d}p(y)\|^2}{r}.
\end{align*}
Since $r>0$, this is possible only if $\left\lVert x^*-\int_Xy\mathrm{d}p(y)\right\rVert=0.$
Therefore,
\[
\int_Xy\mathrm{d}p(y)=x^*.
\]
Substituting this equality into the gradient formula gives $\nabla F_p(x^*)=\mathbf{0}.$
\end{proof}

\begin{proof}[Proof of Corollary \ref{coro:maxRnmeasure}]
By Proposition \ref{prop:maximalelementsnormedspaces}, we have that $C_\dd(p)=r^*(X)$ for all $p\in \arg\max_{q\in \DX}C_\dd(q)$. Then, we have that
\[
r^*(X)=C_\dd(p)\leq \int \lVert x^*-y \rVert \mathrm{d}p(y)\leq r^*(X).
\]
This implies that $p(C(x^*))=1$. Indeed, $y\mapsto r^*(X)-\lVert x^*-y \rVert$ is nonnegative, and hence, $r^*(X)-\int \lVert x^*-y \rVert \mathrm{d}p(y)=0$ implies that $ \lVert x^*-y \rVert=r^*(X)$ $p$-almost surely. Moreover, $x^*$ minimizes $F(\cdot,p):x\mapsto \int \lVert x-y \rVert \mathrm{d}p(y)$. By Lemma \ref{lem:differentiability-contact-median}, we have that
\[
0=\nabla F(x^*,p)=\int \frac{x^*-y}{\lVert x^*-y\rVert}\mathrm{d}
p(y)=\frac{x^*-\int y \mathrm{d}p(y)}{r^*(X)},
\]
then $\int y\mathrm{d}p(y)=x^*$.

For the converse $p(C(x^*))=1$ and $\int y \mathrm{d}p(y)=x^*$. Then,
\[
\nabla F(x^*,p)=\frac{x^*-\int y \mathrm{d}p(y)}{r^*(X)}=0.
\]
Since $F(\cdot,p)$ is convex, we have that $x^*$ is a (global) minimizer of $F(\cdot,p)$. Then, $C_\dd(p)=F(x^*,p)=r^*(X)$. To conclude we show that
\begin{equation}\label{eq:extremepointsRn}
C(x^*)\subseteq \textnormal{ext}(X).
\end{equation}
To prove this, suppose that $z\in C(x^*)$ and that there exist distinct $x,y\in X$ and $\alpha\in (0,1)$ such that $z=\alpha x+(1-\alpha)y$. By the convexity of the norm we have
\begin{align*}
r^*(X)&=\lVert z-x^*\rVert\\
&=\lVert \alpha (x-x^*)+(1-\alpha)(y-x^*)\rVert\\
&\leq \alpha \lVert x-x^*\rVert+(1-\alpha)\lVert y-x^*\rVert\leq r^*(X)
\end{align*}
and, since $\alpha\in (0,1)$, we have $\lVert x-x^*\rVert=\lVert y-x^*\rVert=r^*(X)$. Define
\[
u=\frac{x-x^*}{r^*(X)}\ \textnormal{and}\ v=\frac{y-x^*}{r^*(X)}.
\]
Then, $\lVert u\rVert=\lVert v\rVert=1$ and $u\neq v$. Thus, strict convexity of the Euclidean norm yields
\[
\lVert \alpha u+(1-\alpha)v \rVert<1
\]
and hence
\[
r^*(X)=\lVert z-x^*\rVert=r^*(X)\lVert \alpha u+(1-\alpha)v \rVert<r^*(X)
\]
which is impossible. Therefore, $z\in \textnormal{ext}(X)$, and hence \eqref{eq:extremepointsRn} holds.
\end{proof}

\begin{proof}[Proof of Corollary \ref{coro:maxsimplices}]
Denote by $x^*$ the Chebyshev center of $X$. Since the barycentric coordinates in a simplex are unique, by Corollary \ref{coro:maxRnmeasure}. it follows that 
\[
\arg\max_{p\in \DX}C_\dd(p)=\left\lbrace \sum_{v\in \textnormal{ext}(X)}\alpha_v^*\delta_{v}\right\rbrace
\]
where $(\alpha^*_v)_{v\in \textnormal{ext}(X)}$ are the barycentric coordinates of $x^*$
\end{proof}

\begin{lemma}\label{lem:stupidregularsimplices}
Suppose $X$ is a regular simplex in $\R^n$ endowed with the Euclidean norm. Then,
\[
\forall v\in \textnormal{ext}(X),\ \left\lVert v-\frac{\sum_{v\in \textnormal{ext}(X)}v}{\lvert \textnormal{ext}(X)\rvert} \right\rVert=m
\]
for some $m\geq 0$.
\end{lemma}
\begin{proof}
Suppose $\textnormal{ext}(X)=\left\lbrace v_i:i\in [k]\right\rbrace$ and let $c=\sum_{v\in \textnormal{ext}(X)}v/\lvert \textnormal{ext}(X)\rvert$. Then, define $h_i=v_i-c$ for all $i\in [k]$. Then, $\sum_{i\in [k]}h_i=\mathbf{0}$ and, for all $i\in [k]$,
\begin{align*}
\sum_{j\neq i}\lVert h_i-h_j\rVert^2=\sum_{j\neq i}\lVert v_i-v_j\rVert^2=(k-1)\ell^2
\end{align*}
where $\ell=\lVert v_i-v_j\rVert$ for all $i\neq j$. Then, we have that
\begin{align*}
(k-1)\ell^2&=\sum_{j\neq i}\lVert h_i-h_j\rVert^2\\
&=\sum_{j\neq i}\left[\lVert h_i\rVert^2+\lVert h_j\rVert^2-2\langle h_i,h_j\rangle \right]\\
&=(k-1)\lVert h_i\rVert^2+\sum_{j\neq i}\lVert h_j\rVert^2-2\left\langle h_i,\sum_{j\neq i}h_j \right\rangle\\
&=(k-1)\lVert h_i\rVert^2+\sum_{j\neq i}\lVert h_j\rVert^2+2\left\langle h_i,h_i \right\rangle\\
&=k\lVert h_i\rVert^2+\sum_{j\in [k]}\lVert h_j\rVert^2.
\end{align*}
Then, 
\[
\lVert h_i\rVert^2=\frac{(k-1)\ell^2-\sum_{j\in [k]}\lVert h_j\rVert^2}{k}
\]
and hence, each $\lVert h_i\rVert$ is independent of $i\in [k]$, proving the claim.
\end{proof}

\begin{proof}[Proof of Corollary \ref{coro:maxregularsimplices}]
By Corollary \ref{coro:maxsimplices}, it is sufficient to show that the Chebyshev center $x^*$ of $X$ is the vector with barycentric coordinates $\left(1/\lvert\textnormal{ext}(X)\rvert\right)_{v\in \textnormal{ext}(X)}$. Let $\textnormal{ext}(X)=\left\lbrace v_i:i\in [k]\right\rbrace$ and denote by $c=1/k\sum_{i\in [k]}v_i$. First, notice that
\[
r^*(X)=\inf\limits_{y\in X}\sup\limits_{x\in X}\lVert x-y \rVert=\inf\limits_{y\in X}\max\limits_{i\in [k]}\lVert v_i-y \rVert\leq \max\limits_{i\in [k]}\lVert v_i-c\rVert=r
\]
where $r$ is independent of $i$ because $X$ is a regular simplex (Lemma \ref{lem:stupidregularsimplices}).
Now let $(\alpha_i)\in \bigtriangleup^{k-1}$ and $y\in X$, then
\begin{align*}
\max\limits_{i\in [k]}\lVert v_i-y \rVert^2&\geq \sum_{i\in [k]}\alpha_i\lVert v_i-y \rVert^2\\
&=\sum_{i\in [k]}\alpha_i\left\lVert \left(v_i-\sum_{j\in [k]}\alpha_jv_j\right)-\left(y-\sum_{j\in [k]}\alpha_jv_j\right) \right\rVert^2\\
&=\sum_{i\in [k]}\alpha_i\left[\left\lVert v_i-\sum_{j\in [k]}\alpha_jv_j\right\rVert^2+\left\lVert y-\sum_{j\in [k]}\alpha_jv_j \right\rVert^2\right]\\
&-2\sum_{i\in [k]}\alpha_i\left\langle v_i-\sum_{j\in [k]}\alpha_jv_j,y-\sum_{j\in [k]}\alpha_jv_j\right\rangle\\
&=\sum_{i\in [k]}\alpha_i\left[\left\lVert v_i-\sum_{j\in [k]}\alpha_jv_j\right\rVert^2+\left\lVert y-\sum_{j\in [k]}\alpha_jv_j \right\rVert^2\right].
\end{align*}
Taking the infimum on both sides with respect to $y$, it follows that
\[
(r^*(X))^2=\inf\limits_{y\in X}\max\limits_{i\in [k]}\lVert v_i-y \rVert^2\geq \sum_{i\in [k]}\alpha_i\left\lVert v_i-\sum_{j\in [k]}\alpha_jv_j\right\rVert^2.
\]
Choosing $\alpha_i=1/k$ for all $i\in [k]$, it follows that $r^*(X)\geq r$. Therefore, $r^*(X)=r$, and hence, $c$ is the Chebyshev center, concluding the proof.
\end{proof}

\subsubsection*{Proofs of Section \ref{sect:affinecore}}

\begin{proof}[Proof of Proposition \ref{prop:affinecoremeasure}]
For notational simplicity we define $\succcurlyeq_\dd$ as:
\[
p\succcurlyeq_\dd q\Longleftrightarrow \forall x\in X, \int \dd(x,y)\mathrm{d}p(y)\geq \int \dd(x,y)\mathrm{d}q(y).
\]
Clearly, $C_\dd$ is monotone with respect to $\succcurlyeq_\dd$. Suppose $p\succcurlyeq_\dd q$. Then, by affinity of $\succcurlyeq_\dd$ and monotonicity of $C_\dd$ in $\succcurlyeq_\dd$, for all $\alpha\in (0,1)$ and $\ell\in \DX$,
\[
C_\dd(\alpha p+(1-\alpha)\ell)\geq C_\dd(\alpha q+(1-\alpha)\ell).
\]
By the arbitrariness of $\alpha$ and $\ell$, it follows that $p \geq_\dd q$.
\par\medskip
Passing to the converse, suppose that $p\geq_\dd q$ and $x\in X$. Then, by Lemma \ref{basic properties}, we have that,
\[
\frac{1}{2}\int \dd(x,y)\mathrm{d}p(y)=C_\dd\left(\frac{1}{2}p+\frac{1}{2}\delta_x\right)\geq C_\dd\left(\frac{1}{2}q+\frac{1}{2}\delta_x\right)=\frac{1}{2}\int \dd(x,y)\mathrm{d}q(y)
\]
and hence, by the arbitrariness of $x\in X$,
\[
\forall x\in X,\ \int \dd(x,y)\mathrm{d}p(y)\geq \int \dd(x,y)\mathrm{d}q(y)
\]
that is $p\succcurlyeq_\dd q$. Therefore, $\geq_\dd=\succcurlyeq_\dd$.
\end{proof}

\begin{lemma}\label{lem:convolutionization}
Suppose $X=\prod_{i\in [k]}Z_i$ for some compact metric spaces $(Z_i,\dd_i)$ and let $\dd=\sum_{i\in [k]}\omega_i\dd_i$ for some $\omega_i>0$ and $k\geq 1$. Then, for all $p\in \DX$ with marginals on each $Z_i$ denote by $p^i$, we have
\[
\forall x\in X, \int_{X}\dd(x,y)\mathrm{d}p(y)=\sum_{i\in [k]}\omega_i\int_{Z_i} \dd_i(x_i,t)\mathrm{d}p^i(t)
\]
and also
\[
C_\dd(p)=\sum_{i\in [k]}\omega_i C_{\dd_i}(p^i).
\]
\end{lemma}
\begin{proof}
The proof follows from the linearity of the integral and the form of the metric $\dd$, for all $x\in X$,
\[
\int_{X}\dd(x,y)\mathrm{d}p(y)=\int_X\sum_{i\in [k]}\omega_i\dd_i(x_i,y_i)\mathrm{d}p(y)=\sum_{i\in [k]}\omega_i\int_{Z_i} \dd_i(x_i,t)\mathrm{d}p^i(t)
\]
and
\begin{align*}
C_\dd(p)&=\min\limits_{(x_1,\ldots,x_k)\in X}\sum_{i\in [k]}\omega_i\int_{Z_i} \dd_i(x_i,t)\mathrm{d}p^i(t)=\min\limits_{x_1\in Z_1,\ldots,x_k\in Z_k}\sum_{i\in [k]}\omega_i\int_{Z_i} \dd_i(x_i,t)\mathrm{d}p^i(t)\\
&=\sum_{i\in [k]}\omega_i\min\limits_{x_i\in Z_i}\int \dd_i(x_i,t)\mathrm{d}p^i(t)=\sum_{i\in [k]}\omega_i C_{\dd_i}(p^i)
\end{align*}
for all $p\in \DX$.
\end{proof}

\begin{proposition}\label{prop:convexorder}
If $X=[a,b]$ and $\dd$ is induced by the absolute value metric, then
\[
p\geq_\dd q \Longleftrightarrow p\geq_{\textnormal{cvx}}q
\]
for all $p,q\in \DX$.
\end{proposition}
\begin{proof}
Since each $\dd(x,\cdot)$ is convex, by Lemma \ref{lem}, we have that $\geq_{\textnormal{cvx}}\subseteq \geq_\dd$. Passing to the only if, suppose that $p\geq_\dd q$.
Then, Proposition \ref{prop:affinecoremeasure} yields that for all $\ell\in [a,b]$,
\begin{equation}\label{diostronzo}
\int_{\supp(p)} \lvert x-\ell\rvert \mathrm{d}p(x) \geq \int_{\supp(q)} \lvert x-\ell\rvert \mathrm{d}q(x).
\end{equation}
Let $\ell \leq \inf(\supp(p)\cup\supp(q))$ and $\ell'\geq \sup(\supp(p)\cup\supp(q))$. Then \eqref{diostronzo} yields,
\[
-\ell+\int x\mathrm{d}p(x)\geq -\ell+\int x\mathrm{d}q(x)\ \textnormal{and}\ \ell'-\int x\mathrm{d}p(x)\geq \ell'-\int x\mathrm{d}q(x)
\]
Therefore,
\[
\int x\mathrm{d}p=\int x\mathrm{d}q
\]
and hence, by $p\geq_\dd q$, Theorem 3.A.2 in \cite{ShakedShantikumar} yields that $p\geq_{\textnormal{cvx}}q$. 
\end{proof}

Denote by $\DXint$ the set of integrable probability measures. That is,
\[
\DXint=\left\lbrace p\in \DX:\int \dd(x,y)\mathrm{d}(p\otimes p)(x,y)<\infty\right\rbrace.
\]
\begin{lemma}\label{lem:samemeangeneralR}
Suppose $p,q\in \DXint$ with $X=\R$ and $\dd$ is induced by the absolute value metric. If $p\geq_\dd q$, then $\mu(p)=\mu(q)$.
\end{lemma}
\begin{proof}
For all $t\in \R$ and $\ell\in \DXint$, we have that
\begin{align*}
\int\lvert x-t\rvert \mathrm{d}\ell(x)&=\int 2(x\vee t)-x-t\mathrm{d}\ell(x)\\
\int\lvert x-t\rvert \mathrm{d}\ell(x)&=\int x+t-2(x\wedge t) \mathrm{d}\ell(x)
\end{align*}
and hence, for all $t\in \R$
\begin{align*}
\int\lvert x-t\rvert \mathrm{d}(p-q)(x)&=\int 2(x\vee t)\mathrm{d}(p-q)(x)-(\mu(p)-\mu(q))\geq 0\\
\int\lvert x-t\rvert \mathrm{d}(p-q)(x)&=\mu(p)-\mu(q)-\int 2(x\wedge t) \mathrm{d}(p-q)(x)\geq 0.
\end{align*}
By dominated convergence theorem (applied separately to both integrals and later subtracted), we have that
\begin{align*}
0\leq \lim_{t\to -\infty}\int\lvert x-t\rvert \mathrm{d}(p-q)(x)&=\int \lim_{t\to -\infty}2(x\vee t)\mathrm{d}(p-q)(x)-(\mu(p)-\mu(q))\\
&=2(\mu(p)-\mu(q))-(\mu(p)-\mu(q))=\mu(p)-\mu(q)\\
0\leq \lim\limits_{t\to\infty}\int\lvert x-t\rvert \mathrm{d}(p-q)(x)&=\mu(p)-\mu(q)-\int \lim\limits_{t\to\infty}2(x\wedge t) \mathrm{d}(p-q)(x)\\
&=\mu(p)-\mu(q)-2(\mu(p)-\mu(q))=-(\mu(p)-\mu(q)).
\end{align*}
Thus, $\mu(p)=\mu(q)$.
\end{proof}

\begin{lemma}\label{lem:samemeanL1}
Suppose $X=\prod_{i\in [k]} I_i$ for some compact intervals $(I_i)_{i\in [k]}$ and $\dd$ is induced by the $L^1$-norm. If $p,q\in \DX$ and $p\geq_\dd q$, then $\mu(p)=\mu(q)$.
\end{lemma}
\begin{proof}
Suppose that $X=\prod_{i\in [k]}[a_i,b_i]$ for some $k\geq 1$ and $a_i<b_i$. Suppose that $p\geq_\dd q$ and let $G\subseteq [k]$. Then, define $x(G)\in X$ as
\[
x_i(G)=a_i \mathbf{1}_G(i)+b_i\mathbf{1}_{G^c}(i).
\]
Then, for all $\ell\in \DX$, by Lemma \ref{lem:convolutionization},
\[
\int \lVert x(G)-y \rVert \mathrm{d}\ell(y)=\sum_{i\in [k]}\int \lvert x_i(G)-t\rvert \mathrm{d}\ell_i(t)=\sum_{i\in G}[\mu_i(\ell)-a_i]+\sum_{i\notin G}[b_i-\mu_i(\ell)].
\]
Then, it follows that for all $G\subseteq [k]$,
\[
\int \lVert x(G)-y \rVert \mathrm{d}(p-q)(y)=\sum_{i\in G}[\mu_i(p)-\mu_i(q)]-\sum_{i\notin G}[\mu_i(p)-\mu_i(q)]\geq 0.
\]
If $G=\{i\}$, then
\[
\mu_i(p)-\mu_i(q)\geq \sum_{j\neq i}[\mu_j(p)-\mu_j(q)]
\]
whereas if $G=[k]\setminus \{i\}$, then 
\[
\mu_i(p)-\mu_i(q)\leq \sum_{j\neq i}[\mu_j(p)-\mu_j(q)].
\]
Thus, for all $i\in [k]$,
\[
\mu_i(p)-\mu_i(q)=\sum_{j\neq i}[\mu_j(p)-\mu_j(q)].
\]
Moreover, letting $G=[k]$ and $G=\emptyset$, we get
\begin{equation}\label{eq;stupidinacarinina}
\sum_{i\in[k]}\mu_i(p)=\sum_{i\in [k]}\mu_i(q).
\end{equation}
Then, it follows that
\[
\mu_i(p)-\mu_i(q)=\sum_{j\neq i}[\mu_j(p)-\mu_j(q)]\Longrightarrow 2\mu_i(p)-\mu_i(q)=\sum_{j\in [k]}\mu_j(p)-\sum_{j\neq i}\mu_j(q)
\]
and hence, by \eqref{eq;stupidinacarinina},
\[
2\mu_i(p)-2\mu_i(q)=0.
\]
By the arbitrariness of $i\in [k]$, it follows that $\mu(p)=\mu(q)$.
\end{proof}

\begin{proof}[Proof of Proposition \ref{prop:marginals_convexorder}] Suppose that $I_i=[a_i,b_1]$ for all $\in [k]$ and $a_i<b_i$. By Lemma \ref{lem:convolutionization} and Lemma \ref{prop:convexorder} it is clear that if, for all $i\in [k]$, $p^i\geq_{\textnormal{cvx}}q^i$, then $p\geq_\dd q$. For the converse, suppose that $p\geq_\dd q$g. Thanks to Lemma \ref{lem:convolutionization}, we have that for all $x=(x_1,\ldots,x_k)\in X$,
\[
\int_{X}\lVert x-y\rVert_1\mathrm{d}p(y)=\sum_{i\in [k]}\int_{[a_i,b_i]} \lvert x_i-t\rvert\mathrm{d}p^i(t)
\]
and the same formulation holds for $q$. By hypothesis, we have that
\[
\sum_{i\in [k]}\int_{[a_i,b_i]} \lvert x_i-t\rvert\mathrm{d}(p^i-q^i)(t)\geq 0
\]
for all $(x_1,\ldots,x_k)\in X$. Now by assumption and the fact that $\mu_j(p)=\mu_j(q)$ for all $j\in [k]$, taking as vector $x=(a_1,\ldots,a_{i-1},t,a_{i+1},\ldots,a_k)$ for some $t\in [a_i,b_i]$, we have that 
\[
\sum_{j\in [k]}\int_{[a_j,b_j]} \lvert x_j-s\rvert\mathrm{d}(p^j-q^j)(s)=\int_{[a_i,b_i]} \lvert t-s\rvert\mathrm{d}(p^i-q^i)(s)\geq 0
\]
and hence, by the arbitrariness of $t$, it follows that $p^i\geq_{|\cdot|}q^i$ and hence, by Proposition \ref{prop:convexorder}, it follows that $p^i\geq_{\textnormal{cvx}}q^i$. By the arbitrariness of $i\in [k]$, the claim follows.
\end{proof}

We now provide the proof of Proposition \ref{prop:extremepoints}. To this end we first recall the classic Choquet-Cartier dilation result.\footnote{See \cite{DavidsonKennedy2021}, \cite{CartierFellMeyer1964}, and \cite{PhelpsChoquet}.}

\begin{theorem}[Choquet-Cartier dilation theorem]\label{thm:Cartier}
Let $X$ be a nonempty, compact, and convex subset of a normed space. For all $p\in \DX$ there exists a Borel probability kernel $K:X\times \mathcal{B}(X)\to [0,1]$ such that:\footnote{We recall that $\mathcal{B}(X)$ denotes the Borel sigma algebra and $K:X\times \mathcal{B}(X)\to [0,1]$ is said to be a \textit{Borel probability kernel} if:
\begin{enumerate}
\item For all $x\in X$. $K(x)(\cdot)\in \DX$. 
\item For all $B\in \mathcal{B}(X)$, $K(\cdot)(B)$ is Borel measurable.
\end{enumerate}
We use the notation $K()()$ to strengthen the importance of the first and second entry of $K$.}
\begin{enumerate}
\item $\int_Xy \mathrm{d}K(x)(y)=x$ for $p$-almost every $x\in X$.
\item $K(x)(\textnormal{ext}(X))=1$ for $p$-almost every $x\in X$.
\item The probability measure $q:B\mapsto \int K(x)(B)\mathrm{d}p(x)$ is such that $q(\textnormal{ext}(X))=1$ and $q\geq_{\textnormal{cvx}}p$.
\end{enumerate}
\end{theorem}

\begin{lemma}\label{lem:forextremepoints}
Let X be a compact convex subset of a normed space and suppose $\dd$ is induced by a norm. For all $p\in \DX$ there exists a kernel $K:X\to \DX$ such that $\int_Xy \mathrm{d}K(x)(y)=x$ and $K(x)(\textnormal{ext}(X))=1$ for $p$-almost every $x\in X$, $q:B\mapsto \int K(x)(B)\mathrm{d}p(x)$ is such that $q(\textnormal{ext}(X))=1$, and $q\geq_{\dd}p$. Moreover, if $p\left(\left\lbrace x\in X:K(x)\neq \delta_x \right\rbrace\right)>0$, then $q>_\dd p$.
\end{lemma}
\begin{proof}
By Theorem \ref{thm:Cartier}, we have that there exists a kernel $K:X\to \DX$ such that  $q:B\mapsto \int K(x)(B)\mathrm{d}p(x)$ satisfies $q(\textnormal{ext}(X))=1$ and $q\geq_{\textnormal{cvx}}p$. Thus, since $\dd$ is induced by a norm, we have $q\geq_\dd p$. For all $x,z\in X$ define the function:
\[
g_z:x\mapsto \int_X\lVert z-y\rVert \mathrm{d}K(x)(y)-\lVert z-x \rVert.
\]
By Jensen's inequality, for $p$-almost all $x\in X$,
\[
\int_X\lVert z-y\rVert \mathrm{d}K(x)(y)\geq \left\lVert \int_X z-y \mathrm{d}K(x)(y) \right\rVert=\lVert z-x \rVert
\]
and hence, $g_z\geq 0$ $p$-almost everywhere. Moreover,
\begin{align*}
\lvert g_z(x)-g_w(x)\rvert &=\left\lvert \int_X\lVert z-y\rVert \mathrm{d}K(x)(y)-\lVert z-x \rVert - \int_X\lVert w-y\rVert \mathrm{d}K(x)(y)+\lVert w-x \rVert  \right\rvert\\
&=\left\lvert \int_X\lVert z-y\rVert - \lVert w-y\rVert \mathrm{d}K(x)(y)+\lVert w-x \rVert-\lVert z-x \rVert \right\rvert\\
&\leq 2 \lVert z-w \rVert.
\end{align*}
Thus, for all $x\in X$, $z\mapsto g_z(x)$ is 2-Lipschitz, and hence continuous. Now suppose that $M=\left\lbrace x\in X:K(x)\neq \delta_x \right\rbrace$ is not $p$-null, i.e., $p(M)>0$. Let $E\subseteq X$ be a Borel set with $p(E)=1$ and such that for all $x\in E$, $\int_Xy \mathrm{d}K(x)(y)=x$. Let $N=E\cap M$, clearly $p(N)>0$. Notice that for all $x\in N$,
\[
g_x(x)=\int_X \lVert x-y\rVert \mathrm{d}K(x)(y)>0
\]
Let $Z=\left\lbrace z_m:m\geq 1\right\rbrace$ be a dense subset of $X$. Since $z\mapsto g_z(x)$ is continuous, there must exist $z_m\in Z$ sufficiently close to $x$ such that $g_{z_m}(x)>0$. This implies that
\[
N\subseteq \bigcup_{m\geq 1}\left\lbrace x\in X:g_{z_m}(x)>0 \right\rbrace.
\]
Since $p(N)>0$ and $p$ is a probability measure, it follows that there must exist $m\geq 1$ such that
\[
p\left(\left\lbrace x\in X:g_{z_m}(x)>0 \right\rbrace\right)>0.
\]
Then, since $g_{z_m}$ is nonnegative, it follows that
\[
\int_X g_{z_m}(x)\mathrm{d}p(x)>0
\]
and hence,
\begin{align*}
\int_X &\lVert z_m-x \rVert\mathrm{d}q(x)-\int_X \lVert z_m-x \rVert\mathrm{d}p(x)\\
&=\int_X \int_X \lVert z_m-y \rVert\mathrm{d}K(x)(y)\mathrm{d}p(x)-\int_X \lVert z_m-x \rVert\mathrm{d}p(x)=\int_X g_{z_m}(x)\mathrm{d}p(x)>0.   
\end{align*}
This yields $q>_\dd p$.
\end{proof}

\begin{proof}[Proof of Proposition \ref{prop:extremepoints}]
It follows from Lemma \ref{lem:forextremepoints}.
\end{proof}

\begin{proof}[Proof of Proposition \ref{prop:extremepointsuniqueness}]
Suppose by contradiction that $p(\textnormal{ext}(X))<1$. Then, $p(X\setminus \textnormal{ext}(X))>0$. By Lemma \ref{lem:forextremepoints}, there exists a kernel $K$ such that $K(x)(\textnormal{ext}(X))=1$ for $p$-almost all $x\in X$. Thus, if $x\notin \textnormal{ext}(X)$, we have $\delta_x(\textnormal{ext}(X))\neq K(x)(\textnormal{ext}(X))$ and hence $K(x)\neq \delta_x$ for $p$-almost all $x\in X\setminus \textnormal{ext}(X)$. Thus, we have that 
\[
p\left(\left\lbrace x\in X:K(x)\neq \delta_x \right\rbrace\right)>0
\]
and hence, by Lemma \ref{lem:forextremepoints}, there exists $q\in \DX$ such that $q>_\dd p$, contradicting the maximality of $p$.
\end{proof}

\subsection{Proofs of Section \ref{Section:CAEU}}
\begin{proof}[Proof of Proposition \ref{prop:prefforconcent}]
Suppose that $\lambda\geq 0$. Let $x,y\in X$, $p,q\in \DX$, $\alpha\in \left[0,1\right]$, and assume that $\delta_x\sim p$ and $\delta_y\sim q$. Then,
\begin{align*}
V(\delta_{\alpha x+(1-\alpha)y})&=u(\alpha x+(1-\alpha)y)\\
&\geq \alpha u(x)+(1-\alpha)u(y)\\
&=\alpha V(\delta_x)+(1-\alpha)V(\delta_y)\\
&=\alpha \left[\EUp-\lambda C_\dd(p)\right]+(1-\alpha)\left[\EUq-\lambda C_\dd(q)\right]\\
&=\mathbb{E}_{\alpha p+(1-\alpha)q}[u]-\lambda\left[\alpha C_\dd(p)+(1-\alpha)C_\dd(q)\right]\\
&\geq \mathbb{E}_{\alpha p+(1-\alpha)q}[u]-\lambda C_\dd(\alpha p+(1-\alpha)q)=V(\alpha p+(1-\alpha)q).
\end{align*}
Thus, $V$ exhibits preference for concentration.
\end{proof}

\begin{proof}[Proof of Proposition \ref{prop:affinecoreCAEU}]
Suppose that $p\succsim^* q$ and let $x\in X$ and $\alpha\in (0,1/2]$. By Lemma \ref{basic properties} we have that
\begin{equation}\label{eq:alpha_12}
C_\dd\left(\alpha p+(1-\alpha)\delta_x\right)=\alpha\int \dd(x,y)\mathrm{d}p(y)
\end{equation}
and the same for $q$ in place of $p$. By $p\succsim^* q$,
\begin{align*}
\mathbb{E}_{\alpha p+(1-\alpha)\delta_x}\left[u\right]-\lambda C_\dd(\alpha p+(1-\alpha)\delta_x)\geq \mathbb{E}_{\alpha q+(1-\alpha)\delta_x}\left[u\right]-\lambda C_\dd(\alpha q+(1-\alpha)\delta_x)
\end{align*}
that, by \eqref{eq:alpha_12} and the fact that $\alpha>0$, holds if and only if
\[
\int u(y)-\lambda \dd(x,y)\mathrm{d}p(y)\geq  \int u(y)-\lambda \dd(x,y)\mathrm{d}q(y).
\]
By the arbitrariness of $x\in X$, we have that sufficiency holds.

Let us pass to the converse. Suppose that
\[
\forall x\in X, \int u(y)-\lambda \dd(x,y)\mathrm{d}p(y)\geq  \int u(y)-\lambda \dd(x,y)\mathrm{d}q(y)
\]
for some $p,q\in \DX$. Let $\alpha\in (0,1)$ and $\ell\in \DX$. Then, by affinity of the integral, for all $x\in X$,
\begin{equation}\label{eq:forallx}
\int u(y)-\lambda \dd(x,y)\mathrm{d}(\alpha p+(1-\alpha)\ell)(y)\geq  \int u(y)-\lambda \dd(x,y)\mathrm{d}(\alpha q+(1-\alpha)\ell)(y).
\end{equation}
 By \eqref{eq:forallx}, 
\begin{align*}
\sup\limits_{x\in X}\left\lbrace \int u(y)-\lambda \dd(x,y)\mathrm{d}(\alpha p+(1-\alpha)\ell)(y)\right\rbrace\geq \\
\quad \quad \quad \sup\limits_{x\in X}\left\lbrace \int u(y)-\lambda \dd(x,y)\mathrm{d}(\alpha q+(1-\alpha)\ell)(y)\right\rbrace.
\end{align*}
Suppose first that $\lambda \geq 0$. Then,
\[
\sup\limits_{x\in X}-\lambda \int \dd(x,y)\mathrm{d}(\alpha p+(1-\alpha)\ell(y)=-\lambda \min\limits_{x\in X}\int \dd(x,y)\mathrm{d}(\alpha p+(1-\alpha)\ell)(y)
\]
from which it follows that $\alpha p+(1-\alpha)\ell \succsim \alpha q+(1-\alpha)\ell$. If $\lambda < 0$, the same steps apply but taking the infimum instead of the supremum. Therefore, $p\succsim^* q$.
\end{proof}

\begin{proof}[Proof of Corollary \ref{prop:FSDCAEU}]
Suppose that $\geq_{\textnormal{FSD}}$ is a subrelation of $\succsim$. Since $\geq_{\textnormal{FSD}}$ is affine and $\succsim^*$ is the largest affine subrelation of $\succsim$, we have that $\geq_{\textnormal{FSD}}\subseteq \succsim^*$. If $x\succeq y$. Then, $\delta_{x}\geq_{\textnormal{FSD}}\delta_y$ and hence $\delta_x\succsim^* \delta_y$. Thus, by Proposition \ref{prop:affinecoreCAEU},
\[
\forall z\in X,\  u(x)-\lambda \dd(x,z)\geq u(y)-\lambda \dd(y,z).
\]
This implies that 
\[
u(x)\geq u(y)-\lambda \dd(x,y)\ \textnormal{and}\ u(x)-\lambda \dd(x,y)\geq u(y) 
\]
and hence, $u(x)-u(y)\geq \max \left\lbrace \lambda \dd(x,y),-\lambda \dd(x,y)\right\rbrace=|\lambda|\dd(x,y)\geq 0$.

For the converse, suppose that $x\succeq y$. Then, for all $z\in X$,
\[
u(x)-u(y)\geq |\lambda| \dd(x,y)\geq |\lambda| |\dd(x,z)-\dd(y,z)| \geq \lambda [\dd(x,z)-\dd(y,z)]
\]
and hence
\[
u(x)-\lambda \dd(x,z)\geq u(y)-\lambda \dd(y,z)
\]
thus, we have that each function $v_z:x\mapsto u(x)-\lambda \dd(x,z)$ is increasing, and hence $\succsim$ is consistent with first-order stochastic dominance. Indeed, if $p\geq_{\textnormal{FSD}}q$, then
\[
\forall z\in X,\ \int_X v_z(x)\mathrm{d}p(x)\geq \int_X v_z(x)\mathrm{d}q(x)
\]
and hence $p\succsim^* q$, that in turn yields $p\succsim q$.
\end{proof}
\begin{proof}[Proof of Corollary \ref{coro:CAEUstrongrisk}]
Suppose that $p\geq_{\textnormal{cve}}q$. Then, by affinity of the concave order,
\[
p\geq_{\textnormal{cve}}q \Longrightarrow \alpha p+(1-\alpha)\ell \geq_{\textnormal{cve}} \alpha q+(1-\alpha)\ell 
\]
for all $\alpha\in (0,1]$ and $\ell\in\DX$. Then, since $\succsim$ exhibits strong risk aversion, it follows that
\[
\alpha p+(1-\alpha)\ell \succsim \alpha q+(1-\alpha)\ell 
\]
for all $\alpha\in (0,1]$ and $\ell\in\DX$ and hence, $p\succsim^* q$. Since, $\delta_{\mu(p)}\geq_{\textnormal{cve}} p$, it follows that $\delta_{\mu(p)}\succsim^* p$ for all $p\in \DX$. By Proposition \ref{prop:affinecoreCAEU}, it follows that taking $p=\alpha \delta_x+(1-\alpha)\delta_y$ for $\alpha\in [0,1]$ and $x,y\in X$, we have that for all $z\in X$,
\[
u(\alpha x+(1-\alpha)y)-\lambda \dd(\alpha x+(1-\alpha)y,z)\geq \alpha [u(x)-\lambda \dd(x,z)]+(1-\alpha)[u(y)-\lambda \dd(y,z)].
\]
Thus, each $y\mapsto u(y)-\lambda \dd(y,z)$ is concave. 
\par\medskip
The converse follows from observing that if each $y\mapsto u(y)-\lambda \dd(y,z)$ is concave, then $\geq_{\textnormal{cve}}\subseteq \succsim^*\subseteq \succsim$, and hence $\succsim$ exhibits strong risk aversion.
\par\medskip
We assume now that $\dd$ is induced by a continuous norm. Suppose first that $\lambda \leq 0$. Then, by the equivalence we just proved, $\succsim$ exhibits strong risk aversion if and only if each $y\mapsto u(y)-\lambda \dd(y,z)$ is concave. Since $\lambda \leq 0$, this implies that taking $x\neq y$, $\alpha\in (0,1)$, and $z=\alpha x+(1-\alpha)y$,
\begin{align*}
u(z)-[\alpha u(x)+(1-\alpha)u(y)]\geq \lambda \left[-\alpha \dd(x,z)-(1-\alpha)\dd(y,z)\right]\geq 0
\end{align*}
where the last inequality follows from the fact that $\dd$ is induced by a norm, $z=\alpha x+(1-\alpha)y$, and $\lambda \leq 0$. Therefore, $u$ is concave.
\par\medskip
Suppose now that $u$ is convex. Choose $x\neq y$, $\alpha\in (0,1)$, $z=\alpha x+(1-\alpha)y$. Then, using the same argument we just employed, we get 
\begin{align*}
0\geq u(z)-[\alpha u(x)+(1-\alpha)u(y)]\geq \lambda \left[-\alpha \dd(x,z)-(1-\alpha)\dd(y,z)\right]
\end{align*}
and hence, since $\dd$ is induced by a norm, it follows that $\lambda \geq 0$.
\end{proof}

\begin{proof}[Proof of Corollary \ref{prop:SSDCAEU}]
Suppose first that every local utility is increasing and concave. If $p\geq_{\mathrm{SSD}}q$, then
\[
\int_X v_x(y)\mathrm{d}p(y)\geq\int_X v_x(y)\mathrm{d}q(y)
\]
for every local utility $v_x:y\mapsto u(y)-\lambda\dd(x,y).$
Proposition \ref{prop:affinecoreCAEU} therefore implies
$p\succsim^*q$, and hence $p\succsim q$.

Conversely, suppose that $\succsim$ is consistent with second-order
stochastic dominance. Since
\[
\geq_{\mathrm{FSD}} \cup \geq_{\mathrm{cve}}\subseteq\geq_{\mathrm{SSD}},
\]
the preference is consistent with both first-order stochastic dominance and the concave order. Corollaries \ref{prop:FSDCAEU} and \ref{coro:CAEUstrongrisk} then imply that every local utility is,
respectively, increasing and concave.

The additional conclusions also follow directly from the corresponding conclusions of Corollaries \ref{prop:FSDCAEU}
and \ref{coro:CAEUstrongrisk}.
\end{proof}


\begin{lemma}\label{lem:CAEU_maximization}
Suppose $\succsim$ are CAEU preferences represented by $V$ with utility $u$, $\lambda\in \R$, and a continuous metric $\dd$. For all subsets $A$ of $\DX$,
\begin{align*}
\lambda\geq 0\Longrightarrow \sup\limits_{p\in A}V(p)=\max\limits_{x\in X}\sup\limits_{p\in A}\left[\EUp-\lambda\int \dd(x,y)\mathrm{d}p(y)\right].
\end{align*}
If $X$ is a compact and convex subset of a normed space, $y\mapsto \dd(x,y)$ is convex for all $x\in X$, and $A$ is a convex subset of $\DX$, then
\[
\lambda<0\Longrightarrow \sup\limits_{p\in A}V(p)=\min\limits_{x\in X}\sup\limits_{p\in A}\left[\EUp-\lambda\int \dd(x,y)\mathrm{d}p(y)\right].
\]
\end{lemma}
\begin{proof}
Suppose that $\lambda\geq 0$. Then,
\begin{align*}
\sup\limits_{p\in A}V(p)&=\sup\limits_{p\in A}\left[\EUp-\lambda\min\limits_{x\in X}\int \dd(x,y)\mathrm{d}p(y)\right]\\
&=\sup\limits_{p\in A}\left[\EUp-\min\limits_{x\in X}\lambda\int \dd(x,y)\mathrm{d}p(y)\right]\\
&=\sup\limits_{p\in A}\left[\EUp+\max\limits_{x\in X}-\lambda\int \dd(x,y)\mathrm{d}p(y)\right]\\
&=\sup\limits_{p\in A}\max\limits_{x\in X}\left[\EUp-\lambda\int \dd(x,y)\mathrm{d}p(y)\right]\\
&=\max\limits_{x\in X}\sup\limits_{p\in A}\left[\EUp-\lambda\int \dd(x,y)\mathrm{d}p(y)\right].
\end{align*}
Now suppose that $\lambda<0$. Then, define the function
\[
L:(p,x)\mapsto \EUp-\lambda\int \dd(x,y)\mathrm{d}p(y).
\]
Clearly, for all $p\in \DX$ and $x\in X$, the functions $L(p,\cdot)$ and $L(\cdot,x)$ are continuous. Moreover, $L(p,\cdot)$ is convex, as $-\lambda>0$ and
\[
x\mapsto \int \dd(x,y) \mathrm{d}p(y)
\]
is convex for all $p\in \DX$. For all $x\in X$, $p\mapsto L(p,x)$ is affine. Then, by Sion's minimax theorem (\cite{SionMinimax}), it follows that
\[
\sup\limits_{p\in A}\min\limits_{x\in X}L(p,x)=\min\limits_{x\in X}\sup\limits_{p\in A}L(p,x)
\]
and hence the claim follows.
\end{proof}

\begin{proposition}\label{prop:certaintyalways}
Suppose $\succsim$ are CAEU preferences represented by $V$ and $\lambda\geq 0$ with utility $u$, $\lambda\in \R$, and a continuous metric $\dd$. If $A\subseteq \DX$ is compact and contains all degenerate lotteries, then
\[
\max\limits_{p\in A}V(p)= \max\limits_{x\in X}u(x).
\]
\end{proposition}
\begin{proof}
Let $p^*\in \arg\max_{p\in A}V(p)$. Then, by compactness of $X$, it follows that $\supp(p^*)$ is compact. Let $x^*\in \arg\max_{x\in \supp(p^*)}u(x)$. It follows that
\[
V(\delta_{x^*})=u(x^*)\geq \mathbb{E}_{p^*}[u]\geq \mathbb{E}_{p^*}[u]-\lambda \min_{x\in X}\int \dd(x,y)\mathrm{d}p^*(y)=V(p^*).
\]
Therefore, for all $p^*\in \arg\max_{p\in A}V(p)$, there exists $x^*$ such that $u(x^*)=V(\delta_{x^*})\geq V(p^*)$. Since $X\subseteq A$, the claim follows.
\end{proof}

\subsection{Extending $C_\dd$ to the noncompact case}
Suppose $(X,\dd)$ is a metric space. We recall that by $\DXint$ we denote the set of integrable probability measures,
\[
\DXint=\left\lbrace p\in \DX:\int \dd(x,y)\mathrm{d}(p\otimes p)(x,y)<\infty\right\rbrace.
\]
Define $F_\dd:X\times\DXint\to \R$ as
\[
F_\dd(x,p)=\int \dd(x,y)\mathrm{d}p(y)
\]
for all $(x,p)\in X\times\DXint$. 
\begin{lemma}
Suppose $(X,\dd)$ is a metric space. Then $F_\dd$ is nonnegative and
\begin{itemize}
\item $F_\dd(\cdot,p)$ is 1-Lipschitz for all $p\in \DXint$.
\item For all $p\in \DXint$ and $t\in \R$,
\[
\left[F_\dd(\cdot,p)\leq t\right]=\left\lbrace x\in X:F_\dd(x,p)\leq t \right\rbrace
\]
is closed and bounded.
\item Let $x_0\in X$. If $\dd(x_n,x_0)\to \infty$, then, for all $p\in \DXint$, $F_\dd(x_n,p)\to \infty$.
\end{itemize}
\end{lemma}
\begin{proof}
Nonnegativity is obvious. Let $x,y\in X$ and $p\in \DXint$. Then,
\[
\lvert F_\dd(x,p)-F_\dd(y,p)\rvert\leq \int |\dd(x,z)-\dd(y,z)|\mathrm{d}p(z)\leq \dd(x,y).
\]
Thus, $F_\dd(\cdot,p)$ is 1-Lipschitz. Since $F_\dd(\cdot,p)$ is continuous it follows that $\left[F_\dd(\cdot,p)\leq t\right]$ is closed. Now suppose that $x\in \left[F_\dd(\cdot,p)\leq t\right]$ and $y\in X$. Then,
 by triangle inequality,
 \begin{align*}
\dd(x,y)\leq \int\dd(x,z)\mathrm{d}p(z)+\int\dd(y,z)\mathrm{d}p(z)=F_\dd(x,p)+F_\dd(y,p)\leq t+F_\dd(y,p).
\end{align*}
Thus, by the arbitrariness of $x\in X$, we have $\left[F_\dd(\cdot,p)\leq t\right]\subseteq B(y,t+F_\dd(y,p)+\varepsilon)$, where $\varepsilon>0$. Thus, $\left[F_\dd(\cdot,p)\leq t\right]$ is bounded.
\par\medskip
Suppose that $\dd(x_n,x_0)\to \infty$. Then, for all $p\in \DXint$,
\[
F_\dd(x_n,p)=\int \dd(x_n,y)\mathrm{d}p(y)\geq \dd(x_n,x_0)-\int \dd(x_0,y)\mathrm{d}p(y)\to \infty.
\]
\end{proof}
\begin{lemma}
If for all $t\in \R$ and $p\in \DXint$, the set $\left[F_\dd(\cdot,p)\leq t\right]$ is compact, then $F_\dd(\cdot,p)$ admits a minimum.
\end{lemma}
\begin{proof}
Let $p\in \DXint$ and $x_0\in X$. Since $F_\dd(\cdot,p)$ is 1-Lipschitz, it admits a minimum $x_*$ in $\left[F_\dd(\cdot,p)\leq F_\dd(x_0,p)\right]$. Then, we have that
\[
F_\dd(x_*,p)\leq F_\dd(x,p)\leq F_\dd(x_0,p)\leq F_\dd(y,p)
\]
for all $x\in \left[F_\dd(\cdot,p)\leq F_\dd(x_0,p)\right]$ and all $y\notin \left[F_\dd(\cdot,p)\leq F_\dd(x_0,p)\right]$. Thus, $x_*$ is a minimum.
\end{proof}
\noindent This lemma implies that whenever $(X,\dd)$ is so that each $\left[F_\dd(\cdot,p)\leq t\right]$ is compact, then the complexity measure $C_\dd$ is well-defined. Therefore in many cases one does not need to require that $(X,\dd)$ is compact. An important class of examples is composed of metric spaces with the \textit{Heine-Borel property}, meaning that closed and bounded
subsets of $X$ are compact, e.g., $X$ closed subset of $\R^n$ endowed with any norm. In general, the Hopf-Rinow Theorem (Theorem 2.5.28 of \cite{GromovMisha})
shows that any complete, locally compact, length metric space has the Heine-Borel property. This includes complete and connected Riemannian manifolds equipped with their geodesic metric (Chapter 6 of \cite{LeeRiemannian}).\footnote{The interested reader can consult also \cite{AndreinoPatatinoBrillosino}.} 
\par\medskip
Apart from metric spaces of the Heine-Borel property one can employ the previous lemmas to discuss the existence of minimizers of each $F_\dd(\cdot,p)$ in normed spaces. 
\begin{lemma}
Suppose $X$ is a convex subset of a normed space $(V,\lVert\cdot\rVert)$. If $\dd$ is induced by $\lVert\cdot \rVert$, then, for all $p\in \DXint$, $F_\dd(\cdot,p)$ is weakly lower semicontinuous.
\end{lemma}
\begin{proof}
Define $\bar{F}_\dd:V\times \DXint\to \R$ as 
\[
\bar{F}_{\dd}:(v,p)\mapsto \int_X\|v-y\|\mathrm dp(y).
\]
By triangle inequality, for all $v\in V$ and $p\in \DXint$,
\[
\bar{F}_{\dd}(v,p) \leq \| v\| +\int_X \| y\|\mathrm{d}p(y)<\infty
\]
where the last inequality follows since $p\in \DXint$, so it has finite first moment. Let $p\in \DXint$. Then, $\bar{F}_\dd(\cdot,p)$ is an extension of $F_{\dd}(\cdot,p)$. Moreover, $\bar{F}_\dd(\cdot,p)$ is convex and norm-continuous. Hence its epigraph is convex and
norm-closed. By the Hahn--Banach separation theorem, every
norm-closed convex set is weakly closed. Therefore the epigraph of
$\bar{F}_{\dd}(\cdot,p)$ is weakly closed, so $\bar{F}_{\dd}(\cdot,p)$ is weakly
lower semicontinuous. Consequently, $F_{\dd}(\cdot,p)$ is also weakly lower semicontinuous.
\end{proof}
An immediate observation from this lemma is the following.
\begin{lemma}
If $X$ is a weakly compact convex subset of a normed space $(V,\lVert\cdot\rVert)$ and $\dd$ is induced by $\lVert\cdot \rVert$, then, for all $p\in \DXint$, $F_\dd(\cdot,p)$ admits a minimum in $X$.
\end{lemma}

\subsection{Remarks on support comparisons and entropy}\label{remarksuppentropy}

Here, we notice that, whenever $|X|\geq 3$, there is no metric $\dd$ on $X$ such that $C_\dd$ is ordinally equivalent to $p\mapsto |\supp(p)|$ or $p\mapsto H(p)$. 
\begin{remark}
Suppose $X$ is finite and $|X|\geq 3$ and it is endowed with a metric $\dd$. Denote by $D^*$ the $\dd$-diameter of $X$ and by $D_*=\min_{x\neq y}\dd(x,y)>0$. Fix $x\in X$. Then, let $\varepsilon\in (0,1)$ sufficiently small such that
\[
p_{\varepsilon}(x):=1-(|X|-1)\varepsilon \in (0,1).
\]
Let $p_\varepsilon(y):=\varepsilon$ for all $y\neq x$. Then, we have that
\[
\min\limits_{z\in X}\sum_{y\in X}p_\varepsilon(y)\dd(y,z)\leq (|X|-1)\varepsilon D^*.
\]
Now let $q=1/2\delta_t+1/2\delta_s$ for some $t\neq s$ in $X$. Then,
\[
\min\limits_{z\in X}\sum_{y\in X}q(y)\dd(y,z)=\frac{1}{2}\dd(t,s)\geq \frac{1}{2}D_*.
\]
Now let $\varepsilon$ be sufficiently small so that $(|X|-1)D^*\varepsilon<1/2D_*$. Thus, we have that:
\[
C_\dd(q)\geq \frac{1}{2}D_*>C_\dd(p_\varepsilon)
\]
even though $|\supp(p)|=|X|>2=|\supp(q)|$.
\end{remark}

\begin{remark}
Suppose $X$ is finite with $|X|\geq 3$ and $\dd$ is a metric on $X$. Let $x,y\in X$ be such that $\dd(x,y)$ is the $\dd$-diameter of $X$. Let $p=1/2\delta_x+1/2\delta_y$. By Lemma \ref{basic properties},
\[
C_\dd\left(p\right)=\frac{1}{2}\dd(x,y).
\]
Moreover, let $\lambda\in (0,1)$ and let $\ell$ be the uniform on $X\setminus \left\lbrace x\right\rbrace$. Let $q=\lambda \delta_x+(1-\lambda)\ell$. Then,
\[
C_\dd(q)\leq (1-\lambda)\sum_{z\neq x}\frac{\dd(x,z)}{|X|-1}\leq (1-\lambda)\dd(x,y).
\]
If $\lambda=2/3$, then
\[
C_\dd(q)\leq \frac{\dd(x,y)}{3}<\frac{\dd(x,y)}{2}=C_\dd(p)
\]
and 
\[
H(q)=-\lambda \log(\lambda)-(1-\lambda)\log(1-\lambda)+(1-\lambda) \log(|X|-1)>\log(2)=H(p)
\]
since $|X|\geq 3$.
\end{remark}

\end{document}